\documentclass[12pt]{article}
\usepackage{multicol}
\RequirePackage{amsthm,amsmath}
\RequirePackage[colorlinks,citecolor=blue,urlcolor=blue]{hyperref}

\usepackage{graphicx,psfrag,epsf,subcaption}
\usepackage{enumerate}
\usepackage{enumitem}
\usepackage{natbib}
\usepackage{url} %
\usepackage{amssymb}
\usepackage{amsthm}
\usepackage{bm}
\usepackage{url}
\usepackage{float}
\usepackage{siunitx}
\usepackage{bbm}
\usepackage{booktabs}
\usepackage{array}
\usepackage{hyperref}

\usepackage[toc,page]{appendix}
\usepackage[font=small,labelfont=bf]{caption}
\RequirePackage[normalem]{ulem}

\newcommand{\blind}{1}

\let\hat\widehat
\let\tilde\widetilde

\ifx\BlackBox\undefined
\newcommand{\BlackBox}{\rule{1.5ex}{1.5ex}}  %
\fi

\ifx\QED\undefined
\def\QED{~\rule[-1pt]{5pt}{5pt}\par\medskip}
\fi

\ifx\proof\undefined
\newenvironment{proof}{\par\noindent{\bf Proof\ }}{\hfill\BlackBox\\[2mm]}
\fi

\ifx\theorem\undefined
\newtheorem{theorem}{Theorem}
\fi
\ifx\example\undefined
\newtheorem{example}[theorem]{Example}
\fi
\ifx\property\undefined
\newtheorem{property}{Property}
\fi
\ifx\lemma\undefined
\newtheorem{lemma}[theorem]{Lemma}
\fi
\ifx\proposition\undefined
\newtheorem{proposition}[theorem]{Proposition}
\fi
\ifx\remark\undefined
\newtheorem{remark}[theorem]{Remark}
\fi
\ifx\corollary\undefined
\newtheorem{corollary}[theorem]{Corollary}
\fi
\ifx\definition\undefined
\newtheorem{definition}[theorem]{Definition}
\fi
\ifx\conjecture\undefined
\newtheorem{conjecture}[theorem]{Conjecture}
\fi
\ifx\fact\undefined
\newtheorem{fact}[theorem]{Fact}
\fi
\ifx\claim\undefined
\newtheorem{claim}[theorem]{Claim}
\fi
\ifx\assumption\undefined

\fi
\numberwithin{equation}{section}
\numberwithin{theorem}{section}
\usepackage{multirow}

\newcommand\floor[1]{\lfloor#1\rfloor}

\newcommand{\cE}{\mathcal{E}}

\newcommand{\E}{\mathbb{E}}

\newcommand{\btheta}[0]{\boldsymbol{\theta}}

\newcommand{\inner}[2]{\langle #1, #2\rangle}
\newcommand{\cA}{\mathcal{A}}
\newcommand{\cI}{\mathcal{I}}

\renewcommand{\Pr}{P}
\newcommand{\by}{\mathbf{y}}
\newcommand{\cX}{\mathcal{X}}
\newcommand{\cY}{\mathcal{Y}}
\newcommand{\bY}{\mathbf{Y}}
\newcommand{\cZ}{\mathcal{Z}}

\newcommand{\cH}{\mathcal{H}}

\newcommand{\hetafull}{\widehat{\eta}_{n,\lambda}}
\newcommand{\hetanull}{\widehat{\eta}^0_{n,\lambda}}

\newcommand{\bE}{\mathbb{E}}
\newcommand{\pd}{\Phi_{n,\lambda}}

\newcommand{\teta}{\widetilde{\eta}}
\newcommand{\ellnl}{\ell_{n,\lambda}}
\newcommand{\lnl}{\ell_{n,\lambda}}
\newcommand{\err}{\mbox{Err}}
\newcommand{\var}{\mbox{Var}}

\newcommand{\bbR}{\mathbb{R}}

\newcommand{\angular}[1]{\langle #1 \rangle}
\newcommand{\cHx}{\cH^{\angular{X}}}

\newcommand{\eb}{\mathbf{e}}

\newcommand{\bc}{\bm{c}}

\newcommand{\bt}{\bm{t}}

\newcommand{\EE}{\mathbb{E}}

\newcommand{\PP}{\mathbb{P}}
\newcommand{\QQ}{\mathbb{Q}}
\newcommand{\RR}{\mathbb{R}}

\newcommand{\bomega}{\bm{\omega}}

\newcommand{\argmin}{\mathop{\mathrm{argmin}}}
\newcommand{\argmax}{\mathop{\mathrm{argmax}}}

\newcommand{\bbN}{\mathbb{N}}

\newtheorem{assumption}{Assumption}

\date{}
\begin{document}

\def\spacingset#1{\renewcommand{\baselinestretch}%
{#1}\small\normalsize} \spacingset{1}

\if1\blind
{
\title{\bf Minimax Nonparametric Two-Sample  Test under Smoothing%
}

\author
{
    Xin Xing\thanks{ Department of Statistics, Virginia Tech},    $\quad$
    Zuofeng Shang\thanks{ Department of Mathematical Sciences, New Jersey Institute of Technology},   $\quad$
    Pang Du \thanks{ Department of Statistics, Virginia Tech}$\quad$
    Ping Ma\thanks{ Department of Statistics, University of Georgia},  $\quad$\\
    Wenxuan Zhong, \footnotemark[4] $\quad$
    Jun S. Liu \thanks{ Department of Statistics, Harvard University}
}
\maketitle
} \fi

\if0\blind
{
  \bigskip
  \bigskip
  \bigskip
  \begin{center}
    {\LARGE\bf Minimax Nonparametric Two-Sample  Test under Smoothing}
\end{center}
  \medskip
} \fi

\bigskip

\begin{abstract}
We consider the  problem of comparing probability densities between 
two groups. A new probabilistic tensor product smoothing spline framework is developed to model the joint density of two variables. Under such a framework, the probability density comparison is equivalent to testing the presence/absence of interactions. We propose a penalized likelihood ratio test for such interaction testing and show that the test statistic is asymptotically chi-square distributed under the null hypothesis. Furthermore, we derive a sharp minimax testing rate  based on the Bernstein width for  nonparametric two-sample tests and show that our proposed test statistics is minimax optimal.  In addition, a data-adaptive tuning criterion is developed to choose the penalty parameter. Simulations and real applications demonstrate that the proposed test outperforms the conventional approaches under various scenarios. 
\end{abstract}

\noindent%
\vfill

\newpage
\spacingset{1.5}
\section{Introduction}

A fundamental problem in statistics is to test whether the probability densities underlying two groups of observed data are the same, which is called the two-sample test. It plays an essential role in different scientific fields ranging from modern biological sciences to deep learning. For instance, in metagenomics studies,
comparing  densities of specific microbial species (or strains) 
from different treatment groups helps researchers gain insights on the disease and treatments \citep{turnbaugh2009core};  in genomics, identifying differentially expressed genes between two groups or conditions is fundamental to many downstream analyses \citep{tang2009mrna}; in machine learning,
the two-sample test is becoming an essential component in some deep learning algorithms \citep{li2017mmd}. 

In these modern applications, the underlying distributions usually demonstrate complex patterns, including multi-modality and long-tails.
Hence, it is often difficult to specify their distributional families. 
Classical normality-based tests such as the two-sample t-test \citep{anderson1958introduction} and the Shapiro-Wilk test \citep{shapiro1965analysis}
are generally inappropriate.
Nonparametric approaches are more appealing due to their distribution-free feature.
Examples include distance-based tests such as the Kolmogorov-Smirnov (KS) test \citep{darling1957kolmogorov}, the Anderson-Darling test \citep{scholz1987k},
and their variants; an alternative direction is using discretization (“slicing”) of continuous random variables. \citet{jiang2015nonparametric} proposed the dynamic slicing test (DSLICE), which penalizes the number of slices to regularize the test statistics.  Recently,  \citet{gretton2007kernel,gretton2012kernel} proposed maximum mean discrepancy (MMD) two-sample tests via embedding the probability distribution into a reproducible kernel Hilbert space (RKHS). \cite{eric2008testing} proposed the regularized MMD test by regularizing  eigenvalues of the kernel matrix. In addition, a class of approaches based on kernel density estimation was proposed
\citep{anderson1994two,cao2006empirical,martinez2008k,zhan2014testing}. One common challenge for MMD based and kernel density-based testing approaches is the choice of tuning parameters, e.g., kernel bandwidth or roughness penalty, since these parameters sensitively affect the methods' power.
Besides,  they have some drawbacks when applied to data of long-tailed distributions:  since the kernel bandwidth is fixed across the entire sample \citep{silverman1986density}, they tend to  lack  power to detect changes in tails. In  many  applications  such  as  gene  expression  analyses,  metagenomics,  and economics,  long-tailed  distributions  are  widespread.

To overcome these limitations, we propose a likelihood-based test which can automatically adapt to densities with different shapes and further develop a dadaptive tuning criteria to choose the penalization parameter.
Let $X$ be a continuous random vector and $Z$ be a binary random variable that indicates the group information. Instead of directly comparing the two group densities, we characterize the dependence between $X$ and $Z$ through its log-transformed joint  density $\eta(x,z)$ within 
a space  $\cH$. The key idea is to uniquely decompose the log-transformed joint  density $\eta\in \cH$ into the main effects $\eta_X, \eta_Z$ and the interaction effect $\eta_{XZ}$ through a novel probabilistic decomposition of $\cH$ so that 
the magnitude 
of the interaction exactly quantifies the density difference between two groups.  The two-sample test  is thus equivalent to the interaction test
\begin{equation}\label{eq:hypothesis}
    H_0: \eta_{XZ}(x, z)=0 \mbox{ vs. } H_1: \eta_{XZ}(x, z)\ne 0.
\end{equation} 
We propose a penalized likelihood ratio (PLR) test by evaluating the penalized log-likelihood functional of $\eta$ under $H_0$ and $H_1$, and establish its null distribution as a chi-square distribution. 
Compared with distance-based testing methods, the proposed PLR test can be easily  generalized to a $k$-sample test by letting $Z\in\{1,\dots, k\}$.  We further propose a data-adaptive rule to select the tuning parameter to guarantee testing optimality. The PLR test makes full use of the distribution information and is sensitive to the density difference between the null and alternative hypotheses.

This work has three main contributions. First, we propose a probabilistic decomposition of the tensor product RKHS in Section \ref{sec:3}. Existing references on functional decomposition without considering probabilistic measures  \citep{gu2013smoothing,wahba1990spline} mainly focus on estimation while leaving the hypothesis testing an open problem. 
Embedding the probability measures of $X$ and $Z$ into the tensor product decomposition of $\cH$, we can %
transform the problem of density  comparison to the problem of significance test of the interaction between $X$ and $Z$,  which  provides a foundation to establish the minimax testing principle (see Section \ref{sec:minimax}). This new probabilistic decomposition framework can be generalized to a broader class of dependence tests, including higher order independence tests and conditional independence tests, by using the magnitudes of the decomposed terms to measure the corresponding dependency. 
Second, we establish the minimax lower bound for density comparison problems based on the Bernstein width \citep{pinkus2012n}. 
Existing minimax lower bounds of the testing rate are commonly derived based on Gaussian sequence models \citep{ingster1989asymptotic, ingster1993asymptotically, wei2018local} in a simple regression setting, thus cannot be adapted to density comparison. In contrast, our result can be easily generalized to a wide range of dependence testing problems. We further prove the PLR based two-sample test is minimax optimal. Third, we reveal an interesting connection between PLR and MMD test in   Section \ref{sec:mmd}. 
We show that the MMD test (with a particularly selected kernel) is exactly the squared norm of the gradient of the log-likelihood ratio. Compared with our proposed PLR test, the log-likelihood ratio without a penalty term does not enjoy the  minimax optimality. Parallel to our work, \citet{li2019optimality} proposed a normalized MMD by appropriately choosing scaling parameters of the Gaussian kernel, and established its minimax property. Similar to the original MMD \citep{gretton2007kernel}, the approach in \citet{li2019optimality} is also based on the fixed kernel bandwidth, which reduces the sensitivity when the underlying densities are long-tailed. However, our proposed approach is based on the penalized likelihood estimators,  which can automatically adapt to long-tailed distributions. As shown in various simulation and real data studies in Sections \ref{sec:simulation} and \ref{sec:realdata}, our proposed test shows a higher power when  the underlying densities have complex features such as long-tails and multi-modality.

The penalized likelihood-based estimation is widely used under nonparametric regression settings  \citep{silverman1986density, eggermont2001maximum, wood2011fast}. 
There are only a few exceptions of applying the likelihood principle to hypothesis testing in nonparametric models. \cite{fan2001generalized} proposed
a generalized likelihood ratio test statistic based on a local polynomial estimator of the regression function. \cite{shang2013local} laid out a coherent theory of local and global inference for a smoothing spline regression model. However, these developments all focus on the likelihood-based inference under regression settings. 
To the best of our knowledge, our proposed PLR test is the first likelihood-based test under nonparametric density estimation framework with an optimality guarantee.

The rest of this paper is organized as follows. In Section \ref{sec:plr}, we 
 construct our proposed penalized likelihood ratio test.
Section \ref{sec:plr} introduces the  construction of the probabilistic decomposition of tensor product RHKS and main theoretical results, including the asymptotic distribution of the PLR test and its power performance. Section \ref{sec:minimax} established the minimax lower bound of density comparisons. In Section \ref{sec:mmd}, we build the connection between our PLR test and the MMD test. 
In Section \ref{sec:simulation}, 
we demonstrate the finite sample performance of our test through simulation studies.
Section \ref{sec:realdata} is the analysis of two real-world examples using our test. 
Section \ref{sec:disc} contains some discussion.
Section \ref{sec:appendix} is the appendix holding the proofs of the main results.
Additional proofs for the lemmas are deferred to a supplement document.

\vspace{-10pt}
\section{Penalized Likelihood Ratio for Two-sample Test}\label{sec:plr}

The two-sample problem can be stated as follows. Suppose that we have $n$ independent $d$-dimensional observations $X_i$'s and the associated labels $Z_i$'s, where $Z_i$ is either 0 or 1 indicating that $X_i$ is taken from either the population with a probability density function $f_0$ or another population with  a probability density function $f_1$. We aim to test whether $f_0$ and $f_1$ are the same. Other than a smoothness constraint, we will not impose any other constraints on the probability density functions $f_0$ and $f_1$. 
For the convenience of presentation, instead of the  marginal probability formulation, we consider an equivalent formulation 
in terms of  conditional independence. That is, we have $n$ i.i.d. observations $\bY_i=(X_i, Z_i)$, $i=1,\dots,n$, taken from a population
$Y=(X,Z)$ with a joint probability density $f(x,z)$. Let $f_{X|Z=z}(x)$ be the conditional density of $X$ given $Z=z$ for $z=0,1$. The two-sample problem is equivalent to testing whether $X$ and $Z$ are independent, or whether $f_{X|Z=0} (\cdot) = f_{X|Z=1}(\cdot)$, i.e., 
\begin{equation}\label{eq:twosampletest}
H_0: f_{X|Z=0} (\cdot) = f_{X|Z=1}(\cdot) \quad v.s. \quad  H_1:f_{X|Z=0} (\cdot) \ne f_{X|Z=1}(\cdot)
\end{equation}

We characterize the dependence between $X$ and $Z$  by their interaction  with respect to their joint density, and show that testing the significance of such interaction is equivalent to the two-sample test. In order to characterize the interaction between $X$ and $Z$, we introduce a probabilistic decomposition.  Let $\eta(x,z) = \log(f(x,z))$ be the log-transformed joint density function. We define two averaging operators acting on the bivariate function $\eta(x,z)$. For any $x$, the operator $\cA_x$ maps $\eta(x,z)$ to $\E_X\eta(X,z)$, a function in $z$; and for any $z$, the operator $\cA_z$ maps $\eta(x,z)$ to $\E_Z\eta(x,Z)$. The interaction term is then
\begin{equation}\label{eq:interaction}
\eta_{XZ}(x,z) = (\cI-\cA_x)(\cI-\cA_z)\eta(x,z)
\equiv \eta(x,z)-(\cA_x \eta)(z) -(\cA_z \eta)(x) +\cA_x\cA_z \eta, \end{equation}
where $\cI$ is the identity operator. Note that \eqref{eq:interaction} is essentially derived from a functional ANOVA decomposition of $\eta(x,z)$ where $\cA_x\cA_z \eta$ is the constant, $(\cI-\cA_x)\cA_z\eta$ and $(\cI-\cA_z)\cA_x\eta$ are respectively the main effects of $x$ and $z$, and $(\cI-\cA_x)(\cI-\cA_z)\eta$ is the interaction effect. A straightforward derivation shows that the two sample test is equivalent to testing whether $\eta_{XZ}$ is zero or not (Proposition~\ref{lemma:1} in the Appendix).

We assume that $\eta$ is in a reproducing kernel Hilbert space (RKHS) $\cH$ and let $\cH_0 = \{\eta\in \cH \mid \eta_{XZ}=0\}$ be the subspace of $\cH$ containing all the bivariate functions whose ANOVA decomposition has a zero interaction term.  Based on Proposition~\ref{lemma:1},
 the two-sample test problem in (\ref{eq:twosampletest}) is equivalent to testing
\begin{equation}\label{eq:newtwosampletest}
H_0: \eta \in \cH_0 \quad v.s. \quad H_1: \eta\in \cH\backslash\cH_0.
\end{equation}

Consider estimating $\eta$ by the minimizer of the penalized likelihood 
\begin{equation}\label{penalizedlikelihood}
\ell_{n,\lambda}(\eta) = -\frac{1}{n}\sum_{i=1}^{n}\eta(x_i, z_i) + \sum_{z\in\{0,1\}}\int_{\mathcal{X}}e^{\eta(x,z)}dx + \frac{\lambda}{2}J(\eta),
\end{equation}
where the first two sums from the negative log-likelihood representing the goodness-of-fit, $J(\cdot)$ is a quadratic functional enforcing a roughness penalty on $\eta$, and  $\lambda>0$ is a tuning parameter controlling the trade-off. Note that the integrals in \eqref{penalizedlikelihood} are to guarantee the unitary constraint of a probability density function (see Theorem 3.1 in \citet{silverman1982estimation}).

We propose the following penalized likelihood ratio test statistic
\begin{equation}\label{eq:plr1}
PLR = 2\left\{\sup_{\eta\in \cH_0}\ell_{n,\lambda}(\eta)-\sup_{\eta\in\cH}\ell_{n,\lambda}(\eta)\right\},
\end{equation}
where the first and second terms are respectively the optimal penalized likelihoods under the reduced model $\cH_0$ and the full model $\cH$.

\vspace{-10pt}
\subsection{Penalized likelihood functional under the full model}
Under the full model,
the minimization of (\ref{penalizedlikelihood}) is performed in  $\mathcal{H}$. Let $\cH^{\angular{X}}$ be an RKHS of functions on the marginal domain $\bbR^d$ and $\cH^{\angular{Z}}$ be an RKHS of functions on $\{0,1\}$. Then the full space $\cH=\cH^{\angular{X}}\otimes \cH^{\angular{Z}}$ is their tensor product and also an RKHS, where $\otimes$ denotes the tensor product of two linear spaces. Correspondingly, if $K^{\angular{X}}$ and $K^{\angular{Z}}$ are respectively the reproducing kernels (RKs) uniquely associated with the RKHS $\cH^{\angular{X}}$ and $\cH^{\angular{Z}}$, then the RK for $\cH$ is simply the product of $K^{\angular{X}}$ and $K^{\angular{Z}}$, that is, 
\begin{equation}
    K(\bY_i,\bY_j) = K^{\angular{X}}(X_i, X_j)K^{\angular{Z}}(Z_i, Z_j).
\end{equation}

One example for $\cH^{\angular{X}}$ is the $m$th order homogeneous Sobolev space $\{ f \;|\;  f^{(m)}\text{ is square integrable}, \\f^{(j)} \text{ is absolutely continuous and } f^{j}(0)=f^{j}(1)\mbox{ for } j=0,1 \dots, m-1\}$  associated with the kernel
$
K^{\angular{X}}(X_i, X_j) = 1 + (-1)^{m-1}k_{2m}(X_i-X_j)
$,
where $k_{2m}(x)$ %
is the $2m$-th order scaled Bernoulli polynomial \citep{abramowitz1948handbook}.
When $m=2$, $k_4(x) = \frac{1}{24}((x-0.5)^4- 0.5(x-0.5)^2+\frac{7}{240})$ and the corresponding $K^{\angular{X}}$ is known as the homogeneous cubic spline kernel. An example for the discrete kernel is $K(Z_i,Z_j) = \mathbbm{1}_{\{Z_i=Z_j\}}$. 

By the representer theorem \citep{kimeldorf1971some}, the minimizer $\eta(\cdot)$ of \eqref{penalizedlikelihood} has the form 
\begin{equation}\label{eq:etakernel1}
    \eta(\cdot) = \sum_{i=1}^n K(\bY_i, \cdot) c_i = \zeta^T \bc,  \quad \forall \eta \in \cH
\end{equation}
where $\zeta^T = ( K(\bY_1, \cdot), \cdots, K(\bY_n, \cdot) )$ is the vector of functions obtained from kernel $K$ with its first argument fixed at $\bY_i$, and $\bc = (c_1, \cdots, c_n)$ is coefficient vector. 
This representation converts the infinite-dimensional minimization problem of (\ref{penalizedlikelihood}) with respect to $\eta$ to the finite-dimensional optimization problem with respect to the coefficient vector $\bc$ by solving
\begin{equation} 
     \widehat{\bc}  = \argmin_{\bc}
     \left\{-\frac{1}{n}\mathbf{1}_n^T Q \bc + \int_{\cY}\exp\{\zeta^T \bc\}dy + \frac{\lambda}{2} \bc^T Q \bc \right\} .\label{eq:c}\\ 
\end{equation}
where $\mathbf{1}_n$ is a $n\times 1$ vector of ones, $Q\in R^{n\times n}$ is the empirical kernel matrix with its $(i,j)$-th entry being $Q_{ij}= K(Y_i,Y_j)$, and the second term is the same as the second term in (\ref{penalizedlikelihood}) with summation and integration over  $(x,z)$ replaced by integration over $y$ for the convenience of presentation.  The objective function in (\ref{eq:c}) is strictly convex \citep{tapia1978nonparametric}. %
Its optimization with respect to $\bc$ can be performed via a standard convex optimization procedure such as the Newton-Raphson algorithm; see, e.g.,
\cite{gu2013smoothing} and \cite{wang2011smoothing}. The integrals in (\ref{eq:c}) can be calculated by numerical integration (see 7.4.2 in \citet{gu2013smoothing} for details). %
When $n$ is large, the representation (\ref{eq:etakernel1}) involves a large number of coefficients, which may lead to numerical instability. To tackle this, one may consider only a subsample of $\{\cY_i: i=1,\ldots,n\}$ to use in the presentation \citep{ma2015efficient}. In general, we denote by
\begin{equation}\label{eq:etafull}
\widehat{\eta}_{n,\lambda} = \zeta^T \widehat{\bc}
\end{equation}
the penalized maximum likelihood estimate under the full model.

\vspace{-10pt}
\subsection{Penalized likelihood functional under the reduced model}
Under $H_0$ in (\ref{eq:newtwosampletest}), we denote  $\widehat{\eta}^0_{n,\lambda}$ to  be the penalized likelihood estimator of $\eta$, that is, 
\begin{equation}\label{eq:null:est}
\hetanull = \text{argmin }_{\eta\in\cH_0} \ell_{n,\lambda}(\eta).
\end{equation}
In Section 3.1, we show that $\cH_0$ is also an RKHS equipped with kernel function $K^0(\cdot, \cdot)$, which enables us to use a similar reparameterization trick to solve the problem in (\ref{eq:null:est}). In the following, we show the expression of the kernel function $K^0$.
\begin{equation*}
    K^0(Y_i, Y_j) = K_0^{\angular{X}}(X_i,X_j)K_0^{Z}(Z_i,Z_j) +  K_1^{\angular{X}}(X_i,X_j)K_0^{Z}(Z_i,Z_j) +  K_0^{\angular{X}}(X_i,X_j)K_1^{Z}(Z_i,Z_j),
\end{equation*}
where $K^{\angular{X}}_0(X_i,X_j) =  \E_X[K^{\angular{X}}(X, X_j)] + \E_{X}[K^{\angular{X}}(X_i,X)] - \E_{X,\tilde{X}} K^{\angular{X}}(X,\tilde{X})$, $K^{\angular{X}}_1 = K^{\angular{X}}-K_0^{\angular{X}}$, $K^{\angular{Z}}_0(Z_i,Z_j) = \omega_{Z_i} + \omega_{Z_j} - \sum_{\ell=0}^{1} \omega_{\ell}^2$, $K_1^{\angular{Z}}=K^{\angular{Z}}-K_1^{\angular{Z}}$, $\omega_0=\Pr(Z=0)$, and  $\omega_1=\Pr(Z=1)$. The detailed derivation of $K^0$ depends on our proposed  probabilistic decomposition of $\cH$, and is deferred to Section 3.1.

Similar to (\ref{eq:etakernel1}), we apply the representation theorem and have 
\begin{equation}\label{eq:etakernel0}
    \eta(\cdot) = \sum_{i=1}^n K^0(\bY_i, \cdot) c_i  = \zeta_0^T \bc_0,  \quad \forall \eta \in \cH_0.
\end{equation}
The penalized likelihood estimators can be obtained by first solving the quadratic program 
\begin{align}
  \widehat{\bc}_0  = &\argmin_{\bc}\left\{-\frac{1}{n}\mathbf{1}_n^T Q_0 \bc + \int_{\cY}\exp\{\zeta_0^T \bc\} + \frac{\lambda}{2} \bc^T Q_0 \bc  \right\} \label{eq:c0}
\end{align}
where the $(i,j)$-th entry of $Q_0$ is $K^0(\bY_i, \bY_j)$ and the $(i,j)$-th entry of $Q_0$ is $K^{0}(\bY_i, \bY_j)$. Numerically, We could express 
\begin{multline*}
Q_0 = [(I_n-H)Q^{\angular{X}}(I_n-H)]\circ [(I_n-H)Q^{\angular{Z}}(I_n-H)] + [HQ^{\angular{X}}H]\circ [(I_n-H)Q^{\angular{Z}}(I_n-H)]\\ + [(I_n-H)Q^{\angular{X}}(I_n-H)]\circ [HQ^{\angular{Z}}H]
\end{multline*}
where $Q^{\angular{X}}$ is the empirical kernel matrix of $\cH^{\angular{X}}$ with $(i,j)$-th entry $Q^{\angular{X}}_{ij}=K^{\angular{X}}(X_i, X_j)$, $Q^{\angular{Z}}$ is the empirical kernel matrix of $\cH^{\angular{Z}}$ with $(i,j)$-th entry $Q^{\angular{Z}}_{ij}=K^{\angular{Z}}(Z_i, Z_j)$, and $H= I_n-\frac{1}{n}\mathbf{1}_n\mathbf{1}_n^T$ for $I_n$ as a $n\times n$ identity matrix and $\mathbf{1}_n$ as a $n\times 1$ vector of ones. 
Then we solve the quadratic optimization similar to (\ref{eq:c}) and  output the function estimates
\begin{equation}\label{eq:etareduce}
	 \hetanull = {\zeta^0}^T \widehat{\bc}^0.
\end{equation}

\vspace{-10pt}
\subsection{Test statistics} 
Plugging the maximizers of penalized likelihood functional under full and reduced model into (\ref{eq:plr1}), we have the penalized likelihood ratio (PLR) statistic
\begin{equation}
PLR_{n,\lambda} = \ell_{n,\lambda}(\widehat{\eta}^0_{n,\lambda}) - \ell_{n,\lambda}(\widehat{\eta}_{n,\lambda}).
\end{equation}
We will show in Section  \ref{sec:main} that  $PLR_{n,\lambda}$ is asymptotically $\chi^2$ distributed under $H_0$ in the sense that 
$
(2b_{n,\lambda})^{-1/2}(2 PLR_{n,\lambda} -b_{n,\lambda})\to N(0,1)
$
with $b_{n,\lambda}$ diverges for a wide range of $\lambda$. Also, the distribution of $PLR_{n,\lambda}$ is independent with the nuisance parameters in $\cH_0$, which fulfills the Wilks' phenomenon. 

For the nonparametric two-sample test, the parameter space under $H_0$ is infinite-dimensional as $n\to \infty$. The assumptions of the Neyman-Pearson Lemma cannot be satisfied. Thus the uniformly most powerful test may not exist in general. We evaluate the power performance by the minimax rate of testing, which is defined as the minimal distance between the null and the alternative hypotheses such that valid testing is possible \citep{ingster1989asymptotic}. 
 For any generic 0-1 valued testing rule  $\Phi=\Phi(\bY_1,\ldots,\bY_n)$ and a separation rate $d_n>0$, 
define the \textit{total error} $\err(\Phi,d_n)$ of $\Phi$ under $d_n$ as 
\begin{equation}\label{def:err}
\err(\Phi,  d_n) = \bE_{H_0} \left\{\Phi\right\} + \sup_{\|\eta_{XZ}\|_2 \geq d_n} \bE_{\eta}\left\{1-\Phi\right\},
\end{equation}
where $\bE_{H_0}\left\{\cdot\right\}$ denotes the expectation under $H_0$, $d_n$ measures the distance between the null and the alternative hypotheses. 
The first and second terms on the right side of (\ref{def:err}) represent type I and type II errors of $\Phi$ respectively. %
In Section \ref{sec:3}, we show that the distinguishable rate of our proposed PLR test is related to tuning parameter $\lambda$ and derive the optimal distinguishable rate by carefully selecting $\lambda$. A data-adaptive tuning method is developed for applications. 
In Section \ref{sec:minimax}, we will use information theory to establish the minimum separation rate $d_n$ for general testing rules, which extends the minimax testing principle pioneered in  \cite{ingster1989asymptotic} to density comparison.

\vspace{-10pt}
\section{Theoretical  Properties of  PLR Test}\label{sec:3}
 
In this section, we first introduce the probabilistic decomposition of tensor produce RKHS, enabling us to construct kernel on the subspace $\cH_0$. Such decomposition is also of independent interest for studying different kinds of dependence among random variables. 
We derive the null asymptotic distribution of our proposed test statistics and derive the optimal power for the proposed test. Then we develop data-adaptive tuning method to choose the penalty parameter. 

\vspace{-10pt}
\subsection{Probabilistic decomposition of the tensor product RKHS}\label{subsec:decomposition}
We assume that bi-variate function $\eta(x,z)$ 
belong to a tensor product RKHS $\cH = \cH^{\angular{X}} \otimes \cH^{\angular{Z}}$, in which  $\cH^{\angular{X}}$ and $\cH^{\angular{Z}}$ represent the marginal RKHS of $X$ and $Z$ respectively. We aim to decompose $\cH$ into orthogonal subspaces with a hierarchical structure similar to the main effects and interactions in smoothing spline ANOVA \citep{wahba1990spline,gu2013smoothing,lin2000tensor,wang2011smoothing}, while embedding the probabilistic distribution of $X$ and $Z$ into the decomposition. Such decomposition enables us to convert the two-sample test problem into testing whether the interaction presents or not.
It includes two steps: decompose each marginal RKHS into mean and main effect; apply distributive law to expand the tensor product of marginal RKHS into a series of subspaces.

We first introduce the probabilistic tensor decomposition of the discrete domain $\cH^{\angular{Z}}$ via a  probabilistic averaging operator. The kernel on $\cH^{\angular{Z}}:=\{f: f\in \{0,1\}\} = \bbR^2$ with the Euclidean inner product $(\inner{\cdot}{\cdot}_{\cH^{\angular{Z}}})$ is $K^{\angular{Z}}(z,\tilde{z}) = \mathbbm{1}_{\{z=\tilde{z}\}}$. 
Consider a discrete probabilistic measure $\PP$ on $\cZ=\{0,1\}$ such that $\PP(Z=j)=\omega_j\geq 0$ with 
$w_0 + w_1 =1$. Let $\bomega= (\omega_0, \omega_{1})$, and define 
the probabilistic averaging operator as $\cA_Z := f \to \E_Z f(Z) = \inner{\bomega}{f}_{\cH^{\angular{Z}}}$. %
Since $\E_Z[K^{\angular{Z}}_Z] = \bomega $, we can rewrite the probabilistic averaging operator as $\cA_Z := f \to \E_Z f(Z) = \inner{\E_Z[K^{\angular{Z}}_Z]}{f}_{\cH^{\angular{Z}}}$. Then $\E_Z[K^{\angular{Z}}_Z]$ can be treated as a mean embedding of $\PP$ in $\cH^{\angular{Z}}$.  
We further define the tensor sum decomposition of $\cH^{\angular{Z}}$ as 
\begin{equation}\label{eq:decom_z}
\cH^{\angular{Z}} = \cH^{\angular{Z}}_0 \oplus \cH^{\angular{Z}}_1 := span\{\E_Z K^{\angular{Z}}_{Z}\} \oplus \{f\in\cH: \E_Z\{f(Z)\}=0\},
\end{equation}
where $\cH^{\angular{Z}}_0$ is the grand mean space, $\cH^{\angular{Z}}_1$ is the main effect space. 
Each subspace in (\ref{eq:decom_z}) is an RKHS with their corresponding kernels stated in Lemma \ref{lemma:pdk}.
If more than two samples present, we can extend the decomposition in (\ref{eq:decom_z}) to a general discrete domain where $Z\in\{0,\cdots, a\}$ for $a\geq 1$ by considering $\cH^{\angular{Z}}$ as an $a$-dimensional Euclidean space.

Consider the continuous random variable $X\in \cX$ and let $\PP$ be a probability measure on $\cX$. We suppose $\cH^{\angular{X}}$ is the $m$th order Soblev space, i.e.,
\begin{multline*}
\cH^{\angular{X}}= \{ f \in L^2(\PP) \;|\; f^{(j)} \mbox{ is absolutely continuous and }
\mbox{ for } j=0 , 1 , \dots , m-1, f^{(m)}\in L_2(\PP)\} ,
\end{multline*}
with inner product  $\inner{\cdot}{\cdot}_{\cH^{\angular{X}}}$. The results also hold for its homogeneous subspace. Let $K^{\angular{X}}$ be the corresponding kernel satisfying $\inner{f}{K^{\angular{X}}_x}_{\cH^{\angular{X}}}=f(x)$ for any $f\in \cH^{\angular{X}}$. 
Similarly, the probabilistic averaging operator is  $\cA_X := f \to \E_X f(X) = \E_X \inner{K^{\angular{X}}_X}{f}_{\cH^{\angular{X}}} = \inner{\E_X K^{\angular{X}}_X}{f}_{\cH^{\angular{X}}}$. $\E_X K^{\angular{X}}_X$ has the same role as $\bomega$ in the Euclidean space. Then, the tensor sum decomposition of a functional space is defined as  
\begin{equation} 
\cH^{\angular{X}} = \cH^{\angular{X}}_0 \oplus \cH^{\angular{X}}_1 := span\{\E_X K^{\angular{X}}_{X}\} \oplus \{f\in \cH^{\angular{X}} : \cA_X f=0\}.
\end{equation}
Analogously,
we name $\cH^{\angular{X}}_0$ as the grand mean space and $\cH^{\angular{X}}_1$ as the main effect space. $\E_X K^{\angular{X}}_X$ is known as the kernel mean embedding which is well established in the statistics literature \citep{berlinet2011reproducing}.  The construction of  kernel functions for $\cH^{\angular{X}}_0$ and $\cH^{\angular{X}}_1$  are included in \ref{lemma:continuouskernel}.

We are now ready to consider the RKHS  $\cH = \cH^{\angular{X}} \otimes \cH^{\angular{Z}}$ on the product domain $\cY= \cX\times \cZ$.  Applying the distributive rule, the decomposition of $\cH$ is written as
\begin{equation}\label{eq:anovaH}
\mathcal{H} = (\cHx_0\oplus\cHx_1)\otimes(\cH^{\angular{Z}}_0\oplus\cH^{\angular{Z}}_1) \equiv \mathcal{H}_{00}\oplus\mathcal{H}_{10}\oplus\mathcal{H}_{01}\oplus\cH_{11},
\end{equation}
where $\cH_{ij} = \cHx_i\otimes\cH^{\angular{Z}}_j$ for $i=0,1$ and $j=0,1$. Analogous to the classic ANOVA, $\cH_{10}$ and $\cH_{01}$ are the RKHS's for the main effects, and $\cH_{11}$ is the RKHS for the interaction. 
We call the decomposition of  $\mathcal{H}$ in (\ref{eq:anovaH}) as  the {\it probabilistic decomposition} of the tensor product RKHS $\cH$ since it embeds the probability measure of the random variable $X$ and $Z$. Based on Theorems 2.6 in \cite{gu2013smoothing}, we can construct the kernels $K^{00}, K^{10}, K^{01}$ and $K^{11}$ for the subspaces $\cH_{00}, \cH_{10}, \cH_{01}$ and $\cH_{11}$ accordingly (detailed construction is given in \ref{lemma1.2}).

\vspace{-10pt}
\subsection{Asymptotic distribution and Wilks' phenomenon}\label{sec:main}

In this subsection, we present the asymptotic distribution of our PLR test (see Theorem \ref{thm1}).
The proof relies on a technical lemma about the eigen-structures of $\cH_0$ and $\cH$; see Lemma \ref{sim:diag:V:J} below.
For any $\eta,\teta\in\cH$, define
\begin{equation}\label{inner}
\inner{\eta}{\widetilde{\eta}} = V(\eta, \widetilde{\eta})+ \lambda J(\eta, \widetilde{\eta}),
\end{equation}
where $V(\eta,\widetilde{\eta})=\bE_{\eta^\ast}\{\eta(\bY)\tilde{\eta}(\bY)\}$
with expectation taken under the true  $\eta^\ast$, and $J$ is a bilinear form corresponding to (\ref{penalizedlikelihood}).
It holds that $\cH$ and $\cH_0$, endowed with the inner product (\ref{inner}), are both RKHSs; see Lemma \ref{rkhs:cH} in the Appendix. In the following lemma, we characterze the eigenvalues and eigenvectors of the Rayleigh quotient $V/J$.

\begin{lemma}\label{sim:diag:V:J}

\begin{enumerate}[label=(\alph*)]
\item There exist a sequence of functions $\{\xi_p\}_{p=1}^{\infty}\subset\cH$
and a sequence of nonnegative eigenvalues $\{\rho_p\}_{p=1}^\infty$ with 
$\rho_p\asymp p^{2m/d}$ such that
\begin{equation}
\textrm{$V(\xi_{p}, \xi_{p^\prime}) = \delta_{p, p^\prime},\,\,\,\,J(\xi_{p}, \xi_{p^\prime}) = \rho_p\delta_{p, p^\prime}$,
for all $p,p'\ge1$,}
\end{equation}
and that any $\eta\in\cH$ can be written as $\eta=\sum_{p=1}^\infty V(\eta,\xi_p)\xi_p$.
\item Moreover, there exists a proper subset $\{\rho_p^0, \xi^0_p\}_{p=1}^\infty$ of $\{\rho_p, \xi_p\}_{p=1}^\infty$ satisfying 
$\{\xi^0_p\}_{p=1}^\infty\subset\cH_0$ and 
for any $\eta\in\cH_0$, $\eta=\sum_{p=1}^\infty V(\eta,\xi^0_p)\xi^0_p$.
Convergence of both series holds under (\ref{inner}).
\item $\rho_p^\perp\asymp p^{2m}$, where
$\{\rho^\perp_p\}_{p=1}^\infty\subset\{\rho_p\}_{p=1}^\infty$ is a subset of eigenvalues corresponding to $\{\xi^\perp_p\}_{p=1}^\infty
\equiv\{\xi_p\}_{p=1}^\infty\backslash\{\xi^0_p\}_{p=1}^\infty$. The set $\{\xi^\perp_p\}_{p=1}^\infty$
generates the orthogonal complement of $\cH_0$ under the inner product (\ref{inner}).
\end{enumerate}
\end{lemma}

Lemma \ref{sim:diag:V:J} introduces an eigensystem
that simultaneously diagonalizes the bilinear forms $V$ and $J$. This eigensystem does not depend on the unknown null density, and only depends on the functional space $\cH$. 
Moreover, $\cH_0$ can be generated by a proper subset of the eigenfunctions,
which is crucial for analyzing the likelihood ratios.

Let $\inner{\cdot}{\cdot}_0$ denote the restriction of $\inner{\cdot}{\cdot}$ on the subspace $\cH_0$.
Specifically, for any $\eta, \tilde{\eta}\in\cH_0$,
$\inner{\eta}{\widetilde{\eta}}_{0} = \inner{\eta}{\widetilde{\eta}}$.
Then $\cH$ and $\cH_0$ are both RKHS's endowed with these inner products.
\begin{lemma}\label{rkhs:cH}
$(\cH,\inner{\cdot}{\cdot})$ and $(\cH_0,\inner{\cdot}{\cdot}_0)$ are both RKHS's with the corresponding inner products.
\end{lemma}

Following Lemma \ref{rkhs:cH}, there exist reproducing kernel functions $\tilde{K}(\cdot,\cdot)$ and $\tilde{K}^0(\cdot,\cdot)$ defined on $\cY\times\cY$
satisfying, for any $\by\in\cY$, $\eta\in\cH$, $\widetilde{\eta}\in\cH_0$:
\begin{eqnarray}\label{rkhs:property}
&&\tilde{K}_\by(\cdot)\equiv \tilde{K}(\by,\cdot)\in\cH,\,\,\quad \tilde{K}^0_\by(\cdot)\equiv \tilde{K}^0(\by,\cdot)\in \cH_0,\nonumber\\
&&\inner{\tilde{K}_\by}{\eta}=\eta(\by),\,\,\quad \inner{\tilde{K}^0_\by}{\widetilde{\eta}}_0=\widetilde{\eta}(\by).
\end{eqnarray}

We further introduce positive definite self-adjoint operators $W_\lambda: \cH\to\cH$ and
$W_\lambda^0: \cH_0\to\cH_0$ such that
\begin{align}\label{sec2:eq:0}
&\textrm{$\inner{W_\lambda \eta}{\widetilde{\eta}}= \lambda J(\eta,\widetilde{\eta})$ for all $\eta,\widetilde{\eta}\in\cH$,}\nonumber\\
&\textrm{$\inner{W^0_\lambda \eta}{\widetilde{\eta}}_0= \lambda J_0(\eta,\widetilde{\eta})$ for all $\eta,\widetilde{\eta}\in\cH_0$},
\end{align}
where $J_0(\eta,\widetilde{\eta})=\theta^{-1}_{01} J_{01}(\eta,\teta) + \theta^{-1}_{10} J_{10}(\eta,\teta)$ 
is the restriction of $J$ over $\cH_0$. By (\ref{sec2:eq:0}) we get $\langle
\eta,\widetilde{\eta}\rangle=V(\eta,\widetilde{\eta})+\langle
W_\lambda \eta,\widetilde{\eta}\rangle$, $\langle
\eta,\widetilde{\eta}\rangle_0=V(\eta,\widetilde{\eta})+\langle
W^0_\lambda \eta,\widetilde{\eta}\rangle_0$. In the following, we give the explicit expression of $\widetilde{K}_{y}(\cdot)$ and $W_{\lambda}\xi_p(\cdot)$.

\begin{proposition}\label{proposition:a1}
For any $\by\in\cY$ and $\eta\in\mathcal{H}$, we have 
\begin{align*}
\|\eta\|^2 &=\sum_{p=1}^{\infty}|V(\eta, \xi_p)|^2(1+\lambda\rho_p),\\
\tilde{K}_\by (\cdot) &= \sum_{p=1}^{\infty}\frac{\xi_p(\by)}{1+ \lambda \rho_p}\xi_p(\cdot), \quad
\tilde{K}^0_\by (\cdot) = \sum_{p=1}^{\infty} \frac{\xi^0_p(\by)}{1+ \lambda \rho^0_p}\xi^0_p(\cdot),\\
W_{\lambda}\xi_p(\cdot)& = 
    \frac{\lambda\rho_p}{1+\lambda\rho_p}\xi_p(\cdot),  \quad
W^0_{\lambda}\xi^0_p(\cdot) = 
    \frac{\lambda\rho^0_p}{1+\lambda\rho^0_p}\xi^0_p(\cdot).
\end{align*}
where $\{\rho_p^0, \xi^0_p\}_{p=1}^\infty$ and $\{\rho_p, \xi_p\}_{p=1}^\infty$ are eigensystems defined in Lemma \ref{sim:diag:V:J}.
\end{proposition}

As shown in Proposition \ref{proposition:a1}, the eigenvalues for $\tilde{K}$ are $\{(1+\lambda \rho_p)^{-1}\}_{p=1}^\infty$,  having a slower decay rate due to the scaling by $\lambda$. $\tilde{K}$ can be viewed as a scaled kernel comparing with the product kernel $K^{\cH}=  K^{00} + K^{01} + K^{10} +K^{11}$ introduced in Lemma \ref{lemma1.2}. Note that 
$$\textrm{trace}(\tilde{K}) = \sum_{p=1}^\infty (1+ \lambda \rho_p)^{-1} \asymp \lambda^{-d/(2m)}$$ 
is the effective dimension that measures the complexity of $\cH$; see \cite{bartlett2005local}. 

Next, we will derive the null asymptotic distribution of the PLR statistics, which relies on the Taylor expansion of the PLR functional. First, we introduce the Frech\'{e}t derivatives of the log-likelihood functional.
Let $D,D^2,D^3$ be the first-, second- and third-order Frech\'{e}t derivatives of $\ell_{n,\lambda}(\eta)$.
Based on the above notations, these derivatives can be summarized as follows.
Let $\by=(x,z)$. 
For any $\eta,\Delta\eta_1,\Delta\eta_2,\Delta\eta_3\in\cH$,
\begin{align}
D\ell_{n,\lambda}(\eta)\Delta\eta_1 &= -\frac{1}{n}\sum_{i=1}^{n}\Delta\eta_1(\bY_i) + \int_{\cY}\Delta\eta_1(\by) e^{\eta(\by)}d\by + \lambda J(\eta,\Delta\eta_1)\nonumber\\
& = \inner{-\frac{1}{n}\sum_{i=1}^n \tilde{K}_{\bY_i} + \mathbb{E}_{\eta}\tilde{K}_{\bY}+W_\lambda \eta}{\Delta \eta_1}\nonumber\\
& \equiv \inner{S_{n,\lambda}(\eta)}{\Delta \eta_1}, \label{eq:Frechet1}
\end{align}
\begin{equation}
D^2\ell_{n,\lambda}(\eta)\Delta\eta_1\Delta\eta_2 = \int_{\cY}\Delta\eta_1(\by)\Delta\eta_2(\by) e^{\eta(\by)}d\by + \lambda J(\Delta\eta_1,\Delta\eta_2), \label{eq:Frechet2}
\end{equation}
\begin{equation}
D^3\ell_{n,\lambda}(\eta)\Delta\eta_1\Delta\eta_2\Delta\eta_3 = \int_\cY\Delta\eta_1(\by)\Delta\eta_2(\by)\Delta\eta_3(\by) e^{\eta(\by)}d\by. \label{eq:Frechet3}
\end{equation}
The second equality of (\ref{eq:Frechet1}) is due to the reproducing property (\ref{rkhs:property}) and that
\[
\int_\cY\Delta\eta(\by) e^{\eta(\by)} d\by=
\bE_\eta \Delta\eta_1(\bY)=\bE_\eta\inner{\tilde{K}_{\bY}}{\Delta\eta_1} = \inner{\bE_\eta \tilde{K}_{\bY}}{\Delta\eta_1}.
\] 
We denote $S_{n,\lambda}(\eta)$ as the score function of the log-likelihood functional $\ell_{n,\lambda}$. Similarly, we define $S^0_{n,\lambda}$ as the score function of the log-likelihood functional $\ell^0_{n,\lambda}$.

Then we have the following Taylor expansion of PLR functional. 
Let $g= \hetanull - \hetafull$, we have 
\begin{align}
PLR_{n,\lambda} & = \ellnl(\widehat{\eta}^0_{n,\lambda}) - \ellnl(\widehat{\eta}_{n,\lambda}) \nonumber\\
& = D\lnl(\hetafull)g + \int_0^1\hspace{-5pt}\int_0^1 s{D^2\lnl(\hetafull+ss'g)gg ds ds'} \nonumber\\
& =  \int_0^1\hspace{-5pt}\int_0^1 s\{D^2\lnl(\hetafull+ss'g)gg - D^2\lnl(\eta^\ast)gg \}dsds' + \frac{1}{2}D^2\lnl(\eta^\ast)gg \nonumber \\ 
& \equiv  I_1+ I_2\label{eq:tylor} 
\end{align}
where $\eta^\ast$ is the underlying truth. In the proof of Theorem \ref{thm1}, we will show that $I_2$ is a leading term compared with $I_1$.
From (\ref{eq:Frechet2}), we have that
$I_2 = \frac{1}{2}\|g\|^2 = \frac{1}{2}\|\hetanull - \hetafull\|^2$.
As we will see, the asymptotic distribution of $\|\hetafull - \hetanull\|^2$ relies on Bahadur representations of $\hetanull$ and $\hetafull$.

 We further prove the following 
Bahadur representations for the difference of the two penalized likelihood estimators, by adapting an empirical processes technique in \cite{shang2013local}. Lemma \ref{lemma:s4} is crucial for proving Theorem \ref{thm1}. 
\begin{lemma}\label{lemma:s4}
Suppose $h=\lambda^{\frac{d}{2m}}$ and $nh^2\rightarrow\infty$. Then we have
\begin{align*}
n^{1/2}\|\hetafull-\hetanull\|
=n^{1/2}\|S^0_{n,\lambda}(\eta^\ast) - S_{n,\lambda}(\eta^\ast)\|+o_P(1). 
\end{align*}
where $S_{n,\lambda}(\eta^\ast)$ and $S^0_{n,\lambda}(\eta^\ast)$ are the score functions for $\ell_{n,\lambda}$ and $\ell^0_{n,\lambda}$, respectively.
\end{lemma}

This lemma shows that the main term $I_2$ in Taylor's expansion of the PLR functional is determined by the norm of the difference between the score function of $\ell_{n,\lambda}$ and the score function of $\ell^0_{n,\lambda}$. Since the score functions have the explicit expression through Proposition \ref{proposition:a1}, we can characterize the null asymptotic distribution of $I_2$ by the eigensystem introduced in Lemma \ref{sim:diag:V:J}.

Before stating our main theorem, we introduce an assumption which is commonly used in literature for deriving the rates of density estimates;
see Theorem 9.3 of \cite{gu2013smoothing}.

\begin{assumption}\label{a2}
There exists a convex set $B\subset\cH$ around $\eta^\ast$ and a constant $c>0$ such that,
for any $\eta\in B$, $c \bE_{\eta^\ast}\{\widetilde{\eta}^2(\bY)\} \leq \bE_{\eta}\{\widetilde{\eta}^2(\bY)\}$. Furthermore,
with the probability approaching one, $\widehat{\eta}_{n,\lambda}\in B$;
and under $H_0$, with the probability approaching one,
$\widehat{\eta}_{n,\lambda}^0\in B$.
\end{assumption}

This condition is satisfied when the members of $B$ have uniform upper and lower bounds on the domain $\cY$,
as well as that $\widehat{\eta}_{n,\lambda}$ and $\widehat{\eta}_{n,\lambda}^0$ are stochastically bounded.
The following theorem provides the asymptotic distribution for the PLR test statistic under Assumption \ref{a2}.

\begin{theorem}\label{thm1} Suppose $m\ge1$ and Assumption \ref{a2} holds. Let $h=\lambda^{\frac{d}{2m}}$ and 
$nh^{2m+d}=O(1)$, $nh^2 \to \infty$ as $n\to\infty$. Under $H_0$, we have 
\begin{equation}\label{lim:dist:globalLRT}
\frac{2n \cdot PLR_{n,\lambda}- \theta_\lambda}{\sqrt{2}\sigma_\lambda}\overset{d}{\longrightarrow}N(0,1),\,\,n\to\infty,
\end{equation}
where $\theta_\lambda=\sum_{p=1}^\infty \frac{1}{1+ \lambda \rho^\perp_p},\,\,
\sigma^2_\lambda=\sum_{p=1}^\infty \frac{1}{(1+\lambda\rho^\perp_p)^2}$.
\end{theorem}
We notice that $h\asymp n^{-c}$ with $\frac{1}{2m+d}\le c\le \frac{1}{2}$ satisfies the rate conditions in Theorem \ref{thm1},
so the asymptotic distribution (\ref{lim:dist:globalLRT}) holds under a wide-ranging choice of $h$.
The quantities $\theta_\lambda$ and $\sigma_\lambda$ solely depend on the eigenvalues $\rho^\perp_p$'s and $\lambda$. 
Based on (\ref{lim:dist:globalLRT}), we propose the following decision rule $\pd$ at the significance level $\alpha$: 
\begin{equation}\label{decisionrule}
\pd(\alpha) = \mathbbm{1}(|2n\cdot PLR_{n,\lambda} - \theta_\lambda|\geq z_{1-\alpha/2}\sqrt{2}\sigma_\lambda)
\end{equation}
where $\mathbbm{1}( \cdot )$ is the indicator function,  $z_{1-\alpha/2}$ is the $1-\alpha/2$ quantile of the standard normal distribution.
Hence, we reject $H_0$ at the significance level $\alpha$ if $\pd=1$.
Similar Wilks' phenomenon is also observed  in the  nonparametric/semiparametric regression framework
\citep{fan2001generalized,shang2013local}. Specifically, let $r_\lambda=\frac{\theta_\lambda}{\sigma^2_\lambda}$,
then (\ref{lim:dist:globalLRT}) implies that, as $n\to\infty$,
\[
\frac{2nr_\lambda\cdot PLR_{n,\lambda}-r_\lambda\theta_\lambda}{\sqrt{2r_\lambda\theta_\lambda}}\overset{d}{\longrightarrow}N(0,1).
\]
Therefore, $2nr_\lambda\cdot PLR_{n,\lambda}$ is asymptotically distributed as a $\chi^2$ distribution with degrees of freedom $r_\lambda\theta_\lambda$. 
In practice, $\rho^\perp_p$'s can be estimated by the sample eigenvalues of the empirical kernel matrix,
from which the quantities $r_\lambda$ and $\theta_\lambda$ can be accurately approximated.
Our numerical study in Sections \ref{sec:simulation} and \ref{sec:realdata} adopt such an approximation 
and the performance is satisfactory.

\vspace{-10pt}
\subsection{Power analysis and minimaxity}

In this section, we investigate the power of PLR under local alternatives.
Define the separation rate as 
\begin{equation}\label{eq:dn}
d_{n}:=\sqrt{\lambda + \sigma_\lambda/n }.
\end{equation}
The separation rate is used to measure the distance between the null and the alternative hypotheses. Theorem \ref{thm2} shows that 
the power of PLR approaches one, provided that the norm of  $\eta^\ast_{XZ}$, the interaction term in the probabilistic decomposition of $\eta^\ast$, has a norm  bounded below by $d_{n}$. The squared separation rate $d_n^2$ consists of two components: $\lambda$ representing the squared bias of the estimator, and $\sigma_\lambda/n$ with the order of $n^{-1}h^{-1/2}$ representing the standard derivation of $PLR_{n,\lambda}$. Since $\sigma_\lambda$ is decreasing with $\lambda$, 
the minimal separation rate for the PLR test is achieved by choosing appropriate $\lambda$ such that $\lambda \asymp \sigma_\lambda/n$.  
Our result owes much to the analytic expression of independence (in terms of interactions) based on the proposed probabilistic tensor product decomposition framework.

Let $P_{\eta^\ast}$ denote the probability measure induced under $\eta^\ast$,
$\|\eta\|_{\sup}$ the supremum norm over $\cY$, and $\|\eta\|_2=\sqrt{V(\eta)}$.

\begin{theorem}\label{thm2} 
Suppose Assumption \ref{a2} holds and let $d_n$ be the separation rate defined in (\ref{eq:dn}), $m>3/2$, 
$\eta^\ast\in\cH$ with 
$\|\eta^\ast_{XZ}\|_{\sup}=o(1)$, $J(\eta^\ast_{XZ})<\infty$, 
$\|\eta^\ast_{XZ}\|_2 \gtrsim d_{n}$. 
For any $\varepsilon\in(0,1)$, there exists a positive $N_\varepsilon$ such that,
for any $n\ge N_\varepsilon$, $\mathbb{P}_{\eta^\ast}(\pd(\alpha)=1)\ge 1-\varepsilon$.
When $\lambda \asymp \lambda^{\ast}\equiv n^{-4m/(4m+d)}$,
$d_{n}$ is upper bounded by $d_n^\ast \equiv n^{-2m/(4m+d)}$.
\end{theorem}
Theorem \ref{thm2} demonstrates that, when $\lambda \asymp \lambda^{\ast}$, PLR can successfully detect any local alternatives, 
provided that they separate from the null at least by $d_n^\ast$.
In Section \ref{sec:minimax}, we show that this upper bound is unimprovable by establishing the minimax lower bound of distinguishable rate for general two-sample test.  It means that no test can successfully detect the local alternatives if they separate from the null by a rate faster than $d_n^\ast$. We claim that our PLR test is minimax optimal.

For any $\varepsilon\in(0,1)$ and $\alpha\in(0,\varepsilon)$,
Theorem \ref{thm1} shows that $\bE_{H_0}\{\Phi_{n,\lambda^\ast}(\alpha)\}$ tends to $\alpha$;
Theorem \ref{thm2} shows that $\bE_{\eta^\ast}\{1-\Phi_{n,\lambda^\ast}(\alpha)\}\le\varepsilon-\alpha$,
provided that $\|\eta^\ast_{XZ}\|_2\ge C_{\varepsilon-\alpha} d_n^\ast$ for a large constant $C_{\varepsilon-\alpha}$. 
That means, asymptotically, 
\begin{equation}\label{err:PLR}
\err(\Phi_{n,\lambda^\ast}(\alpha),C_{\varepsilon-\alpha} d_n^\ast)\le\varepsilon.
\end{equation}
In other words, the total error of PLR is controlled by an arbitrary $\varepsilon$ provided that the null and local alternatives are separated by $d_n^\ast$.

\vspace{-10pt}
\subsection{Data-adaptive tuning parameter selection}
Smoothing parameter selection plays an important role in nonparametric estimation. Classical methods including generalized cross-validation (GCV) \citep{craven1976smoothing} and restricted maximum likelihood (REML) \citep{wood2011fast} provide data-adaptive estimate of the smoothing parameter. However, how to select the smoothing parameter in the non-parametric inference is still an open question. Here, we introduce a data-splitting method to select the smoothing parameter in our proposed PLR test. Here we use the first half of the data to select the smoothing parameter and the second half of the data to calculate the test statistics.
Since we divide the data into two independent parts, the selecting event in the first half is independent with the PLR statistics calculated in the second half. When the sample size is small, data-splitting strategy will suffer from the power reduction. An interesting area for further work would be to establish the post-regularization inference which is recently studied in linear models like \citet{lee2016exact}. 

In practice, how to choose the tuning parameter $\lambda$ is essential to achieve the high power of the proposed test. 
 Theorem \ref{thm2} provides a theoretical guidance that optimal testing rate can be achieved by choosing $\lambda^*$ to minimize the separation rate $d_n$ defined in (\ref{eq:dn}), i.e., satisfying the trade-off between the squared bias of the estimator and the standard deviation of the test statistic. Since $d_n$ is related to the spectral of the population kernel which is usually unknown, we define a sample estimate of $d_n$ by plugging in the empirical eigenvalue of kernel matrix as
\begin{equation*}\label{eq:dnhat}
    \hat{d}_n := \sqrt{\lambda + \hat{\sigma}_{\lambda}} 
\end{equation*}
where $\hat{\sigma}^2_\lambda=\sum_{p=1}^n \frac{1}{(1+\lambda\hat{\rho}^\perp_p)^2}$ and $\hat{\rho}_p,\, p=1,\dots,n$, is the empirical eigenvalue of the kernel matrix $K^{11}$ with the $ij$th entry $K^{11}((x_i,z_i), (x_j, z_j))$. Since $\hat{\sigma}_\lambda$ is a decreasing function of $\lambda$, minimizing $\hat{d}_n$ with respect to $\lambda$ is written as
\begin{equation}\label{eq:adaptiverule}
\hat{\lambda}^* = \max\big\{ \lambda \; | \; \lambda < \hat{\sigma}_{\lambda}/n \big\}.
\end{equation}
We call (\ref{eq:adaptiverule}) as our data-adaptive criterion for choosing $\lambda$.
Notice that $\hat{\lambda}^\ast$ depends on the eigenvalues of the kernel matrix, especially
the first few leading eigenvalues. When the sample size is large, we can approximate $\hat{\sigma}^2_\lambda$ via the top eigenvalues, see \cite{drineas2005nystrom} for fast computation of the leading eigenvalues.

\vspace{-10pt}
\section{Minimax Lower Bound of the  Distinguishable Rate}\label{sec:minimax}

For any $\varepsilon\in(0,1)$, define the minimax separation rate $d_n^\dag(\varepsilon)$ as
\begin{equation}\label{eq:minimaxdn}
d_n^\dag(\varepsilon) = \inf\{d_n>0: \inf_{\Phi} \err(\Phi, d_n)\leq \varepsilon\},
\end{equation}
where the infimum in (\ref{eq:minimaxdn}) is taken over all 0-1 valued testing rules based on samples $\bY_i$'s. 
Here we consider the local alternative by assuming $\|\eta\|_{\cH}<1/2$.
And $d_n^\dag(\varepsilon)$ characterizes the smallest separation between the null and local alternatives such that 
there exists a testing approach with a total error of at most $\varepsilon$.   
Next we establish a lower bound for $d_n^{\dag}$, i.e., if $d_n$ is smaller than a certain lower bound, there exists no test that can distinguish the alternative from the null.

We first introduce a geometric interpretation of the hypothesis testing (\ref{eq:newtwosampletest}). %
Geometrically, $\mathcal{E} = \{\eta\in\cH: \|\eta\|_{\cH}<1/2\}$ is an ellipse with axis lengths equal to eigenvalues of $\cH$ as shown in Figure \ref{fig:geo}. For any $\eta\in\mathcal{E}$, the projection of $\eta$ on $\mathcal{E}_{11} :=\cH_{11} \cap \mathcal{E}$ is $\eta_{XZ}$ where $\cH_{11}$ is defined in (\ref{eq:anovaH}). The magnitude of the interaction $\eta_{XZ}$ can be qualified by $\|\eta_{XZ}\|_2$. The distinguishable rate $d_n$ is the radius of the sphere centered at $\eta_{XZ}=0$ in $\cE_{11}$. 

\begin{figure}[ht!]
\centering
\includegraphics[width=0.5\textwidth]{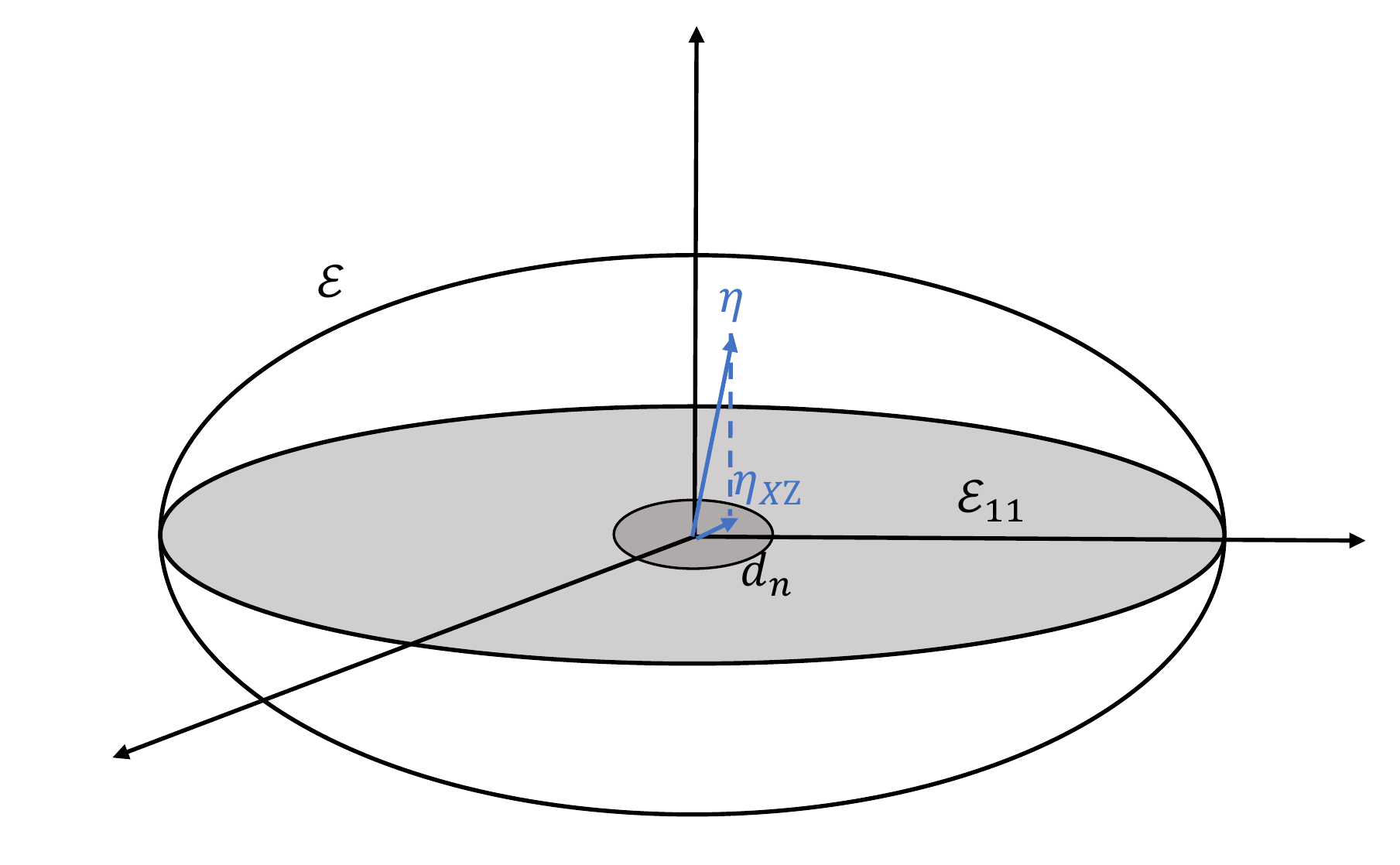}
\caption{\it \footnotesize Geometric interpretation of the distinguishable rate of the testing $H_0$.}\label{fig:geo}
\end{figure}

Intuitively, the testing will be harder when the projection of $\eta$ on $\cH_{11}$ is closer to the original point $\eta_{XZ}=0$.  We then introduce the Bernstein width in \cite{pinkus2012n} to characterize the testing difficulty. For a compact set $C$, the Bernstein $k$-width is defined as
\begin{equation}
b_{k,2}(C) := \argmax_{r\geq 0} \{\mathbb{B}_{2}^{k+1}(r) \subset C\cap S \mbox{ for some subspace } S\in S_{k+1} \}
\end{equation}
where $S_{k+1}$ denotes the set of  all $k+1$ dimensional subspace, $\mathbb{B}_{2}^{k+1}(r)$ is a $(k+1)$-dimensional $L_2$-ball with radius $r$ centered at $\eta_{XZ}= 0$ in $\cH_{11}$. Based on the Bernstein width, we give an upper bound of the testing radius, i.e., for any $\eta$ projected in the ball with radius less than the certain bound, the total error is larger than $1/2$. 

\begin{lemma}\label{lemma:errbound} For any $\eta\in \cH$, we have 
\begin{equation*}
 \err(\Phi,  d_n) \geq 1/2 
\end{equation*}
 for all 
 \begin{equation*}
 \quad d_n \ll r_B(\delta^*): = \sup\{\delta \, |\, \delta \leq \frac{1}{2\sqrt{n}} (k_B(\delta))^{1/4} \}
 \end{equation*}
where $k_B(\delta):= \argmax_{k}\{b^2_{k-1,2}(\cH_{11}) \geq \delta^2\}$ is the Bernstein lower critical dimension and $r_B(\delta^*)$ is called the Bernstein lower critical radius.
\end{lemma}
In Lemma \ref{lemma:errbound}, we show that when $d_n$ is less than $r_B(\delta^*)$, there is no test that can distinguish the alternative from the null. In order to achieve a non-trivial power, we need $d_n$ to be larger than the Bernstein lower critical radius $r_B(\delta^*)$.  The critical radius $r_B(\delta^*)$ depends on the shape of the space $\cH_{11}$.  
The lower bound of $k_B(\delta)$ depends on the decay rate of the eigenvalues for $\cH_{11}$. According to the Liebig's law, the radius of $k$-dimensional ball that can be embedded into $\cH_{11}$ is determined by $k$th largest eigenvalue.  In Lemma \ref{lemma:kb}, we characterize the lower bound of $k_B(\delta)$ by the largest $k$ such that the $k$th largest eigenvalue is larger than $\delta^2$.

\begin{lemma}\label{lemma:kb}
Let  $\gamma_k$ be the $k$th largest eigenvalue of $\cH_{11}$. Then we have 
\begin{equation} 
 k_B(\delta) > \argmax_k \{ \sqrt{\gamma_k} \geq \delta \}
\end{equation} 
\end{lemma} 

Note that $\gamma_k \asymp k^{-2m/d}$, then $ \argmax_k \{ \sqrt{\gamma_k} \geq \delta \} \asymp \delta^{-d/m}$.  Plug in the lower bound of $k_B(\delta)$ to Lemma \ref{lemma:errbound}, we achieve $r_B(\delta^*)$, which is  the minimax lower bound for the distinguishable rate in the following theorem. 

\begin{theorem}\label{thm3}
Suppose $\eta\in \cH$. For any $\varepsilon\in (0, 1)$, the minimax distinguishable rate for the testing hypotheses (\ref{eq:newtwosampletest})  is 
$d_n^\dag(\varepsilon)  \gtrsim n^{-2m/(4m+d)}$.
\end{theorem}
Theorem \ref{thm3} provides a general guidance to justify a local minimax test for testing $\eta_{XZ}=0$. 
The proof of Theorem \ref{thm3} is presented in  Appendix. Comparing $d_n^\dag(\varepsilon)$ with $d_n^*$ derived in Theorem \ref{thm2}, we proved that the PLR test is minimax.

\vspace{-10pt}
\section{Connection to Maximum Mean Discrepancy }\label{sec:mmd}

We first briefly summarize the maximum mean discrepancy (MMD) proposed in \cite{gretton2012kernel}. Given the kernel function $K^{\angular{X}}$ on $\cH^{\angular{X}}$, denote the embedding that maps a probability distribution $f_{X|Z=z}$ into $\cH^{\angular{X}}$ by $\mu_z (\cdot) = \int_\cX K^{\angular{X}}(x,\cdot)f_{X|Z=z}(x)dx$, then the squared MMD between $f_{X|Z=0}$ and $f_{X|Z=1}$ is defined as the squared distance between embeddings of distributions to reproducing kernel Hilbert spaces (RKHS): 
\begin{multline}\label{eq:mmd_est}
 \textrm{MMD}^2_b(\cH^{\angular{X}};f_{X|Z=0},f_{X|Z=1}) 
=  \frac{1}{n_0^2} \sum_{\{i,j\,|\, Z_i= Z_j=0 \}} \hspace{-10pt} K_1^{\angular{X}}(X_i,X_j) \\
- \frac{2}{n_0 n_1} \sum_{\{i,j\,|\, Z_i\ne Z_j \}} \hspace{-10pt} K_1^{\angular{X}}(X_i,X_j) + \frac{1}{n_1^2}\sum_{\{i,j\,|\, Z_i= Z_j=1 \}} \hspace{-10pt} K_1^{\angular{X}}(X_i,X_j) ); 
\end{multline}
where $K^{\angular{X}}_1(X_i ,X_j )$ is introduced in Lemma \ref{lemma:continuouskernel}. 

We next show that the MMD estimate is equivalent to the squared score function based on the likelihood functional without penalty. 
Let $\ell_n$ be the negative likelihood functional defined as $\ell_n (\eta) = -\frac{1}{n} \sum_{i=1}^n \eta(\bY_i)$, 
and $LR_n$ be the likelihood ratio functional defined as 
\begin{multline}\label{eq:lr}
LR_n(\eta) = \ell_{n}(\eta) - \ell_{n}(P_{\cH_0}\eta) =\\ -\frac{1}{n}\sum_{i=1}^{n}\{\eta(\bY_i) - P_{\cH_0}\eta(\bY_i)\}= -\frac{1}{n}\sum_{i=1}^n \{ \inner{K^{\cH}_{\bY_i}}{\eta}_{\cH} - \inner{K^{\cH_0}_{\bY_i}}{\eta}_{\cH} \}
\end{multline}
where $P_{\cH_0}$ is the projection operator from $\cH$ to $\cH_0$ and $K^{\cH} = K^{00} + K^{01} + K^{10}+ K^{11}$ is the kernel for $\cH$ and $K^{\cH_0} =  K^{00} + K^{01} + K^{10}$ is the kernel for $\cH_0$. 

Now we calculate the Fr\'{e}chet derivative of the likelihood ratio functional as the score function, i.e., 
\begin{equation*}
D LR_{n}(\eta) \Delta \eta = \inner{\frac{1}{n}\sum_{i=1}^n (K^{\cH}_{\bY_i} - K^{\cH_0}_{\bY_i})}{\Delta\eta}_{\cH} =  \inner{\frac{1}{n}\sum_{i=1}^n K^{11}_{\bY_i}}{\Delta\eta}_{\cH},
\end{equation*}
where $K^{11}$ is the kernel for $\cH_{11}$. We further define a score test statistics as the squared $\|\cdot\|_{\cH}$ norm of the score function as follows 
\begin{equation}\label{eq:score}
S^2_n = \|\frac{1}{n}\sum_{i=1}^n K^{11}_{\bY_i} \|_{\cH}^2  = \frac{1}{n^2}\sum_{i=1}^n\sum_{j=1}^n K^{11}(\bY_i, \bY_j), 
\end{equation}
where the second equality holds by the reproducing property. Recall that by Lemma \ref{lemma:pdk}  the kernel on $\cH^{\angular{Z}}_1$ is 
$K_1^{\angular{Z}}(Z_i, Z_j) =  \mathbbm{1}{\{Z_i=Z_j\}} - \omega_{Z_i}-\omega_{Z_j}+ \sum_{l=1}^2 \omega_l^2 $, and by Lemma \ref{lemma:continuouskernel}, the kernel on $\cH^{\angular{X}}_1$ is $K_1^{\angular{X}}(X_i, X_j) = K(X_i,X_j)- \E_X[K(X,X_j)] -\E_{\tilde{X}}[K(X_i,\tilde{X})] + \E_{X,\tilde{X}} K(X,\tilde{X})$. 
Then we have $K^{11}(\bY_i, \bY_j) = K_1^{\angular{Z}} (Z_i, Z_j) K^{\angular{X}}_{1}(X_i, X_j)$ based on Lemma \ref{lemma1.2}. Let $\omega_0 = n_0/(n_0+n_1)$ and $\omega_1 = n_1/(n_0+n_1)$ where $n_0$ is the number of observations in group 0 and $n_1$ is the number of observations in group 1.  
Thus, the scaled score test statistic is equivalent to the MMD test statistic, i.e.,
\begin{equation}
\frac{4n_0 n_1}{(n_0 + n_1)^2} S^2_n  =\textrm{MMD}^2_b(\cH^{\angular{X}};f_{X|Z=0},f_{X|Z=1})
\end{equation}
under the null hypothesis. When $n_0=n_1$, i.e. the number of observations are equal in two groups, we have $S^2_n = \textrm{MMD}^2_b(\cH^{\angular{X}};f_{X|Z=0},f_{X|Z=1})$.

The minimax optimality of the score test statistics $S_n^2$ based on the likelihood ratio is yet unknown.  In previous Section \ref{sec:plr}, we established the minimax optimality of the PLR test. We further show the difference between the MMD and our proposed PLR statistic. As shown in the proof of Theorem \ref{thm1}, the PLR test statistic has an asymptotic expression
\begin{equation}                            
PLR_{n,\lambda} \sim \| S^0_{n, \lambda}(\eta) - S_{n, \lambda}(\eta)\|^2 \sim \frac{1}{n} \|\sum_{i=1}^n \tilde{K}^1_{\bY_i}\|^2,
\end{equation}                                                                                                  
where $S_{n,\lambda}$ and $S^0_{n,\lambda}$ are the score functions defined in (\ref{eq:Frechet1}) based on the penalized likelihood ratio functional, and $\tilde{K}^1_{\bY_i} = \tilde{K}_{\bY_i} - \tilde{K}^0_{\bY_i}= \sum_{p=1}^\infty \frac{\xi_p^\perp(\bY_i)\xi_p^\perp}{1+\lambda \rho^\perp_p} $. Notice that $\tilde{K}^1$ 
can be viewed as a scaled version of the product kernel $K^{11}$ by replacing the eigenvalues $\{\rho^\perp_p\}$ with $\{1+\lambda \rho^\perp_p\}$. By choosing $\lambda= \lambda^\ast$, $\textrm{trace}(\tilde{K}^1)= \sum_{p=1}^\infty \frac{1}{1+\lambda^* \rho^\perp_p} \asymp n^{2/(4m+d)} $ matches the lower bound of $k_B(d_n^\dagger)$ with $d_n^\dagger = n^{-2m/(4m+d)}$  as the minimax lower bound for the distinguishable rate in Lemma \ref{lemma:kb}. 
In contrast, the MMD is based on kernel $K^{11}$ without regularization, thus the optimality of the power performance cannot be guaranteed.

\vspace{-10pt}
\section{Simulation Study}\label{sec:simulation}
In this section, we demonstrate the finite sample performance of the proposed test alongside its competitors
through a simulation study. We choose the KS test and Anderson-Darling (AD) as two representers  of the most popular CDF-based tests, the normalized MMD test \citep{li2019optimality} as a representer of kernel-based tests, the ELT \citep{cao2006empirical} as a representer of  density-based tests,  and the dynamic slicing test (DSLICE) \citep{jiang2015nonparametric} as a representer of discretization-based tests. We use the function \textit{ad.test()} provided in the \textit{kSamples} R package for the AD test, conduct the MMD test using the \textit{dHSIC} R package with the default Gaussian kernel, and implement the ELT test using the code provided by the authors. For DSLICE, we follow the authors' suggestions by choosing for each sample size $n$ a penalty parameter so that the rejection region that corresponds to the test statistic being zero is approximate of size $\alpha 100\%$. For our proposed PLR test,  we choose the roughness parameter based on the data-adaptive tuning parameter selection criteria in section 3.4. 
The samples $\bY_i=(X_i,Z_i)$, $i=1,\ldots,n$, were generated as follows.
We first generated $Z_i\overset{iid}{\sim}$ Bernoulli(0.5), with 0/1 representing the control/treatment group. Then $X_i$'s were independently generated from the conditional distribution $f_{X|Z}(x)$
in the following Settings 1 and 2.
In each setting, we chose the averaged sample size $n$ in each group as  125, 250, 375, 500, 625, 750, 875, 1000. 
Size and power were calculated as the proportions of rejection based on $1000$ independent trials.

\noindent Setting 1: we consider the case that  $X$ in each group follows the Gaussian distribution with mean zero and a group-specific variance:
\begin{equation*}
X\mid Z=z\sim N\left(0,(1+\delta_1 \mathbbm{1}_{z=1})^2\right), 
\end{equation*}
where $\delta_1= 0, 0.2, 0.3$. 

\noindent Setting 2: We consider the uni-modal Gaussian distribution versus bi-modal Gaussian distribution: 
\begin{equation*}
X\mid Z=z\sim 0.5 * N\left(-\delta_2 \mathbbm{1}_{z=1} ,(1+\delta_2^2\mathbbm{1}_{z=0})\right) + 0.5 * N\left(\delta_2 \mathbbm{1}_{z=1} ,(1+\delta_2^2\mathbbm{1}_{z=0}) \right)
\end{equation*}
where we set $\delta_2 = 0, 1, 1.2$.

\noindent Setting 3:  $X$ in the two groups follow Gaussian asymmetric mixture distributions, i.e.,
\begin{equation*}
X\mid Z=z\sim 0.5 N(2,1)+0.5 N(-2,(1 - \delta_3\mathbbm{1}_{z=1})^2)
\end{equation*}
where  $\delta_3 = 0, 0.3, 0.45$. 

\noindent Setting 4: we consider symmetric mixtre distributions, i.e., 
\begin{equation*}
X\mid Z=z\sim 0.5 N(2,(1 - \delta_4\mathbbm{1}_{z=1})^2)+0.5 N(-2,(1 - \delta_4\mathbbm{1}_{z=1})^2)
\end{equation*}
where $\delta_4=0, 0.3, 0.6$.

 In particular, $\delta_1=0$, $\delta_2=0$, $\delta_3=0$ or $\delta_4=0$ corresponds to the true $H_0$
which will be used to examine the size of the test statistics. Nonzero $\delta$'s are corresponding to different level of heterogeneity between the two groups.

\begin{figure}[hp]
  \centering
 \begin{tabular}{cc}
    \includegraphics[width=0.4\textwidth]{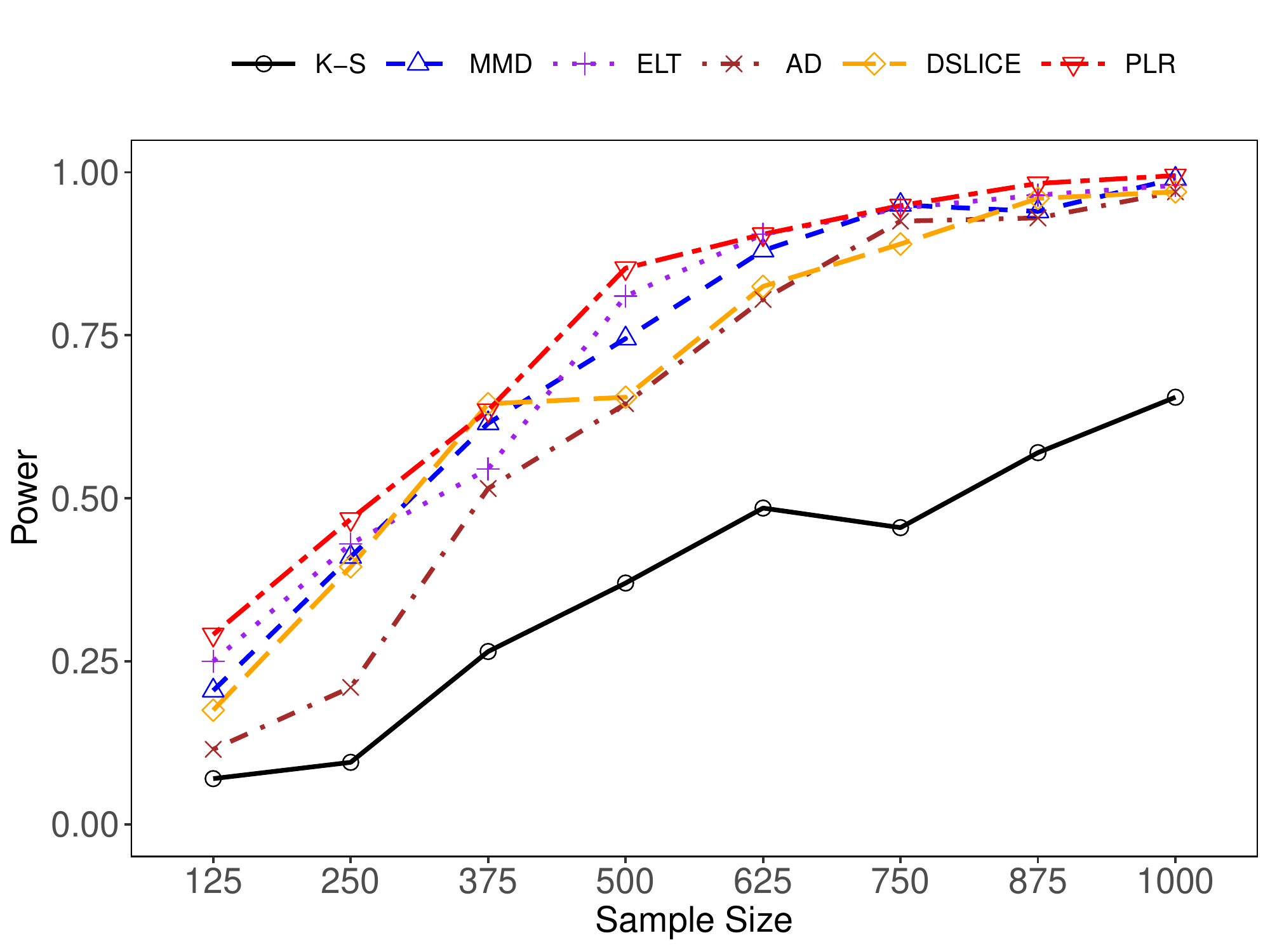}&\includegraphics[width=0.4\textwidth]{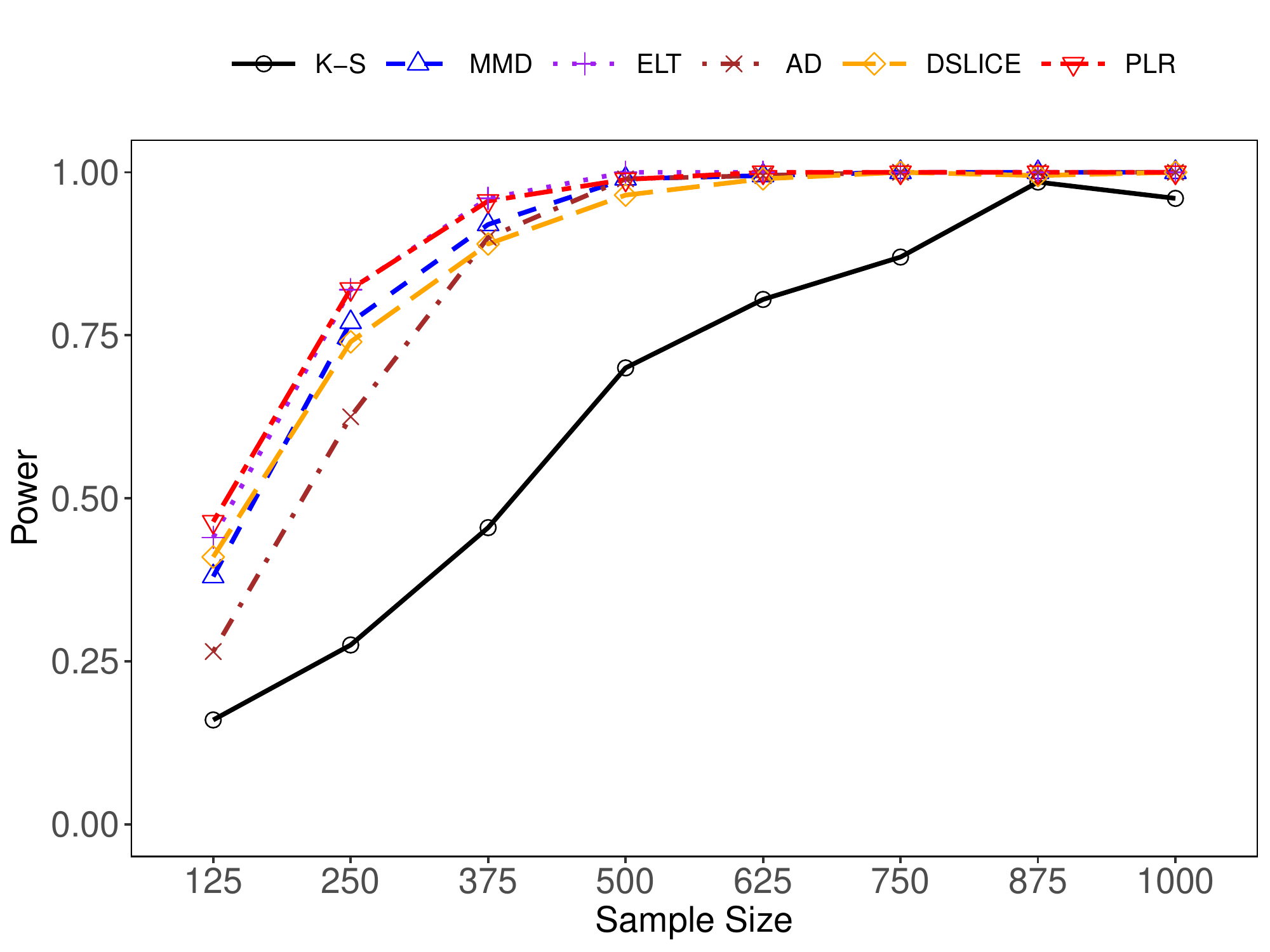} \\
    ~~~~~{\footnotesize (a) Setting 1: $\delta_1=0.2$}&~~~~~ {\footnotesize (b) Setting 1: $\delta_1=0.3$} \\
    \includegraphics[width=0.4\textwidth]{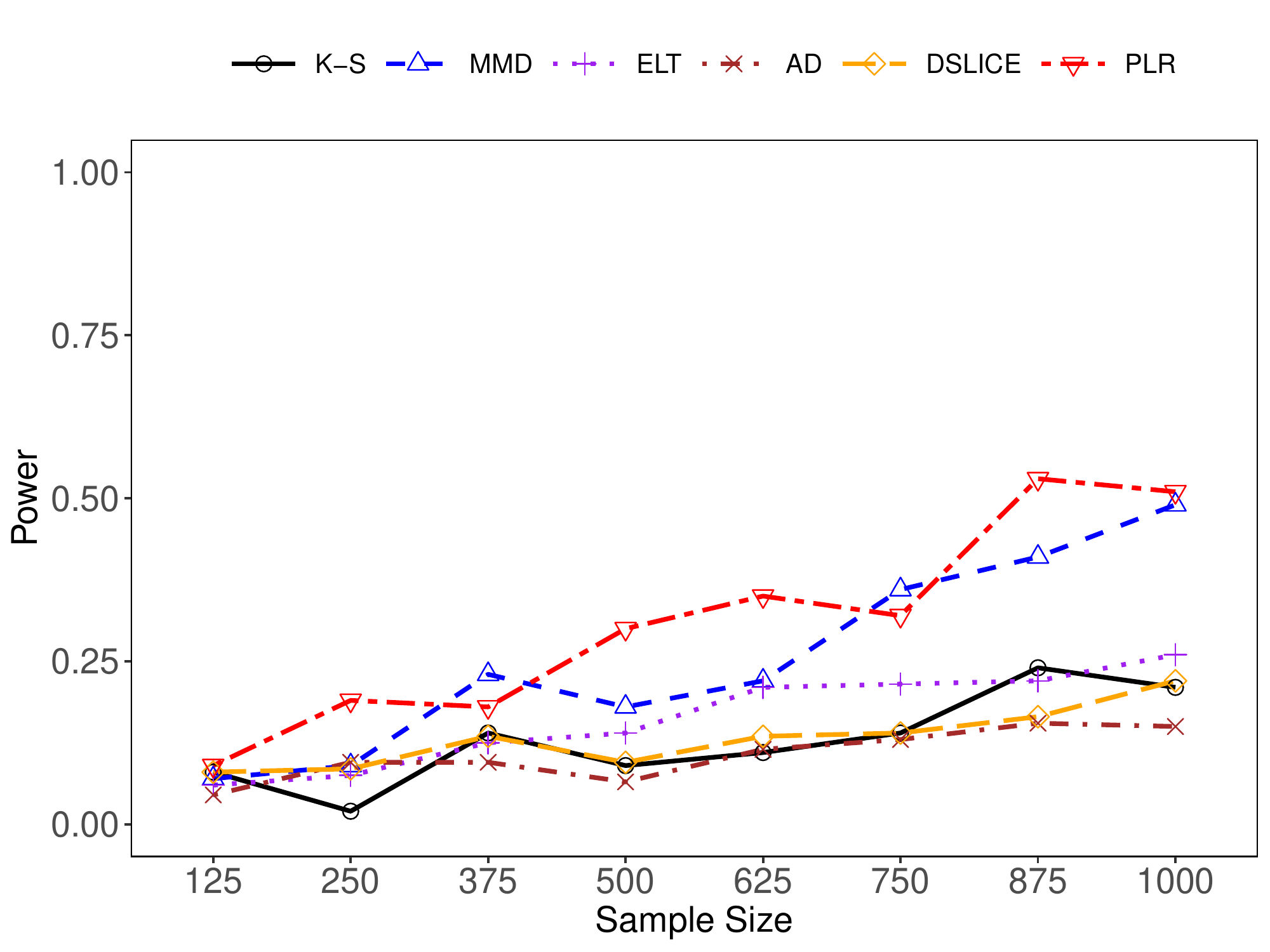}&\includegraphics[width=0.4\textwidth]{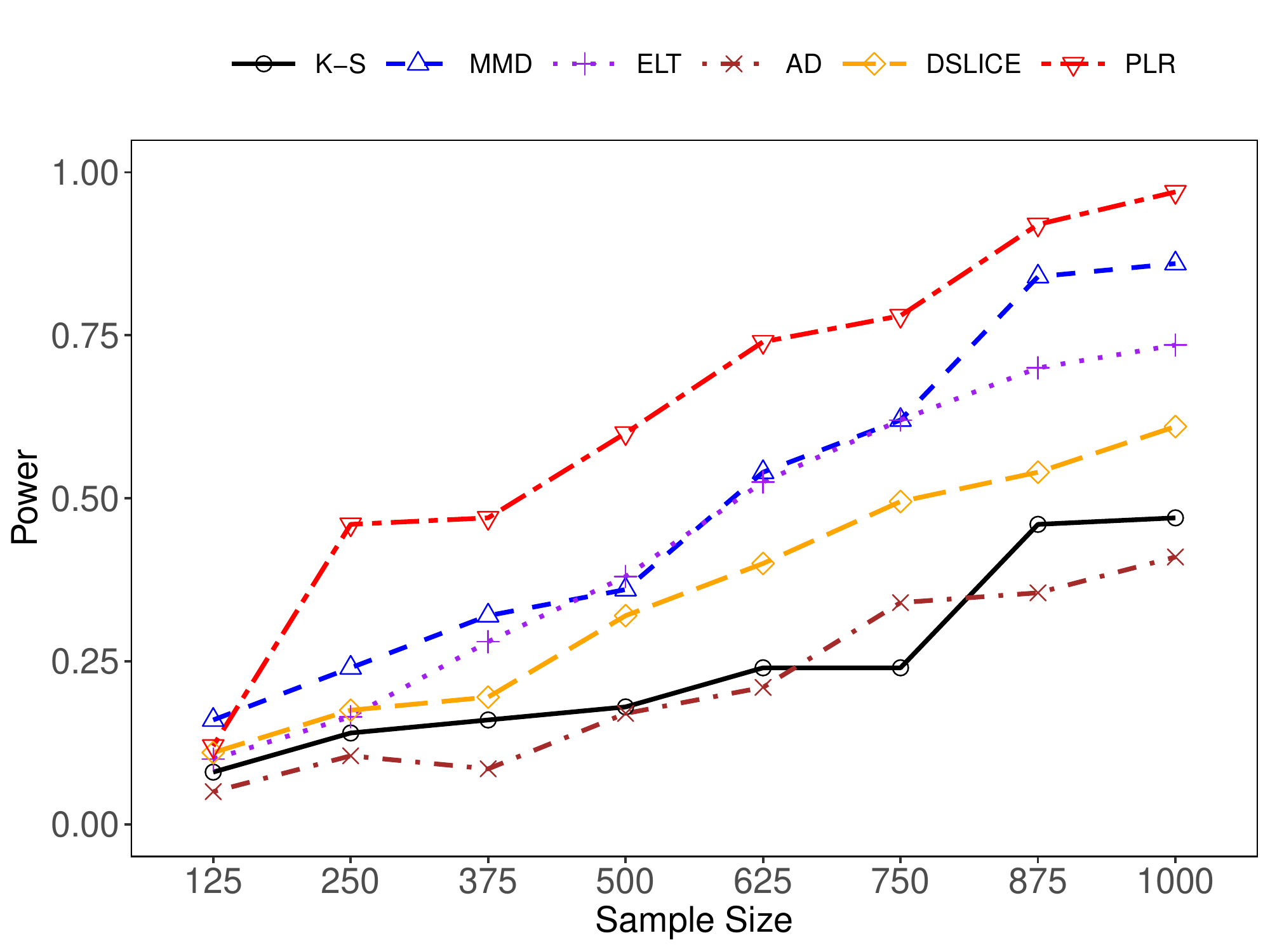} \\
    ~~~~~{\footnotesize (c) Setting 2: $\delta_2=1$}&~~~~~ {\footnotesize (d) Setting 2: $\delta_2=1.2$} \\
        \includegraphics[width=0.4\textwidth]{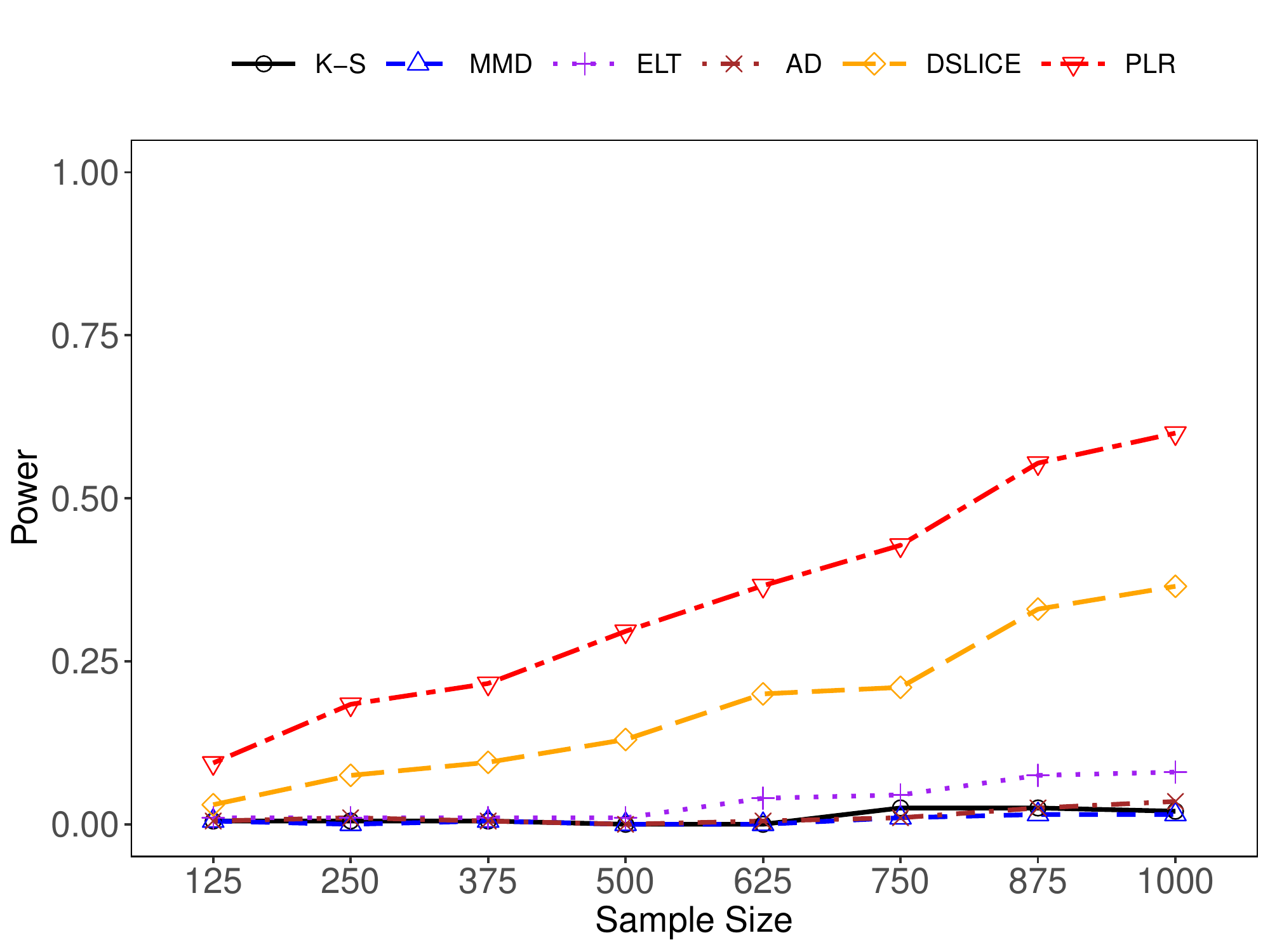}&\includegraphics[width=0.4\textwidth]{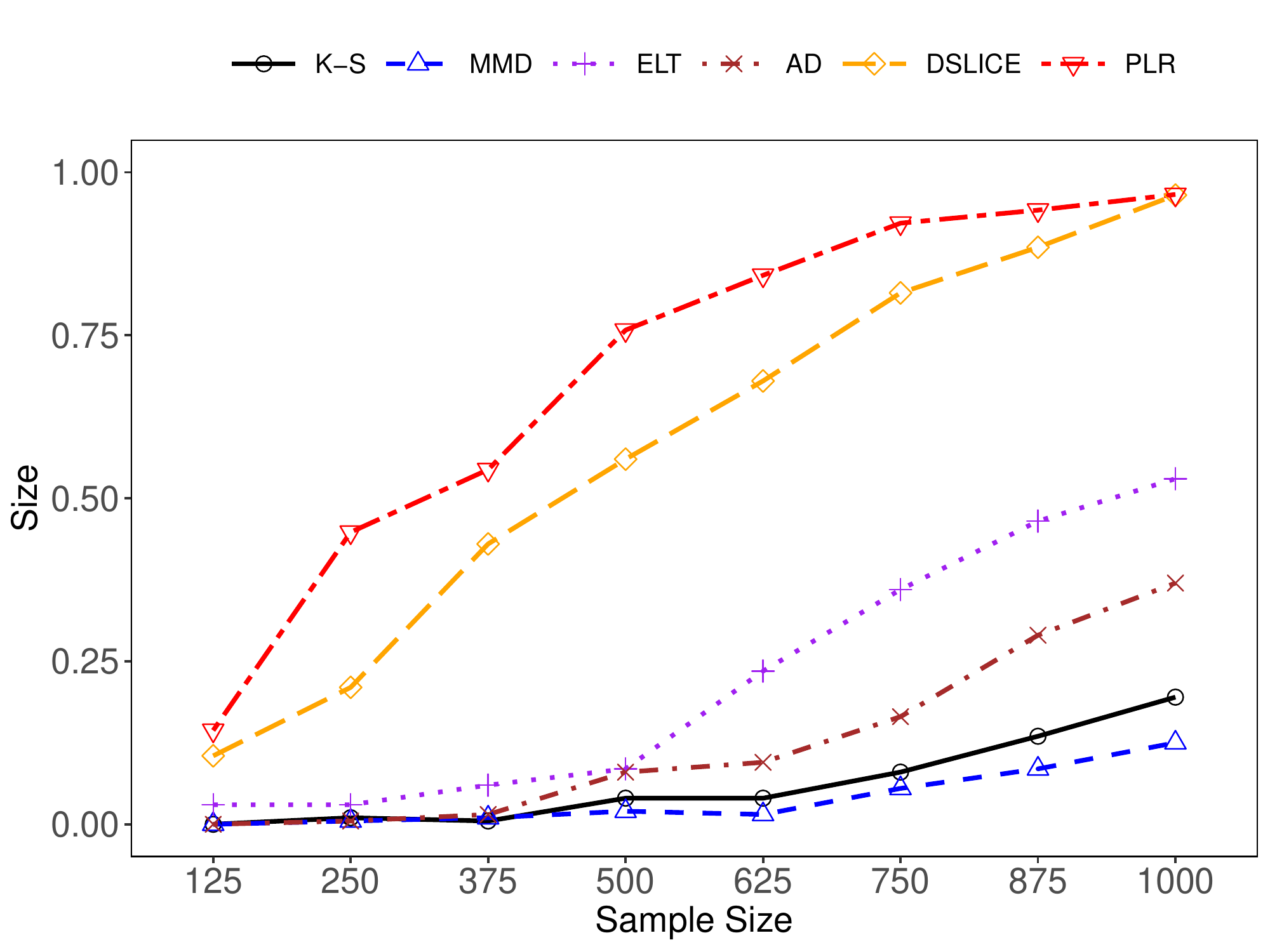} \\
    ~~~~~{\footnotesize (e) Setting 3: $\delta_3=0.3$}&~~~~~ {\footnotesize (f) Setting 3: $\delta_3=0.45$}\\
        \includegraphics[width=0.4\textwidth]{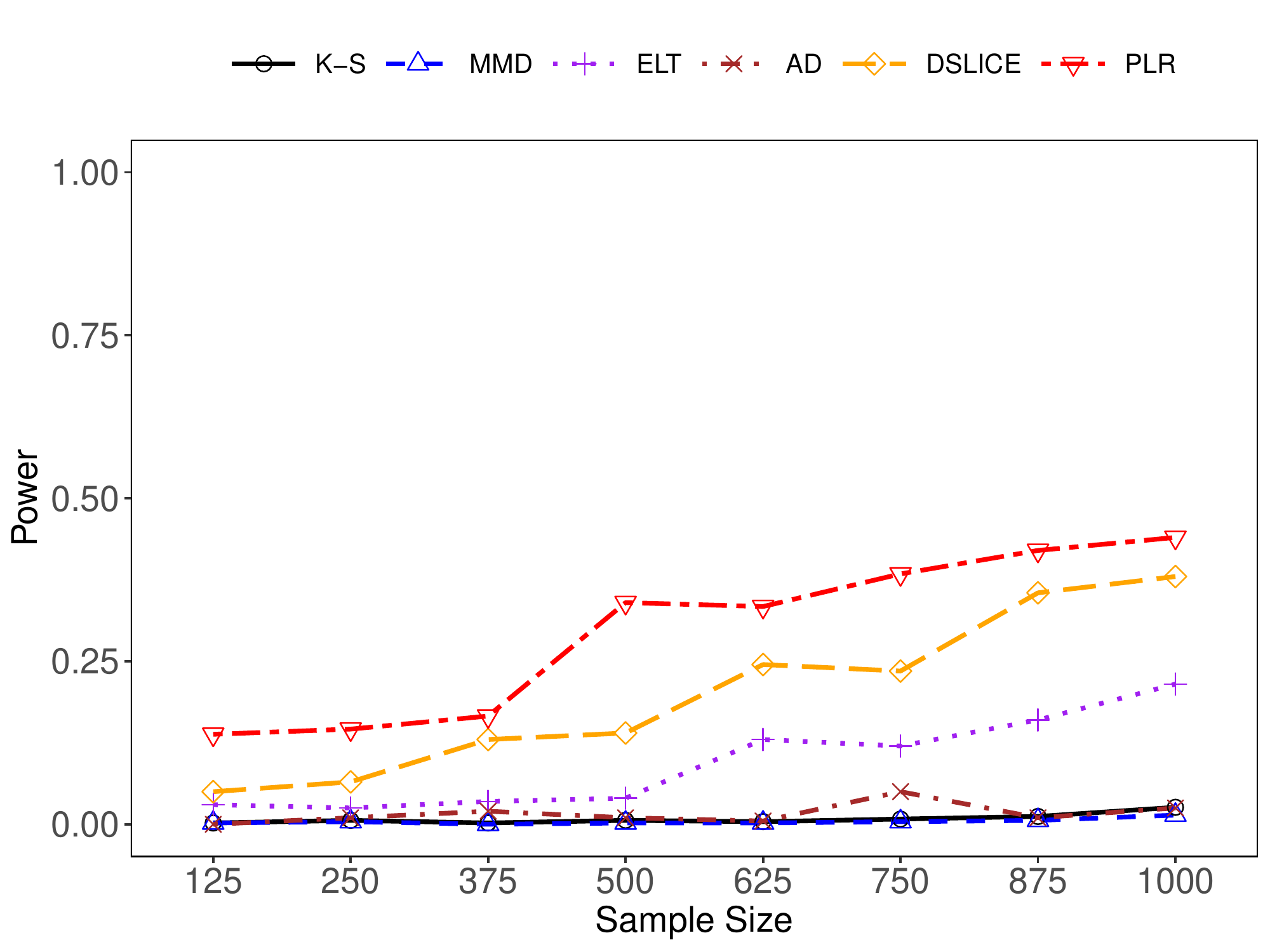}&\includegraphics[width=0.4\textwidth]{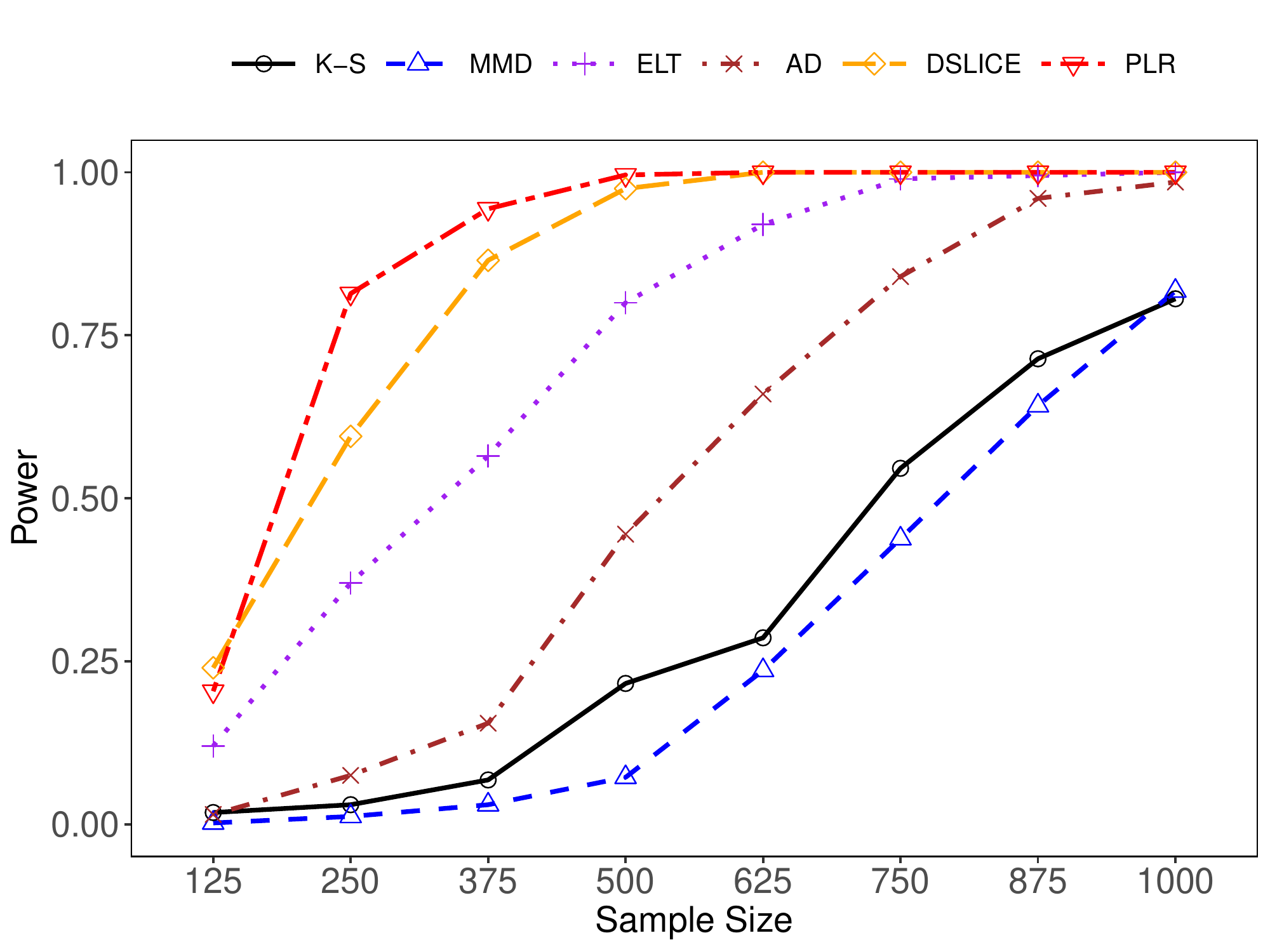} \\
    ~~~~~{\footnotesize (g) Setting 4: $\delta_4=0.3$}&~~~~~ {\footnotesize (h) Setting 4: $\delta_4=0.6$}
\end{tabular}
  \caption{\it\footnotesize Power vs. sample size for PLR, KS, MMD, ELT, AD, and DSLICE. }
  \label{fig:s2}
\end{figure}

\begin{figure}[h!]
  \centering
 \begin{tabular}{cc}
    \includegraphics[width=0.4\textwidth]{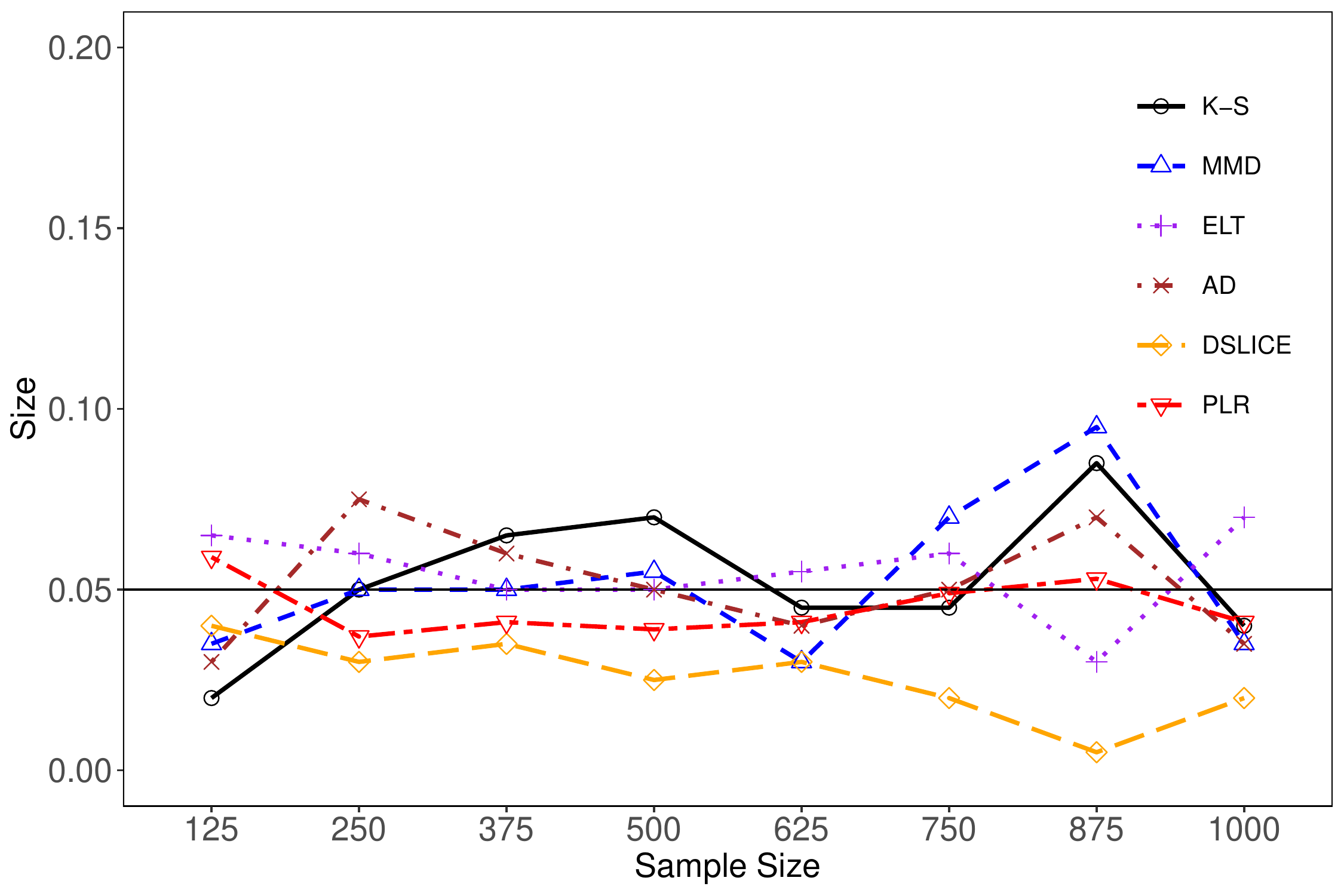}&
    \includegraphics[width=0.4\textwidth]{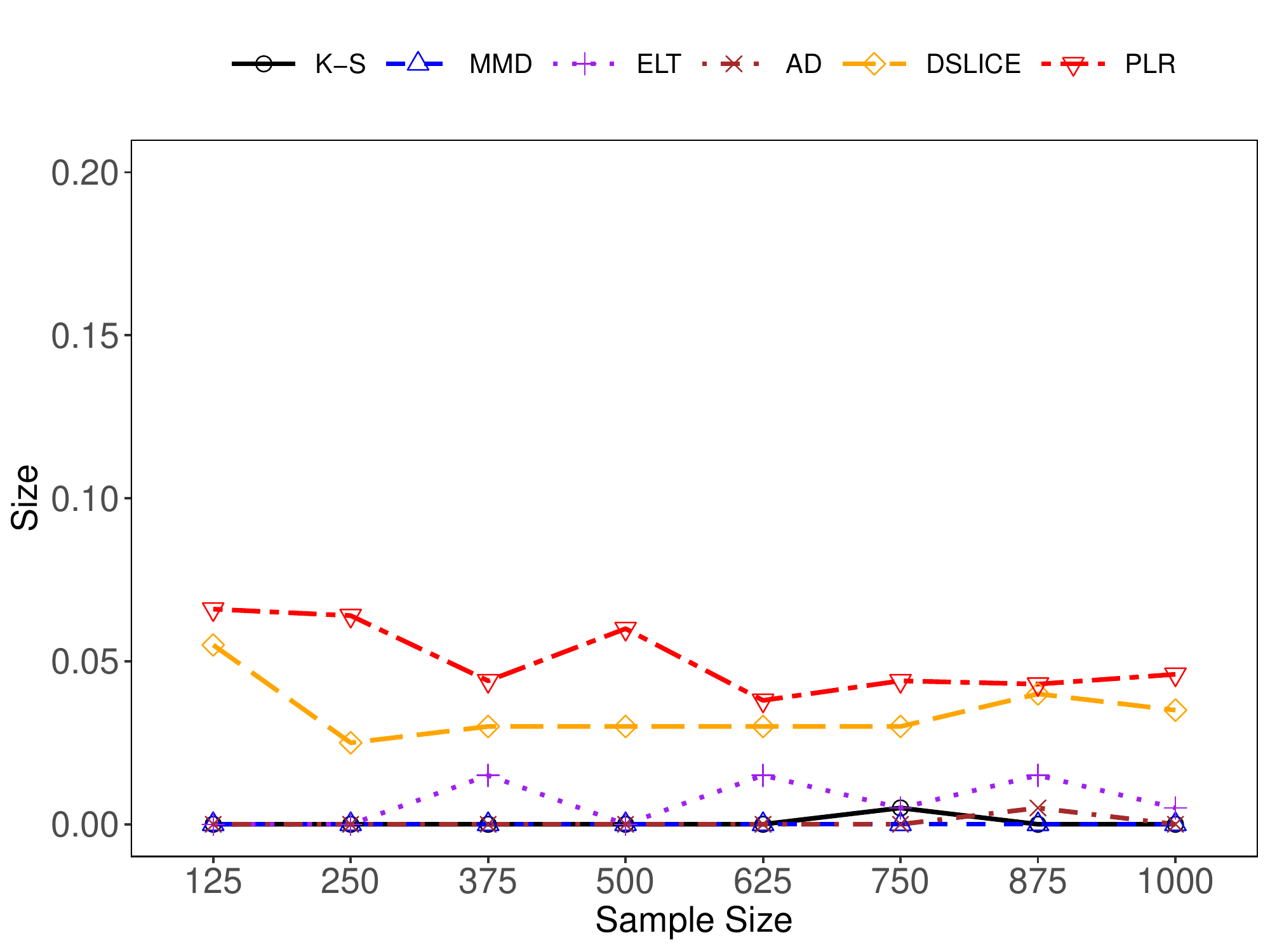}\\
    ~~~~~{\footnotesize $(a)$ Setting 1 \& 2}  &~~~~~ {\footnotesize $(b)$ Setting 3 \& 4}\\
\end{tabular}
  \caption{\it\footnotesize Size vs. sample size for KS, MMD, ELT, AD, DSLICE and PLR tests.}
  \label{fig:s1}
\end{figure}

Figures \ref{fig:s2}  display the powers of the six tests. For Setting 1,  Figure \ref{fig:s2}(a)(b)
show that the powers of the PLR,  MMD, ELT, AD, and DSLICE tests rapidly approach one
when $n$ or $\delta_1$ increases. The power of the KS test increases slightly slower than the other five tests. DSLICE appears to be slightly less powerful than the other four tests, maybe because of its discrete nature and its challenges in choosing a proper penalization parameter in their penalized slicing approach. In Setting 2, as shown 0n \ref{fig:s2}(c)(d), the MMD and PLR test shows comparable power. PLR test has slightly higher power when the heterogeneity is higher. The distinguishable of DSLICE and ELT increases as $\delta_2$ increases.  AD and K-S show significantly lower power.
For Setting 3, Figure \ref{fig:s2}(e)(f) shows again that  the PLR test has the highest power. DSLICE performs quite well here, maybe due to its flexibility in slicing. In contrast, 
the powers of KS, MMD, ELT, AD are significantly lower than both PLR and DSLICE. 
In Setting 4, PLR and DSLICE shows the similar power in Figure \ref{fig:s2}(g)(h). The power of MMD, K-S and AD tests are significantly lower than the others.
The results demonstrate that both  PLR and DSLICE are more adaptive to differently shaped distributions than the other four methods, while PLR enjoys additional advantages than DSLICE when the underlying distribution is smooth.

Figure \ref{fig:s1} displays the size of KS, MMD, ELT, AD, DSLICE and PLR tests.
It can be seen that the sizes of the six tests are around to the nominal level 0.05 in Setting 1\&2, 
confirming that all tests are asymptotically valid.
In Setting 3\&4, the size of the PLR test is still asymptotically correct, and that for DSLICE is reasonably close; while the sizes of
KS, MMD,  and ELT are way below 0.05, showing that these three tests are too conservative in handling bimodal distributions.

\vspace{-10pt}
\section{Real Data Analysis}\label{sec:realdata}
\vspace{-5pt}
In this section, two real-world applications are provided to compare our PLR test with KS and MMD tests.

\textbf{Metagenomic analysis of type II diabetes}: The gut microbiota influences numerous biological functions throughout the body. Recent studies have indicated that gut microbiota plays an important role in many human diseases such as obesity and diabetes. The association between disease and gut microbial composition has been reported in many studies \citep{turnbaugh2009core, qin2012metagenome}. 
Due to the rapid development of metagenomics, it is possible to study the DNA through environmental samples directly. Compared with traditional culture-based methods, metagenomics can study unculturable microorganisms. Recently, several metagenomic binning algorithms such as MetaGen \citep{xing2017metagen} were proposed to estimate the abundance of microbial species with high accuracy. As observed in \cite{turnbaugh2009core}, the microbial distributions demonstrate large cross-individual differences since there are many environmental factors, such as age and antibiotic usage, that could alter the distribution of gut microbiota. 
A powerful test that can detect such distributional differences would be useful in metagenomic analysis.

This study aims to detect whether the microbial species have different distributions between case and control groups. For a particular microbial species, let $X_i$ be the log-transformed abundance for the $i$th individual, and let $Z_i=1/0$ represent the case/control group. 
We applied the proposed PLR test to a metagenomic data set with 145 sequenced gut microbial DNA samples from 71 T2D patients (case group) and 74 individuals unaffected by T2D (control group) using Illumina Genome Analyzer and obtained 378.4 gigabase paired-end reads. We used MetaGen \citep{xing2017metagen} to do the metagenomic binning in which DNA fragments were clustered into species-level bins and estimated the abundance of $2450$ identified species bins.
We applied the KS, MMD, and PLR tests on $1005$ species clusters with abundance larger than $1\%$ of the averaged abundance in more than $50\%$ of the total samples. The 
1005 p-values were calculated by KS, MMD and PLR for each species. We adjusted the p-values by the Benjamini-Hochberg method \citep{benjamini1995controlling}. Through controlling the false discovery at $5\%$, we compared the identified species from the three methods in Figure \ref{real1}(A). The PLR test identified $101$ species, the KS test identified $4$ species, and the MMD test identified $13$ species. The species identified by PLR cover those by KS or MMD. 

\begin{figure}[h!]
    \centering
     \begin{tabular}{cc}
    \includegraphics[width=0.48\textwidth]{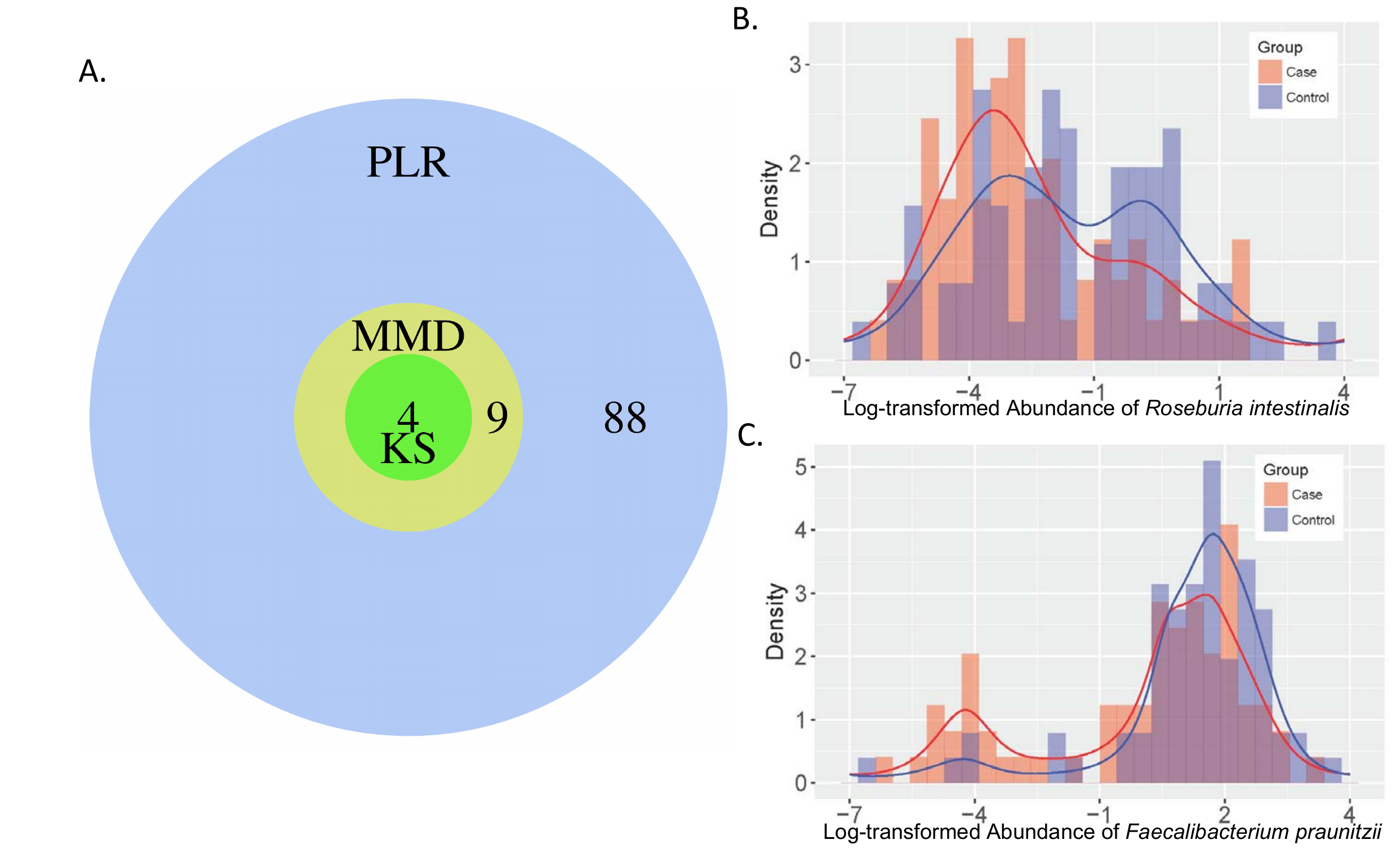}&
    \includegraphics[width=0.48\textwidth]{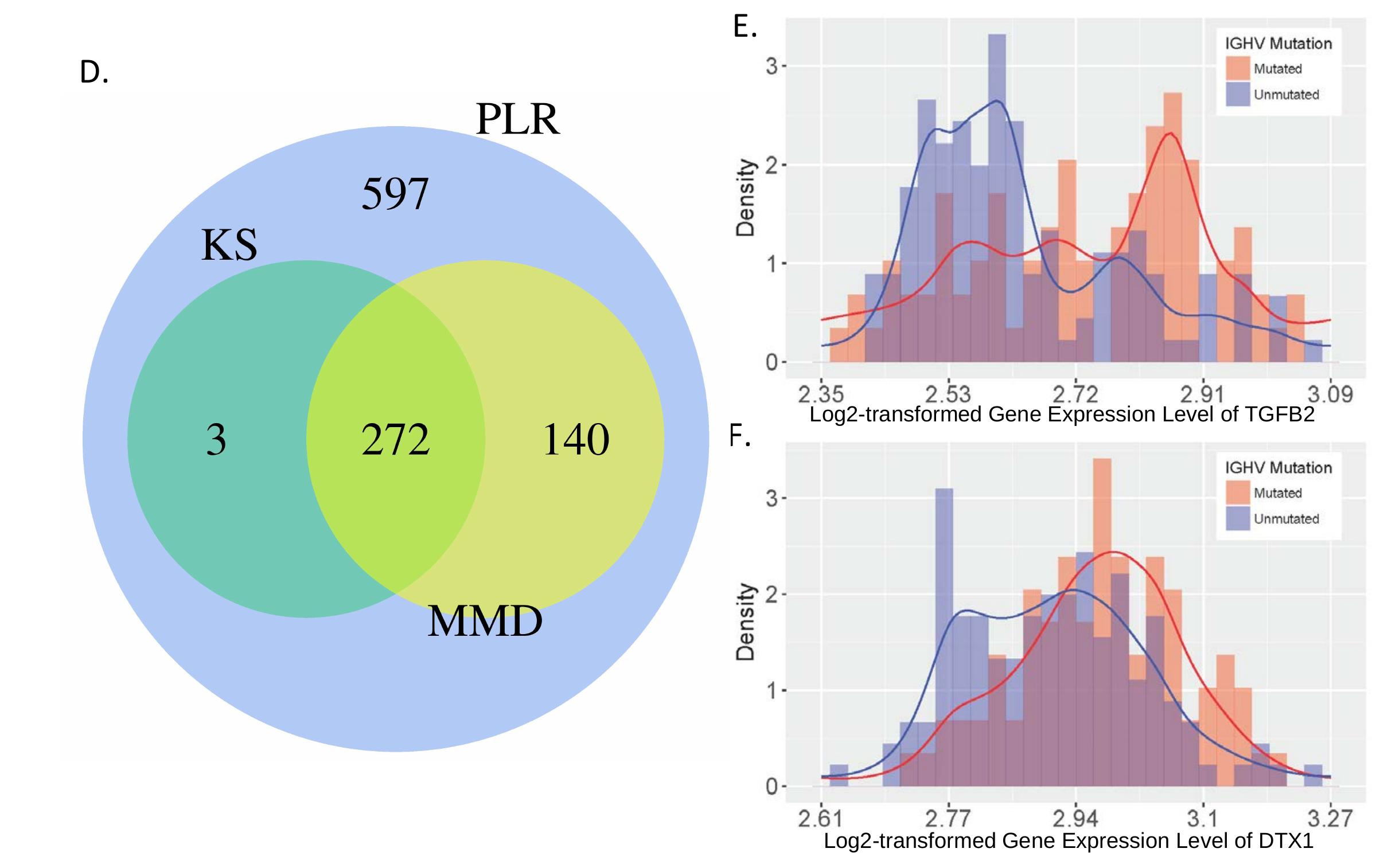}\\
     ~~~~~{\footnotesize  Metagenomic analysis of type II diabetes }&~~~~~ {\footnotesize Gene expression of Chronic Lymphocytic Leukaemia} \\
    \end{tabular}
    \caption{\it\footnotesize (A). A Venn diagram showing the numbers of spiecies identified by PLR, KS and MMD. 
    (B). Densities of log-transformed abundance for Roseburia intestinalis in case/control status. 
    (C). Densities of log-transformed abundance for Faecalibacterium praunitzii in case/control status.
    (D). A Venn diagram showing the numbers of genes selected by PLR, KS and MMD. 
    (E). Densities of gene expression levels from TGFB2 in mutated/unmutated status.
     (F). Densities of gene expression levels from DTX1 in mutated/unmutated status.}\label{real1}
\end{figure}

Moreover, we highlighted two species that were only identified by the PLR test in Figure \ref{real1}(B-C). The densities of these two species are both bimodal in both the case and control groups.  Figure \ref{real1}(B) plots the conditional density of the log-transformed abundance of {\it Roseburia intestinalis}. The majority of the case group has a  significantly low abundance. In Figure \ref{real1}(C), the other species, {\it Faecalibacterium prausnitzii} has lower abundance for a subgroup of patients in the case group. Both species are butyrate-producing bacteria which can exert profound immunometabolic effects, and thus are probiotic less abundant in T2D patients. Our finding is consistent with \cite{tilg2014microbiota} who also observed that the two species' concentrations are lower in T2D subjects. 

\textbf{Gene expression of Chronic Lymphocytic Leukaemia:}
Chronic lymphocytic leukaemia (CLL), the most common leukaemia among adults in Western countries, is a heterogeneous disease with variable clinical presentation and evolution.  
Studies have shown that CLL patients with a mutated Immunoglobulin Heavy  Chain  Variable  (IGHV)  gene  have  a  much more favorable  outcome  and low  probability  of  developing  progressive  disease. In constrast, those with the unmutated  IGHV  gene are much more likely to develop the progressive disease and have shorter survival.  
The molecular changes leading to the pathogenesis of the disease are still poorly understood.  To further investigate the role of the
mutation status in IGHV gene, we aimed to test whether the distributions of the gene expressions are the same between the IGHV mutated and the IGHV unmutated patients. 

This study considered a data set of $225$ CLL patients in which $131$ were IGHV mutated and $85$ were IGHV unmutated. 
 The Affymetrix technique measured the gene expressions in which proper quality control and normalization methods were performed \citep{maura2015association}. We used the Log2-transformed expression value extracted from the CEL files as the measurement of the expression level. 
For the $i$th subject, let $X_i$ denote the expression level and $Z_i$ denote the IGHV mutation status. In particular, $Z_i=0$ denotes the unmutated status, and $Z_i=1$ denotes the mutated status. We aimed to test 
$H_0:f_{X|Z=0}(x) = f_{X|Z=1}(x)$,
i.e. whether the gene expression level's conditional densities are the same between the two IGHV mutation status.
Rejection of $H_0$ implies that the gene expression level distribution varies significantly across the mutation status. 

We applied the PLR, KS, and MMD tests to the $18863$ genes. We performed the Bonferroni correction on the p-values to control the false discovery rate less than $0.05$. 
The three methods selected 1071 genes, 275 genes and 412 genes, respectively. 
Results are summarized in a Venn diagram (Figure \ref{real1}(D))
which demonstrates that the genes selected by PLR cover those selected by KS and MMD. 
There were $272$ genes selected by all methods and 412 genes selected by both PLR and MMD.
For instance, TGFB2 was missed by KS but discovered by PLR and MMD. In literature, it has been verified by real-time quantitative PCR \citep{bomben2007comprehensive}
that TGFB2 is down-regulated in IGHV mutated CLL cases compared with IGHV unmutated cases; see Figure \ref{real1}(E) for a comparison of the conditional densities
from both groups.
So PLR and MMD tends to be more sensititive to select informative genes. There were $597$ genes, including DTX1, uniquely selected by PLR.
DTX1 is a well-established direct target of NOTCH1,  which plays a significant role in a variety of developmental processes as well as in the pathogenesis of certain human cancers and genetic disorders \citep{fabbri2017common}; see Figure~\ref{real1}(F) for a comparison of
the conditional densities, which indicates the expression levels of DTX1 are significantly different between two groups.

\vspace{-10pt}
\section{Discussion}\label{sec:disc}
We proposed a probabilistic decomposition approach for probability densities, and developed the penalized likelihood ratio (PLR) to compare probability densities between groups. 
As demonstrated in simulation studies, our method performs well under various families of  density functions of different modalities.
Notably, our test possesses the Wilks' phenomenon and testing minimaxity. Such results are not easy to derive for distance-based methods. Furthermore, the Wilks' phenomenon leads to an easy-to-execute testing rule that does not involve resampling. An additional natural extension is to test the independence or conditional independence between random variables. This can be carried out through a higher-order probabilistic decomposition of  tensor product RKHS. A challenge for such an extension is to characterize the properties of the eigenvalues of the functional space spanned by interactions. We will explore this direction in future work.

{\footnotesize{\bibliographystyle{apalike}}
\bibliography{ref}}

\newpage
\clearpage
\setcounter{page}{1}
\setcounter{subsection}{0}
\renewcommand{\thesubsection}{S.\arabic{subsection}}
\setcounter{equation}{0}
\renewcommand{\theequation}{S.\arabic{equation}}
\setcounter{theorem}{0}
\renewcommand{\thetheorem}{S.\arabic{theorem}}
\renewcommand{\thelemma}{S.\arabic{lemma}}

\appendix

\begin{center}
  \begin{center}
  {\large\bf Supplementary Material\\
 Minimax Nonparametric Two-Sample Test under Smoothing}%
  \end{center}
\end{center}

Section A contains  proofs of 
the main results in Theorem  \ref{thm1}, \ref{thm2} and \ref{thm3}. Proofs of
Lemma  \ref{sim:diag:V:J}, \ref{rkhs:cH}, \ref{lemma:errbound}, \ref{lemma:kb}, and  Proposition \ref{proposition:a1} as well as some auxiliary results, are also included. 

Section B contains proofs of auxilary Lemmas  S.5-S.9.

Section C contains additional simulation studies on Beta and Beta mixture distribution.

\newtheorem{thm}{Theorem}
\newtheorem{lem}{Lemma}
\newtheorem{pro}{Proposition}

\renewcommand{\thesubsection}{A.\arabic{subsection}}
\renewcommand{\thesubsubsection}{A.\arabic{subsection}.\arabic{subsubsection}}
\renewcommand{\theproposition}{A.\arabic{proposition}}
\renewcommand{\theequation}{A.\arabic{equation}}
\renewcommand{\thelem}{A.\arabic{lem}}
\renewcommand{\thepro}{A.\arabic{pro}}
\renewcommand{\thecorollary}{A.\arabic{corollary}}
\renewcommand{\baselinestretch}{1.2} 
\setlength{\parindent}{15pt}
\section{Proofs of the Main Results}\label{sec:appendix}

\begin{itemize}
\item Section \ref{sec:notation} includes the notation table.
\item Section \ref{sec:proof:auxiliary} includes the proofs of auxiliary lemmas in Section \ref{sec:3}: Lemma \ref{lemma:1}, Lemma \ref{sim:diag:V:J}, Lemma \ref{rkhs:cH}, Proposition \ref{proposition:a1}, and Lemma \ref{lemma:s4}. 
\item Section \ref{sec:proof:test} includes the proof of main Theorems: Theorem \ref{thm1} and Theorem \ref{thm2}.
\item Section \ref{sec:proof:minimax} includes preliminaries for the minimax lower bound, the proof of Lemma \ref{lemma:errbound}, 
Lemma \ref{lemma:kb}, Theorem \ref{thm3}.

\end{itemize}

\subsection{Notation table}\label{sec:notation}
We list the notations in the paper in Table \ref{tab:notation}.

\begin{table}[H]
    \centering
\begin{tabular}{c|p{0.7\textwidth}}
\hline
  $X$ & $d$-dimensional continuous covariate \\
  $Z$ & discrete random variable for the group membership\\
   $Y$ & (X,Z) \\
   $\eta(x,z)$&  log-transformed joint density of $X,Z$\\
  $\mathcal{H}$ & tensor product RKHS \\
  $\langle \cdot, \cdot \rangle_{\mathcal{H}} $, $\|\cdot\|_{\mathcal{H}}$ & the inner product and norm under $\mathcal{H}$ \\
  $K(\cdot, \cdot)$ & kernel function under the norm $\|\cdot \|_{\mathcal{H}}$ \\
  $\cHx= \cHx_0\oplus\cHx_1$ & marginal RKHS of $X$\\
    $\cH^{\angular{Z}}=\cH^{\angular{Z}}_0\oplus\cH^{\angular{Z}}_1$ & marginal RKHS of $Z$\\
    $K^{\angular{X}}_i$ & kernel function for $\cHx_i$, $i=0,1$\\
    $K^{\angular{Z}}_i$ & kernel function for $\cH^{\angular{Z}}_i$, $i=0,1$\\
    $\cH_{ij}$ & RKHS for intercept, main effects, interaction effect\\
    $K^{ij}$ & kernel function for $\cH_{ij}$\\
    $\cA$ & averaging operator\\
  $\{\mu_i,\phi_i\}_{i=0}^\infty$ & eigensystem for $\cHx$  \\
  $\{\nu_i,\psi_i\}_{i=0}^\infty$ & eigensystem for $\cH^{\angular{Z}}$  \\
  $\ell_{n,\lambda}(\eta)$ & negative penalized likelihood function \\
  $\widehat{\eta}^0_{n,\lambda}$ & penalized likelihhod estimator of $\eta$ under $H_0$  \\
  $\widehat{\eta}_{n,\lambda}$ & penalized likelihhod estimator of $\eta$ in $\cH$  \\
  $\langle \cdot, \cdot \rangle $, $\|\cdot \|$ & embedded inner product and norm in $\cH$ \\
  $\langle \cdot, \cdot \rangle_0 $, $\|\cdot \|_0$ & embedded inner product and norm in $\cH_0$ under $H_0$ \\
  $V(\cdot, \cdot)$ & $L_2$ inner product  \\
  $J(\cdot)$ & penalty function\\
  $\tilde{K}(\cdot, \cdot)$ & kernel function equipped with $\|\cdot \|$ in $\cH$ \\
   $\tilde{K}^0(\cdot, \cdot)$ & kernel function equipped with $\|\cdot \|_0$ in $\cH_0$ under $H_0$ \\
  $PLR_{n,\lambda}$ & penalized likelihood ratio test statistic  \\
  $\|\cdot\|_{sup}$ & the supremum norm \\
  $W_\lambda$ &  self-adjoint operator satisfies $\langle W_\lambda \eta,\tilde{\eta}\rangle=\lambda J(\eta, \tilde{\eta})$\\
  $\{\rho_p, \xi_p\}_{p=1}^\infty$ & eigensystem that simultaneously diagonalizes $V$ and $J$ in $\cH$\\
  $\{\rho_p^0, \xi^0_p\}_{p=1}^\infty$ & eigensystem that simultaneously diagonalizes $V$ and $J$ in $\cH_0$\\
  $\{\rho^\perp_p, \xi^\perp_p\}_{p=1}^\infty$ & eigensystem generates the orthogonal complement of $\cH_0$\\
  $D\ell_{n,\lambda}$, $D^2\ell_{n,\lambda}$, $D^3\ell_{n,\lambda}$& first-, second-, third-order Frech\'{e}t derivatives of $\ell_{n,\lambda}(\eta)$\\
  $\pd(\alpha)$& decision rule at the significance level $\alpha$\\
  $d_n^\dag(\varepsilon)$& minimax separation rate \\
  $LR_n(\eta)$ & likelihood ratio function\\
  $\tilde{K}^1(\cdot,\cdot)$ & $\tilde{K}(\cdot,\cdot) - \tilde{K}^0(\cdot,\cdot)$\\
  \hline
\end{tabular}\\
\caption{A table that lists all useful notation and their meanings. }\label{tab:notation}
\end{table}

\subsection{Proofs of Lemmas in Section \ref{sec:3}}\label{sec:proof:auxiliary}

\subsubsection{Some Auxiliary Lemmas}
We first state some auxiliary lemmas in Lemma \ref{lemma:pdk}, Lemma \ref{lemma:continuouskernel}, and Lemma \ref{lemma1.2} to construct kernel functions of the RKHS, which lays the foundation to prove results in Section \ref{sec:3}. 
\begin{lemma}\label{lemma:pdk}
For the RKHS $\cH^{\angular{Z}}$ on the discrete domain $\{0,1\}$ with probability measure $\PP(Z=z)=\omega_z$ for $z=0,1$, there corresponds a unique non-negative definite reproducing kernel $K^{\angular{Z}}$. Based on the tensor sum decomposition $\cH^{\angular{Z}}=\cH^{\angular{Z}}_0\oplus\cH^{\angular{Z}}_1$ where $\cH^{\angular{Z}}_0=\{\E_Z[K^{\angular{Z}}_Z]\}$ and $\cH^{\angular{Z}}_1=\{f\in\cH: \E_Z(f(Z))=0\}$, we have that the kernel for $\cH^{\angular{Z}}_0$ is
\begin{align*}
    K^{\angular{Z}}_0(z,\tilde{z}) = \omega_{z} + \omega_{\tilde{z}} %
\end{align*}
and the kernel for $\cH^{\angular{Z}}_1$ is 
\begin{align*}
    K^{\angular{Z}}_1(z, \tilde{z}) &= \mathbbm{1}_{\{z=\tilde{z}\}}- \omega_{z} - \omega_{\tilde{z}} %
\end{align*}
where $\mathbbm{1}$ is the indicator function.
\end{lemma} 

\begin{lemma}\label{lemma:continuouskernel}
For the RKHS $\cH^{\angular{X}}$ on a continuous domain $\cX$ with probability measure $\PP$ equipped with inner product $\inner{\cdot}{\cdot}_{\cH^{\angular{X}}}$, there corresponds a unique nonnegative definite reproducing kernel $K^{\angular{X}}$. Based on the tensor sum decomposition $\cH^{\angular{X}}=\cH^{\angular{X}}_0\oplus\cH^{\angular{X}}_1$ where $\cH^{\angular{X}}_0=\{\E_X K^{\angular{X}}_X\}$ and $\cH^{\angular{X}}_1=\{f\in\cH: \E_X(f(X))=0\}$, we have that the kernel for $\cH^{\angular{X}}_0$ is
\begin{align}
    K^{\angular{X}}_0(x,\tilde{x}) =  \E_X[K(X,\tilde{x})] + \E_{\tilde{X}}[K(x,\tilde{X})] - \E_{X,\tilde{X}} K(X,\tilde{X}),
\end{align}
and the kernel for $\cH^{\angular{X}}_1$ is 
\begin{multline*}
    K^{\angular{X}}_1(x, \tilde{x}) = \inner{K^{\angular{X}}_x- \bE_X K^{\angular{X}}_X}{K^{\angular{X}}_{\tilde{x}} - \bE_{\tilde{X}} K^{\angular{X}}_{\tilde{X}}}_{\cH^{\angular{X}}}\\ 
    =K^{\angular{X}}(x,\tilde{x}) - \E_{X}[K^{\angular{X}}(X,y)] -\E_{\tilde{X}}[K^{\angular{X}}(x,\tilde{X})] + \E_{X,\tilde{X}} K^{\angular{X}}(X,\tilde{X}).
\end{multline*}
\end{lemma}

\begin{lemma}\label{lemma1.2}
Suppose $K^{\angular{X}}_i$ is the reproducing kernel of $\cH^{\angular{X}}_i$ on $\mathcal{X}$, and $K^{\angular{Z}}_j$ is the reproducing kernel of $\cH^{\angular{Z}}_j$ on $\mathcal{Z}$ for $i=0,1$ and $j=0,1$. Then the reproducing kernels of $\cH^{\angular{X}}_i\otimes \cH^{\angular{Z}}_j$ on $\mathcal{Y}=\mathcal{X}\times \mathcal{Z}$ is $K^{ij}((x,z),(\tilde{x},\tilde{z}))=K^{\angular{X}}_i(x, \tilde{x})K^{\angular{Z}}_j(z, \tilde{z})$ with $x, \tilde{x}\in \cX$ and $z,\tilde{z}\in \cZ$.
\end{lemma}

\subsubsection{The equivalence between the two-sample test and the interaction test}\label{sec:a1}
In the following Proposition \ref{lemma:1}, we show that the two-sample test is equivalent to testing whether the interaction $\eta_{XZ}$ is $0$ or not. 
\begin{proposition}\label{lemma:1}
Let $\eta$ be the log-transformed density function of $(X,Z)$ and $\eta_{XZ}$ be the interaction term defined in (\ref{eq:interaction}), we have $\eta_{XZ}=0$ if and only if $f_{X|Z=0} (\cdot) = f_{X|Z=1}(\cdot)$, where 
$f_{X|Z=z}(x)$ is the conditional density of $X$ given $Z=z$. %
\end{proposition}
 
\begin{proof}
Write the log-transformed joint density as  $\eta(x,z)= \eta_{0} + \eta_{X}(x)+\eta_Z(z)+\eta_{XZ}(x,z)$ according to (\ref{eq:anovaH}). 
if $\eta_{XZ}=0$, then $f(x,z)\propto e^{\eta_X(x)} e^{\eta_Z(z)}$, and hence,
$X,Z$ are independent.

On the other hand, if $X$ and $Z$ are independent, then the joint density $f(x,z) = f_X(x)f_Z(z)$,
where $f_X, f_Z$ are the marginal densities of $X$ and $Z$. Take log-transformations on both sides, i.e.,
$\eta(x,z) = \log(f(x,z)) = \log(f_X(x)) + \log(f_Z(z))$. By the decomposition (\ref{eq:anovaH}), we have $\cA_X \eta_{XZ} = 0$ and $\cA_Z\eta_{XZ}=0$. If  we have $\eta_{XZ}\ne0$, then $f(x,z)$ can not be factorized. Hence, we have $\eta_{XZ}=0$

\end{proof}

\subsubsection{Proof of Lemma \ref{sim:diag:V:J}}\label{sec:a10}

\begin{proof}
We aim to construct the eigensystems on the marginal domain $\cH^{\angular{X}}$ and $\cH^{\angular{Z}}$,
based on which the eigensystem on $\cH$ will be constructed. First, we consider $\cX=[0,1]^d$. 
Recall the Sobolev norm $V_X(g_1,g_2)+J_X(g_1,g_2)$ on $\cH^{\angular{X}}$.
Let $\bbN_0$ denote the set of non-negative integers.
Following \cite{shang2013local}, we choose the eigenvalues and eigenfunctions of $\cH^{\angular{X}}$ as the solution to the following systems of
partial differential equations: for integer $k\in\bbN_0$ and $\alpha_1,\ldots,\alpha_d\in\bbN_0$
satisfying $\alpha_1+\cdots+\alpha_d=m$,
\begin{equation}\label{eq:ode}
(-1)^m \frac{\partial^m}{\partial^{\alpha_1}\cdots\partial^{\alpha_d}}\phi_k(x_1,\ldots,x_d) = \mu_k f_X(x_1,\ldots,x_d) \phi_k(x_1,\ldots,x_d)
\end{equation}
with boundary conditions: for any $l=m,\ldots,2m-1$ and non-negative integers $\beta_1,\ldots,\beta_d$
satisfying $\beta_1+\cdots+\beta_d=l$,
\begin{eqnarray*}
\textrm{ $\frac{\partial^m}{\partial^{\alpha_1}\cdots\partial^{\alpha_d}}\phi(x_1,\ldots,x_d)= 0$
 for $(x_1,\ldots,x_d)\in\partial[0,1]^d$,}
\end{eqnarray*}
where $f_X$ is the marginal density of $X$,
$\partial[0,1]^d$ denotes the boundary of $[0,1]^d$, $\mu_k$'s are non-negative, non-decreasing and normalized so that $V_X(\phi_k,\phi_k)=1$ for any $k\ge0$. 
Simple integration by parts can show that
the solutions to (\ref{eq:ode}) satisfy $V_X(\phi_k,\phi_{k'})=\delta_{kk'}$ and $J_X(\phi_k,\phi_{k'})=\mu_k\delta_{kk'}$.
Meanwhile, the null space has dimension $M= \binom{m+d-1}{d}$, so one has
$0=\mu_0=\mu_1=\cdots=\mu_{M-1}\le\mu_{M}\le\mu_{M+1}\le\cdots$ with $\mu_k\asymp k^{2m/d}$.
Furthermore, one can actually choose $\phi_0\equiv1$.
To see this, note that $\phi_0,\ldots,\phi_{M-1}$ are basis of the null space of monomials on $[0,1]^d$ with orders up to $m-1$. For $0\le k\le M-1$, 
there exists $\bt=(t_1,\ldots,t_d)\in \bbN_0^d$ satisfying $|\bt|\equiv\sum_{l=1}^d t_l<m$ such that one can write $\phi_k(x)\equiv\phi_{\bt}(x)=\sum_{i=1 }^M a_{i,k}x_1^{t_1}\dots x_d^{t_d}$.
For $\bt,\bt^\prime\in\bbN_0^d$ satisfying $0\le |\bt|,|\bt^\prime|<m$, define $M_{\bt\bt^{\prime}}=\int_{[0,1]^d} x_1^{t_1+t_1^{\prime}}\dots x_{d}^{t_d + t_d^\prime}f_X(x)dx$. Let $A_k=(a_{1,k},\ldots,a_{M,k})^T$ and $\textbf{M}=[M_{\bt \bt^{\prime}}]_{|\bt|,|\bt^\prime|=0}^{m-1}$.
Since $V_X(\phi_k,\phi_{k'})=\delta_{kk'}$ for $k,k'=1,\ldots,M$, we have
$A_k^T \textbf{M} A_{k'}=\delta_{kk'}$.
Purposely choose $A_1=(1,0,\ldots,0)^T$ and treat the rest $A_2,\ldots,A_{M}$
as unknowns to be determined.
This leaves us $M^2-M$ unknown coefficients and $\frac{M^2+M}{2}-1$ equations.
Since $M^2-M\ge\frac{M^2+M}{2}-1$ for any positive integer $M$,
there always exist $A_k$'s for $k=2,\ldots,M$ that satisfy $A_k^T M A_{k'}=\delta_{kk'}$.
This shows that we can choose $\phi_0\equiv 1$ while maintaining the simultaneous diagonalization.  

The space $\cH^{\angular{Z}}$ is an $a$-dimensional Euclidean space endowed with Euclidean norm.  
Let $\{\psi_l\}_{l=0}^{a-1}$ denote the orthonormal eigenvectors. The corresponding eigenvalues are $\nu_0=\cdots=\nu_{a-1}=1$.
To see this, note that the reproducing kernel is $R(z,z')=1(z=z')$,
hence, $\langle R_z,\psi_l\rangle_Z=\psi_l(z)$. On the other hand,
$R(z,z')=\sum_{l=0}^{a-1}\nu_l\psi_l(z)\psi_l(z')$,
hence, $\langle R_z,\psi_l\rangle_Z=\psi_l(z)\nu_l$, leading to $\nu_l=1$.
For convenience, we choose $\psi_0$ as constant function, i.e., $\psi_0(z)\equiv1/\sqrt{a}$ for $z=1,\ldots,a$.
 
Let $\|\cdot\|_{\cH^{\angular{X}}\otimes \cH^{\angular{Z}}}$ denote the tensor product norm induced by  
$V_X(g_1,g_2)+J_X(g_1,g_2)$
on $\cH^{\angular{X}}$ and the Euclidean norm on $\cH^{\angular{Z}}$.
The marginal basis for $\cH^{\angular{X}}$ and $\cH^{\angular{Z}}$
naturally provide a basis for the tensor space, i.e., $\{\phi_k\psi_l: k\ge0, 0\le l\le a-1\}$,
that satisfy
\begin{equation}\label{natural:tensor:inner:prod}
\langle \phi_k\psi_l,\phi_{k'}\psi_{l'}\rangle_{\cH^{\angular{X}}\otimes \cH^{\angular{Z}}}=(1+\mu_k\nu_l)\delta_{kk'}\delta_{ll'}.
\end{equation}
The right hand side $\mu_k\nu_l$ of (\ref{natural:tensor:inner:prod})
is the eigenvalue corresponding to basis $\phi_k\psi_l$.
Indeed, they form the eigenvalues of the Rayleigh quotient $\|\cdot\|^2_{L^2(X)\otimes L^2(Z)}/
\|\cdot\|^2_{\cH^{\angular{X}}\otimes \cH^{\angular{Z}}}$ since $\phi_k$ and $\psi_l$
are eigenvalues of the marginal Rayleigh quotients; see \cite[Section 2.3]{lin2000tensor}. We arrange the eigenvalues $\{\mu_k\nu_l\}$ in an increasing order, and denote them as 
$\pi_1\le\pi_2\le\cdots$,
i.e., $\pi_{ra+s}=\mu_r$ for $r\ge0$ and $1\le s\le a$.

Consider the orthogonal decomposition $\cH=\cH_0\oplus\cH_1$ in (\ref{inner}).
By \cite{weinberger1974variational}, we can use the Rayleigh quotient $V/(V+J)$ to produce 
$\xi_p^0\in\cH_0$ and $\xi_p^\perp\in\cH_1$ with corresponding eigenvalues $\rho_p^0$ and $\rho_p^\perp$
that satisfy:
$V(\xi_p^j,\xi_{p'}^j)=\delta_{pp'}$, 
$J(\xi_p^j,\xi_{p'}^j)=\rho_p^j\delta_{pp'}$,
for $j=0,\perp$.
Let $\{\xi_p\}_{p=1}^\infty=\{\xi_p^0,\xi_p^\perp\}_{p=1}^\infty$ and
$\{\rho_p\}_{p=1}^\infty=\{\rho_p^0,\rho_p^\perp\}_{p=1}^\infty$,
where $\rho_p$ are arranged in an increasing order.
It is easy to verify that $\xi_p$'s are Rayleigh quotient eigenvalues of $V/(V+J)$ over $\cH$
as defined in \cite[Section 2]{weinberger1974variational}.
We also have
\begin{equation*}\label{eq:eigen}
V(\xi_p, \xi_{p^\prime}) = \delta_{p p^\prime}, \quad  J(\xi_p, \xi_{p^\prime}) = \rho_p \delta_{p p^\prime}.
\end{equation*}
By (\ref{an:equiv:norm}), the Rayleigh quotients corresponding to $(\|\cdot\|_{L^2(X)\otimes L^2(Z)},
\|\cdot\|_{\cH^{\angular{X}}\otimes \cH^{\angular{Z}}})$ and $(V,V+J)$ are equivalent. 
By the Mapping theorem \cite[Section 3.3]{weinberger1974variational},
there exist constants $c_1,c_2>0$ s.t.
\begin{equation}\label{eq:eigenequal}
\frac{c_1}{1+\pi_p}\le \frac{1}{1+\rho_p}\le\frac{c_2}{1+\pi_p},\,\,p\ge1.
\end{equation}
Following (\ref{eq:eigenequal}) we have $\rho_p\asymp\pi_p\asymp p^{2m/d}$.
By Fourier expansion, we have $\eta=\sum_{p=1}^\infty V(\eta,\xi_p)\xi_p$.

When restricted on $\cH_0$, the Rayleigh quotients corresponding to $(V, V+J)$ and $(\|\cdot\|_{L^2(X)\otimes L^2(Z)},\|\cdot\|_{\cH^{\angular{X}}\otimes \cH^{\angular{Z}}})$ are
still equivalent. Similar to (\ref{eq:eigenequal}), by Mapping theorem, 
\begin{equation}\label{eq:eigenequalh0}
\frac{c_1}{1+\pi^0_p}\le \frac{1}{1+\rho^0_p}\le\frac{c_2}{1+\pi^0_p},\,\,p\ge1.
\end{equation}
where $\{\pi_p^0\}_{p=1}^\infty=\{\mu_k,\nu_l: l=0,\ldots,a-1, k\ge0\}$ are
eigenvalues (with increasing order) corresponding to $\{\phi_k, \psi_l: l=0,\ldots,a-1,k\ge0\}$. 
Specifically, $\pi_p^0=\pi_p$ for $p=1,\ldots,a$, and
$\pi_{a+s}^0=\pi_{sa+1}$ for $s\ge1$.
Now remove $\{\pi_p^0\}_{p\ge1}$ from $\{\pi_p\}_{p\ge1}$ and denote the rest as $\{\pi_p^\perp\}_{p\ge1}$.
From (\ref{eq:eigenequal}) and (\ref{eq:eigenequalh0}), we have 
\begin{equation*}
\frac{c_1}{1+\pi^\perp_p}\le\frac{1}{1+\rho^\perp_p}\le\frac{c_2}{1+\pi^\perp_p},\,\,p\ge1.
\end{equation*}
Since $\nu_1=\dots=\nu_{a-1}=1$ which leads to $\pi^\perp_{r(a-1)+s}=\mu_{r+1}$
for $r\ge0$ and $s=1,\ldots,a-1$, we have 
$\rho_p^\perp\asymp\pi^\perp_p \asymp \mu_{\floor{p/(a-1)}} \asymp p^{2m/d}$. 
\end{proof}

\subsubsection{Proof of Lemma \ref{rkhs:cH}}\label{sec:a9}
\begin{proof}
Following \cite{gu2013smoothing}, $J(\cdot)$ is the roughness penalty, hence it is standard in the sense of \cite{lin2000tensor}.
Following \cite{lin2000tensor}, the norm based on $\int_\cY \eta(x,z)^2dxdz+J(\eta)$
is equivalent to $\|\cdot\|_{\cH^{\angular{X}}\otimes \cH^{\angular{Z}}}$, where $\|\cdot\|_{\cH^{\angular{X}}\otimes \cH^{\angular{Z}}}$ is the tensor product norm induced by the Sobolev norm 
$V_X(g_1,g_2)+J_X(g_1,g_2)$
on $\cH^{\angular{X}}$ and the Euclidean norm on $\cH^{\angular{Z}}$.
Since $f(x,z)$ is bounded away from zero and infinity, 
there exist constants $0<c_1\le c_2<\infty$ such that, for any $\eta\in\cH$,
\begin{equation}\label{an:equiv:norm}
c_1\int_\cY \eta(x,z)^2dxdz\le V(\eta,\eta)\le c_2\int_\cY \eta(x,z)^2dxdz.
\end{equation}
Therefore, $\|\cdot\|$ and $\|\cdot\|_{\cH^{\angular{X}}\otimes\cH^{\angular{Z}}}$ are equivalent norms. 
Since $\cH$ endowed with $\|\cdot\|_{\cH^{\angular{X}}\otimes\cH^{\angular{Z}}}$ is an RKHS,
$(\cH,\langle\cdot,\cdot\rangle)$ is an RKHS. Since $\cH_0$ is a closed subset of $\cH$,
and $\langle\cdot,\cdot\rangle_0$ is inherited from $\langle\cdot,\cdot\rangle$, we have that
$(\cH_0, \inner{\cdot}{\cdot}_0)$ is also an RKHS.
\end{proof}

\subsubsection{Proof of  Proposition \ref{proposition:a1}}\label{sec:a11}

\begin{proof}
The proof of $\|\eta\|^2=\sum_{p=1}^{\infty}|V(\eta, \xi_p)|^2(1+\lambda\rho_p)$ follows by (\ref{inner}) and
the Fourier expansion of $\eta$:
$ \eta = \sum_{p=1}^{\infty} V(\eta, \xi_p) \xi_p$. For any $p'\ge1$,
\begin{equation}\label{eq:fourier}
\inner{\eta}{\xi_{p'}} = \inner{\sum_{p=1}^{\infty} V(\eta, \xi_p) \xi_p}{\xi_{p'}} = V(\eta, \xi_{p'})(1 + \lambda \rho_{p'}).
\end{equation}
By (\ref{eq:fourier}), 
$V(\tilde{K}_\by , \xi_{p})  = \frac{\inner{\tilde{K}_\by}{\xi_{p}}}{1 + \lambda\rho_{p}} = \frac{\xi_p(\by)}{1+\lambda\rho_{p}}$. 
Hence $\tilde{K}_\by (\cdot) = \sum_{p=1}^{\infty}\frac{\xi_p(\by)}{1+ \lambda \rho_p}\xi_p(\cdot)$ follows.
Meanwhile, (\ref{eq:fourier}) implies that
$V(W_{\lambda} \xi_{p}, \xi_{p'})  = \frac{\inner{W_{\lambda} \xi_{p}}{\xi_{p'}}}{1 + \lambda\rho_{p'}} = \frac{\lambda \rho_{p}\delta_{p,p'} }{1+ \lambda \rho_{p}}$. 
Thus we have
$W_{\lambda}\xi_p(\cdot) = \frac{\lambda\rho_p}{1+\lambda\rho_p}\xi_p(\cdot)$.

By Lemma \ref{sim:diag:V:J}, any $\eta \in \cH_0$ satisfies $\eta = \sum_{p=1}^{\infty} V(\eta, \xi^0_p) \xi^0_p$. Therefore, 
$V(\tilde{K}^0_\by , \xi^0_{p})  = \inner{\tilde{K}^0_\by}{\xi^0_{p}}_0/(1 + \lambda\rho^0_{p}) $. Hence,
$\tilde{K}_\by^0(\cdot) = \sum_{p=1}^{\infty} \frac{\xi^0_p(\by)}{1+ \lambda \rho^0_p}\xi^0_p(\cdot)$, and likewise,
$W_{\lambda}\xi^0_p(\cdot) = \frac{\lambda\rho^0_p}{1+\lambda\rho^0_p}\xi^0_p(\cdot)$.
\end{proof}

\subsubsection{Proof of Lemma \ref{lemma:s4}}\label{sec:a6}
We first state and prove several preliminary lemmas.
Define
\begin{align}\label{eq:h}
h^{-1} = \sum_{p=1}^{\infty}\frac{1}{(1+ \lambda \rho_p)^2},\,\,\,\,
h_0^{-1} = \sum_{p=1}^{\infty}\frac{1}{(1+ \lambda \rho^0_p)^2}.
\end{align}
From Lemma \ref{sim:diag:V:J}, we have $\rho_p \asymp p^{2m/d}$ and $\rho^0_p\asymp p^{2m/d}$. 
The following lemma provides an relation between $h$ (or $h_0$) and $\lambda$.
\begin{lemma}\label{lemma:s1}
$h \asymp\lambda^{d/2m}$ and $h_0 \asymp\lambda^{d/2m}$.
\end{lemma}
The following Lemma \label{lemma:a1} presents a relationship between the two
norms $\|\cdot\|_{\sup}$ and $\|\cdot\|$.

\begin{lemma}\label{lemma:s2}
There exists an absolute constant $c_m>0$ s.t. $\|\eta\|_{\sup} \leq c_m h^{-1/2} \|\eta\|$.  
\end{lemma}
Proofs of Lemmas \ref{lemma:s1} and \ref{lemma:s2} can be executed similar to \cite{shang2013local}.

The following two lemmas characterize the convergence rates of $\hetafull$ and $\hetanull$ under $H_0$. 
\begin{lemma}\label{lemma:s3}
Assume $\lambda\to0$ and $H_0$. Then $\|\hetanull - \eta^\ast\|_0 = O_{P}((nh_0)^{-1/2} + \lambda^{1/2})$ and $\|\hetafull - \eta^\ast\|= O_{P}((nh)^{-1/2} + \lambda^{1/2})$.
\end{lemma}
Lemma \ref{lemma:s3} can be proved based on a quadratic approximation method
proposed by \cite{gu2013smoothing},
i.e., apply \cite[Section 9.2.2]{gu2013smoothing} 
to both $(\widehat{\eta}_{n,\lambda},\cH)$ and $(\widehat{\eta}^0_{n,\lambda},\cH_0)$.
The optimal rates for both estimators achieve at $h\asymp n^{-1/(2m+d)}$, $h_0\asymp n^{-1/(2m+d)}$.  
Notice that $\|\cdot\|$ and $\|\cdot\|_0$ are equivalent under the null hypothesis for any $\eta\in\cH_0$. 
Thus, in what follows, we will not distinguish the two norms for notation convenience. 
We also do not distinguish $h$ and $h_0$ since they have the same order for achieving optimality. 

Next, we prove Lemma \ref{lemma:s4} as follows.
\begin{proof}
Let $g=\hetafull - \eta^\ast$. By Taylor's expansion we have 
\begin{equation*}
S_{n,\lambda}(\hetafull) = S_{n,\lambda}(\eta^\ast) + DS_{n,\lambda}(\eta^\ast)g + \int_0^1\int_0^1 sD^2 S_{n,\lambda}(\eta^\ast + ss' g)gg dsds'.
\end{equation*}
By (A.16) and (6), one can check that
$\inner{DS_{n,\lambda}(\eta^\ast)g_1}{g_2} = \inner{g_1}{g_2}$,  
and thus, $DS_{n,\lambda} =id $ is an identity operator.  By the fact $S_{n,\lambda}(\hetafull) = 0$, we have
\begin{equation}\label{eq:l4-1}
\|\hetafull - \eta^\ast - S_{n,\lambda}(\eta^\ast)\| = \|\int_0^1\int_0^1 sD^2 S_{n,\lambda}(\eta^\ast + ss' g)gg dsds'\|.
\end{equation}
By (A.17) we have
$D^2 S_{n,\lambda}(\eta^\ast + ss' g)gg = \int_\cY g(\by)^2  \tilde{K}_\by e^{\eta^\ast(\by) + ss'g(\by)} d\by$.
By Proposition A.1 and Lemma A.3, we have 
\begin{equation*}
\sup_{\by\in \cY} |g(\by)|^2\le c_m h^{-1} \|g\|^2  = c_m h^{-1} O_P((nh)^{-1} + h^{2m}),
\end{equation*} 
where $h^{-1}$ is defined in (\ref{eq:h}). 
By (\ref{eq:ky}), we have
$\|\E_{\eta^\ast}\{\tilde{K}_\bY\}\|\le c_m^{1/2} h^{-1/2}$.
Thus, we have 
\begin{equation}\label{eq:l4-2}
\|D^2 S_{n,\lambda}(\eta^\ast + ss' g)gg \|  =  O(h^{-3/2} ((nh)^{-1} + h^{2m})).
\end{equation}
Plugging (\ref{eq:l4-2}) into (\ref{eq:l4-1}), we finish the proof.
\end{proof}

\subsection{Proof of Theorem \ref{thm1} and Theorem \ref{thm2}}\label{sec:proof:test}

\subsubsection{Proof of Theorem \ref{thm1}}
The proof of Theorem \ref{thm1} is sketched as follows.
By Lemma \ref{lemma:s4}, $n^{1/2}\|\hetanull-\hetafull -S^0_{n,\lambda}(\eta^\ast) + S_{n,\lambda}(\eta^\ast)\|=o_P(1)$.
So we have the following
\begin{align*}
n^{1/2}\|\hetafull-\hetanull\|
=n^{1/2}\|S^0_{n,\lambda}(\eta^\ast) - S_{n,\lambda}(\eta^\ast)\|+o_P(1). 
\end{align*}
Thus we only focus on $n^{1/2}\|S^0_{n,\lambda}(\eta^\ast) - S_{n,\lambda}(\eta^\ast)\|$. Moreover, the following expressions
of $S^0_{n,\lambda}(\eta^\ast)$ and $S_{n,\lambda}(\eta^\ast)$ are reserved for future use:
\begin{equation}\label{scorefull}
S_{n,\lambda}(\eta^\ast) = -\frac{1}{n}\sum_{i=1}^n \tilde{K}_{\bY_i} + \bE_{\eta^\ast} \tilde{K}_{\bY} + W_\lambda\eta^\ast,
\end{equation}
\begin{equation}\label{scorenull}
S^0_{n,\lambda}(\eta^\ast) = -\frac{1}{n}\sum_{i=1}^n \tilde{K}^0_{\bY_i} + \bE_{\eta^\ast}\tilde{K}^0_{\bY} + W^0_\lambda \eta^\ast.
\end{equation}

\begin{proof}[Proof of Theorem \ref{thm1}]

Let us first analyze $I_1$. Let $\widetilde{g} = \hetafull + ss' g - \eta^\ast$, for any $0\leq s,s'\leq 1$. 
By Lemma \ref{lemma:s3}, we have $\|\widetilde{g}\| = O_P((nh)^{-1/2} + h^{m/d}) = o_P(1)$. Notice that
\begin{equation}\label{eq:thm1I11}
D^2\ellnl(\hetafull + ss' g)gg = D^2\ellnl(\widetilde{g} + \eta^\ast) gg = \int_{\cY} g^2(\by) e^{\widetilde{g}(\by)+\eta^\ast(\by)}d\by+ \lambda J(g,g),
\end{equation}
and
\begin{equation}\label{eq:thm1I12}
D^2\ellnl(\eta^\ast)gg = \int_{\cY} g^2(\by) e^{\eta^\ast(\by)}d\by + \lambda J(g,g).
\end{equation}
Combining (\ref{eq:thm1I11}) and (\ref{eq:thm1I12}), we have
\begin{eqnarray*}\label{eq1:a:step}
|D^2\ellnl(\hetafull+ss'g)gg-D^2\ellnl(\eta^\ast)gg|\le\int_{\cY} g^2(\by)e^{\eta^\ast(\by)}|e^{\widetilde{g}(\by)} - 1| d\by.
\end{eqnarray*}
By Taylor expansion of $e^{\widetilde{g}(\by) + \eta^\ast(\by)}$ at $\eta^\ast(\by)$ for any $\by\in\cY$, 
it trivially holds that $e^{\eta^\ast(\by)}|e^{\widetilde{g}(\by)} - 1| = e^{\eta^\ast(\by)}O(|\widetilde{g}(\by)|)$. 
Since $\sup_{\by\in\cY}|\widetilde{g}(\by)|\leq c_m h^{-1/2}\|\widetilde{g}\|$ (Lemma \ref{lemma:s2}), 
and $h^{-1/2}((nh)^{-1}+\lambda)^{1/2}=o(1)$, we have
\begin{equation}\label{term:I1}
|I_1|= O_P(h^{-1/2}(\|\widehat{\eta}_{n,\lambda}-\eta^\ast\|+\|g\|) \cdot \|g\|^2)=o_P(\|g\|^2).
\end{equation}

Let us then analyze $I_2$. From (\ref{eq:Frechet2}), we have
$D^2\ellnl(\eta^\ast)g g = \|g\|^2 = \|\hetafull - \hetanull\|^2$,
which dominates $I_1$, since $h^{-1/2}(\|\widehat{\eta}_{n,\lambda}-\eta^\ast\|+\|g\|)=o_P(1)$. 
Next let us analyze
$\|\hetafull - \hetanull\|^2$. By Lemma \ref{lemma:s4}, we have
\[
n^{1/2}\|\hetanull-\hetafull -S^0_{n,\lambda}(\eta^\ast) + S_{n,\lambda}(\eta^\ast)\|= O_P(n^{1/2} h^{-3/2}((nh)^{-1}+h^{2m/d})) = o_P(1).
\]
Thus we only need to focus on $n^{1/2}\|S^0_{n,\lambda}(\eta^\ast) - S_{n,\lambda}(\eta^\ast)\|$. 
Recall
$S_{n,\lambda}(\widehat{\eta}_{n,\lambda}) = 0$ and $S_{n,\lambda}(\eta^\ast)$, $S^0_{n,\lambda}(\eta^\ast)$
have expressions (\ref{scorefull}), (\ref{scorenull}).
For any $\by\in \cY$, define $\tilde{K}^1_{\by} = \tilde{K}_{\by} - \tilde{K}^0_{\by}$ and  $W^1_\lambda = W_\lambda - W^0_\lambda$, then 
$S_{n,\lambda}^0(\eta^\ast)-S_{n,\lambda}(\eta^\ast)
=-\frac{1}{n}\sum_{i=1}^n \tilde{K}^{1}_{\bY_i} + \E\tilde{K}^1_{\bY}+ W^1_\lambda \eta^\ast$.

By Proposition \ref{proposition:a1}, $\tilde{K}_{\by}^1$ can be expressed as 
a series of $\xi_p^\perp(\by)$. 
Since $\xi_p^\perp\in\cH_1$ and $\phi_0\equiv 1\in\cH_0$, we have 
\[
\bE_{\eta^\ast}\{\xi_p^\perp(\by)\}=\bE_{\eta^\ast}\{\xi_p^\perp(\by)\phi_0(X)\}=V(\xi_p^\perp,\phi_0)=0.
\] 
And so $\bE_{\eta^\ast}\{\tilde{K}_\bY^1\}=0$. Therefore,
$S_{n,\lambda}^0(\eta^\ast)-S_{n,\lambda}(\eta^\ast)
=-\frac{1}{n}\sum_{i=1}^n \tilde{K}^{1}_{\bY_i} + W^1_\lambda \eta^\ast$.
Then
\begin{align*}
n\|S_{n,\lambda}^0(\eta^\ast)-S_{n,\lambda}(\eta^\ast)\|^2= & n^{-1}\|\sum_{i=1}^n  \tilde{K}^{1}_{\bY_i}\|^2     
-2\sum_{i=1}^n  \inner{\tilde{K}^1_{\bY_i} }{W^1_\lambda \eta^\ast}
+n\|W^1_\lambda \eta^\ast\|^2 \\
 \equiv& W_1 -2W_2 + W_3.
\end{align*}
Since $\eta^\ast\in \cH_0$, it follows by Lemma \ref{sim:diag:V:J} 
that $\eta^\ast$ is expanded by a series of $\xi_p^0$. 
By Proposition \ref{proposition:a1},
$W_\lambda\xi_p^0\propto \xi_p^0$ which implies $W_\lambda\eta^\ast=W_\lambda^0\eta^\ast$. And hence, 
$W_\lambda^1\eta^\ast = W_\lambda \eta^\ast - W_\lambda^0 \eta^\ast = 0$
which yields that $W_2=W_3=0$. 
Write $W_1 = n^{-1}\|\sum_{i=1}^n \tilde{K}^{1}_{\bY_i}\|^2=n^{-1} \sum_{i=1}^n \|\tilde{K}^{1}_{\bY_i}\|^2+n^{-1}W(n)$,
where $W(n)=\sum_{i\neq j}\tilde{K}^{1}(\bY_i,\bY_j)$. 

Next let us consider the term $\sum_{i=1}^n \tilde{K}^{1}(\bY_i,\bY_i)$. 
Let $\bE$ denote $\bE_{\eta^\ast}$ unless otherwise indicated.
Let $\theta(n) = \bE\{\tilde{K}^1(\bY_i, \bY_i)\}$.
By Lemma \ref{lemma:s2} we have
$\bE \{|\sum_{i=1}^n \{\tilde{K}^1(\bY_i,\bY_i) - \theta(n)\} |^2\} \leq n \bE\{\tilde{K}^1(\bY_i, \bY_i)^2\}  = O(n h^{-2})$,
so
\begin{equation}\label{eq:est:mean}
\sum_{i=1}^n [\tilde{K}(\bY_i, \bY_i) - \theta(n)] = O_p(n^{1/2} h^{-1}).
\end{equation}

Next, we derive the asymptotic distribution of $W(n)$. 
Define $W_{ij}=2\tilde{K}^1(\bY_i,\bY_j)$,
then $W(n)=\sum_{1\le i<j\le n}W_{ij}$. 
Let $\sigma(n)^2=\mbox{Var}(W(n))$ and
\[
G_I=\sum_{i<j} \E\{W_{ij}^4\},
\]
\[
G_{II}=\sum_{i<j<k} (\E\{W_{ij}^2W_{ik}^2\}+\E\{W_{ji}^2W_{jk}^2\}+\E\{W_{ki}^2W_{kj}^2\}),\,\,\,\,\textrm{and}
\]
\[
G_{IV}=\sum_{i<j<k<l} (\E\{W_{ij} W_{ik} W_{lj} W_{lk}\}+\E\{W_{ij} W_{il} W_{kj} W_{kl}\}+ \E\{W_{ik}W_{il}W_{jk}W_{jl}\}).
\]
By $\bE\{\tilde{K}^1_\bY\}=0$ and direct examinations we have 
\begin{align*}
\sigma^2(n) = \mbox{Var}(W(n)) = & \sum_{1\le i<j\le n} \bE\{(\tilde{K}^1(\bY_i, \bY_j) - \bE[\tilde{K}^1(\bY_i,\bY_j)])^2\} \\
= & \sum_{1\le i<j\le n}\bE\{\tilde{K}^1(\bY_i, \bY_j)^2\}
\asymp n^2h^{-1}.
\end{align*}
Since $\bE\{W_{ij}^4\}=16\E\{\tilde{K}^{1}(\bY_i,\bY_j)^4\}=O(h^{-4})$, we have $G_I=O(n^2 h^{-4})$. Obviously,
$\bE\{W_{ij}^2W_{ik}^2\}\le \bE\{W_{ij}^4\}=O(h^{-4})$, implying $G_{II}=O(n^3h^{-4})$. For pairwise different $i,j,k,l$, we have
\begin{eqnarray*}
\bE\{W_{ij}W_{ik}W_{lj}W_{lk}\}
&=&16 \bE\{\tilde{K}^{1}(\bY_i,\bY_j)
\tilde{K}^{1}(\bY_i,\bY_k)\tilde{K}^{1}(\bY_l,\bY_j)\tilde{K}^{1}(\bY_l,\bY_k)\}\\
&=&\sum_{p=1}^\infty \frac{1}{(1+\lambda\rho^\perp_p)^4}=O(h^{-1}),
\end{eqnarray*}
which leads to $G_{IV}=O(n^4h^{-1})$.

It follows by
$h=o(1)$ and $(nh^2)^{-1}=o(1)$ that $G_I$, $G_{II}$ and $G_{IV}$ are of lower order than $\sigma(n)^4$. 
By Proposition 3.2 of \cite{de1987central} we get that
\begin{equation}\label{limit:Wn}
\frac{W(n)}{\sigma(n)} \overset{d}{\rightarrow}N(0,1).
\end{equation}

From (\ref{eq:est:mean}) and (\ref{limit:Wn}), we get $\frac{1}{n}\sum_{i=1}^n \tilde{K}^{1}(\bY_i,\bY_i)^2=\theta(n)+o_P(1)$, which implies $n\|S_{n,\lambda}^0(\eta^\ast)-S_{n,\lambda}(\eta^\ast)\|^2=O_P(h^{-1}+n\lambda+h^{-1/2})=O_P(h^{-1})$, and hence
$n^{1/2}\|S_{n,\lambda}^0(\eta^\ast)-S_{n,\lambda}(\eta^\ast)\|=O_P(h^{-1/2})$.
Thus,
\begin{align}\label{Global:LRT:Approx0}
&2n\cdot PLR_{n,\lambda}\nonumber = n\|\widehat{\eta}_{n,\lambda}-\eta^\ast\|^2+o_P(h^{-1/2})\nonumber\\
=&\left(n^{1/2}\|S_{n,\lambda}^0(\eta^\ast)-S_{n,\lambda}(\eta^\ast)\|+o_P(1)\right)^2+o_P(h^{-1/2})\nonumber\\
=& n\|S_{n,\lambda}^0(\eta^\ast)-S_{n,\lambda}(\eta^\ast)\|^2+2n^{1/2}\|S_{n,\lambda}^0(\eta^\ast)-S_{n,\lambda}(\eta^\ast)\|\cdot o_P(1) + o_P(h^{-1/2})\nonumber\\
=&n^{-1}\|\sum_{i=1}^n \tilde{K}^{1}_{\bY_i}\|^2 +o_P(h^{-1/2}).
\end{align}
By (\ref{limit:Wn}), (\ref{Global:LRT:Approx0}) and Slutsky's theorem,
$\frac{2n \cdot PLR_{n,\lambda}- \theta(n)}{\sigma(n)/n}\overset{d}{\to}N(0,1)$. Since $\theta_\lambda=\sum_{p=1}^\infty \frac{1}{1+ \lambda \rho^\perp_p}$, $
\sigma^2_\lambda=\sum_{p=1}^\infty \frac{1}{(1+\lambda\rho^\perp_p)^2}$, we have $\theta(n)=  \theta_{\lambda}$ and $\frac{\sigma(n)}{n} = \sqrt{\binom{n}{2}\bE(W^2_{ij})}/n=\sqrt{2}\sigma_{\lambda}$.
\end{proof}

\subsubsection{Proof of Theorem \ref{thm2}}\label{sec:a8}
Before proving Theorem \ref{thm2}, we provide some preliminary lemmas.
For $\eta^\ast\in\cH$,
consider decomposition $\eta^\ast = \eta^\ast_0+ \eta^\ast_{XZ}$ where $\eta^\ast_0$ is the projection of $\eta^\ast$ on $\cH_0$.  
The following lemma says that, for general $\eta^\ast\in\cH$, the restricted penalized likelihood estimator $\widehat{\eta}_{n,\lambda}^0$
converges to $\eta^\ast_0$ with rate of convergence provided.

\begin{lemma}\label{lemma:s5}
Suppose that Assumption \ref{a2} is satisfied. 
We have $\|\hetanull - \eta^\ast_0\|_0= O_{P}((nh)^{-1/2} + \lambda^{1/2})$.
\end{lemma}
Parallel to Lemma \ref{lemma:s4}, when $\eta^\ast\in\cH$, we have the following result characterizing the higher order
expansion of $\widehat{\eta}_{n,\lambda}^0$.
\begin{lemma}\label{lemma:s6}
Suppose that $nh^2\rightarrow\infty$. We have
\begin{eqnarray*}\label{baha:reps1}
\|\hetanull-\eta_0^\ast-S^0_{n,\lambda}(\eta_0^\ast)\|_{0}= O_P(h^{-3/2} ((nh)^{-1} + h^{2m/d})).
\end{eqnarray*}
\end{lemma}

\begin{proof}[Proof of Theorem \ref{thm2}]
Let $g= \hetanull - \hetafull$.
Recall the Taylor expansion (\ref{eq:tylor}):
\begin{equation}
\begin{aligned}
PLR_{n,\lambda} =& \ell_{n,\lambda}(\widehat{\eta}^0_{n,\lambda}) - \ell_{n,\lambda}(\widehat{\eta}_{n,\lambda}) \nonumber\\
= & \int_0^1\int_0^1 s\{D^2f(\hetafull+ss'g)gg - D^2f(\eta^\ast)gg \}dsds' + \frac{1}{2}D^2f(\eta^\ast)gg\nonumber\\
= &O_P((\|\widehat{\eta}_{n,\lambda}-\eta^\ast\|_{\sup}+\|g\|_{\sup})\cdot\|g\|^2) + \frac{1}{2}\|g\|^2,\label{taylor:exp:thm2}
\end{aligned}
\end{equation}
where the $O_P$ term in the last equation follows from (\ref{term:I1}).
By Lemmas \ref{lemma:s2} and \ref{lemma:s3}, $\|\widehat{\eta}_{n,\lambda}-\eta^\ast\|_{\sup}=o_P(1)$.
By assumption $\|\eta_{XZ}^\ast\|_{\sup}\le (\log{n})^{-1}=o(1)$ and Lemma \ref{lemma:s5}, we have 
$\|g\|_{\sup}=\|\widehat{\eta}^0_{n,\lambda}-\eta_0^\ast+\eta^\ast-\widehat{\eta}_{n,\lambda}-\eta^\ast_{XZ}\|_{\sup}=o_P(1)$.
Hence, the $O_P$ term in (\ref{taylor:exp:thm2}) is dominated by $\frac{1}{2}\|g\|^2$, for which
we only focus on the latter.
Combining the results of Lemmas \ref{lemma:s4} and \ref{lemma:s6}, we have
\begin{align*}
&\|\widehat{\eta}_{n,\lambda}-\eta^\ast-S_{n,\lambda}(\eta^\ast)\| = O_P(h^{-2} ((nh)^{-1}+ h^{2m/d})), \\
&\|\widehat{\eta}^0_{n,\lambda}-\eta_0^\ast-S^0_{n,\lambda}(\eta_0^\ast)\|_0 =  O_P(h^{-2} ((nh)^{-1} + h^{2m/d})).
\end{align*}
Recalling $\eta^\ast-\eta_0^\ast=\eta^\ast_{XZ}$,
we have $\|g\| = \|\eta^\ast_{XZ} +  S_{n,\lambda}(\eta^\ast) - S^0_{n,\lambda}(\eta_0^\ast)\| + O_P(h^{-2} ((nh)^{-1} + h^{2m/d}))$. In 
what follows, we focus on $\|\eta_{XZ}^\ast +  S_{n,\lambda}(\eta^\ast) - S^0_{n,\lambda}(\eta^\ast)\|$.
By definition of $S_{n,\lambda}(\eta^\ast), S^0_{n,\lambda}(\eta_0^\ast)$ (see (\ref{eq:Frechet1})) and direct calculations, it can be shown that
\begin{align*}
 &\|\eta_{XZ}^\ast+  S_{n,\lambda}(\eta^\ast) - S^0_{n,\lambda}(\eta_0^\ast)\|^2 \\
 =& \|\frac{1}{n}\sum_{i=1}^{n}\tilde{K}_{\bY_i}^1\|^2 + \|\eta^\ast_{XZ}\|^2 + \|\bE \tilde{K}_{\bY} -   \bE \tilde{K}^0_\bY\|^2 + \|W^1_\lambda \eta^\ast_{XZ}\|^2 \\
 & -\frac{2}{n}\sum_{i=1}^{n} \eta^\ast_{XZ}(\bY_i) + 2\bE\eta^\ast_{XZ}(\bY) + 2\inner{W_\lambda^1\eta^\ast_{XZ}}{\eta^\ast_{XZ}} -\frac{2}{n}\sum_{i=1}^n\bE \tilde{K}^1(\bY_i,\bY)\\
 & -\frac{2}{n}(W_\lambda^1\eta^\ast_{XZ})(\bY_i) + 2\bE(W_\lambda^1\eta^\ast_{XZ})(\bY),
\end{align*}
where $\bE$ denotes $\bE_{\eta^\ast}$.
Since $\bE \{\tilde{K}_\bY-\tilde{K}_\bY^0\}=\bE  \tilde{K}_\bY^1=0$, we have
\begin{align*}
&\|\eta_{XZ}^\ast+  S_{n,\lambda}(\eta^\ast) - S^0_{n,\lambda}(\eta_0^\ast)\|^2 \\
 \geq & \|\frac{1}{n}\sum_{i=1}^{n}\tilde{K}_{\bY_i}^1\|^2 + \|\eta^\ast_{XZ}\|^2 + [-\frac{2}{n}\sum_{i=1}^{n} \eta^\ast_{XZ}(\bY_i) + 2\bE_{\eta^\ast} \eta^\ast_{XZ}(\bY)] + 2\inner{W_\lambda^1\eta^\ast_{XZ}}{\eta^\ast_{XZ}}\\
  &+ [-\sum_{i=1}^n\frac{2}{n}(W_\lambda^1\eta^\ast_{XZ})(\bY_i) + 2\bE_{\eta^\ast}(W_\lambda^1\eta^\ast_{XZ})(\bY)]
\equiv V_1 + V_2 + V_3 + V_4 + V_5.
\end{align*}

Since $\var(V_3)\le \frac{4}{n}\bE (\eta^\ast_{XZ}(\bY))^2 \leq \frac{4}{n}\|\eta^\ast_{XZ}\|^2$,
\begin{equation}\label{eq:v3}
V_3= O_P(n^{-1/2})\|\eta^\ast_{XZ}\|.
\end{equation}
By assumption $J(\eta^\ast_{XZ},\eta^\ast_{XZ})\le C$, we have
\begin{equation}\label{eq:v4}
V_4 = \lambda J(\eta^\ast_{XZ}, \eta^\ast_{XZ})\le C\lambda.
\end{equation}
Since $\var (V_5) \leq \bE |(W_\lambda \eta^\ast_{XZ})|^2 = V(W_\lambda \eta^\ast_{XZ}, W_\lambda \eta^\ast_{XZ})$. By Proposition \ref{proposition:a1}, we have
\begin{equation*}
V(W_\lambda \eta^\ast_{XZ}, W_\lambda \eta^\ast_{XZ}) = \sum_{p=1}^\infty |V(\eta^\ast_{XZ}, \xi_p)|^2 \left(\frac{\lambda\rho_p}{1+\lambda\rho_p}\right)^2 = o(\lambda),
\end{equation*}
where the last equality follows by $\sum_{p=1}^\infty |V(\eta^\ast_{XZ}, \xi_p)|^2\rho_p< \infty$ and the dominated convergence theorem. Thus we have 
\begin{equation}\label{eq:v5}
V_5 = o_p(n^{-1/2}\lambda^{1/2})
\end{equation}
Combining (\ref{eq:v3}), (\ref{eq:v4}) and (\ref{eq:v5}) we have
\begin{align*}
&\frac{2n\cdot PLR_{n,\lambda} - \theta(n)}{\sigma(n)}\\
\geq &\frac{2n\cdot V_1 -\theta(n)}{\sigma(n)} + \frac{2n\cdot(V_2+V_3+V_4+V_5)}{\sigma(n)}\\
\geq & O_P(1) + 2n\sigma^{-1}(n)(\|\eta^\ast_{XZ}\|^2 + O_P(n^{-1/2}\|\eta^\ast_{XZ}\|) + O(\lambda) + o_P(n^{-1/2}\lambda^{1/2})).
\end{align*}
For $C_\varepsilon>0$ sufficiently large, let $\eta^\ast_{XZ}$ satisfy $\|\eta^\ast_{XZ}\|^2\ge C_\varepsilon n^{-1/2}\|\eta^\ast_{XZ}\|$, 
$\|\eta^\ast_{XZ}\|^2 \ge C_\varepsilon \lambda $, $nh^{1/2}\|\eta^\ast_{XZ}\|^2\ge C_\varepsilon$,
$n\|\eta^\ast_{XZ}\|^2/\sigma(n)\ge C_\varepsilon$, which implies that with probability greater than $1-\varepsilon$,
$|\frac{2n\cdot PLR_{n,\lambda} - \theta(n)}{\sigma(n)}| \geq c_\alpha$ (i.e., $\Phi_{n,\lambda}(\alpha)=1$),
where $c_\alpha$ is the $1-\alpha$ percentile of standard normal distribution. 
It can be seen that the above conditions on $\eta^\ast_{XZ}$
are satisfied if
$\|\eta^\ast_{XZ}\|^2 \ge C_\varepsilon(\lambda + (nh^{1/2})^{-1})$. 
The result follows immediately by the fact $\|\eta^\ast_{XZ}\|_2\le \|\eta^\ast_{XZ}\|$. Proof is completed.
\end{proof}

\subsection{Proofs of the Minimax Lower Bound in Section \ref{sec:minimax}}\label{sec:proof:minimax}
\subsubsection{Preliminaries for the minimax lower bound}\label{sec:a2}

\begin{lemma}\label{lem:errlowbound}
Let $\PP_0$ be the probability measure under the null, and $\PP_1$ be the probability with density in $\{\eta\mid \|\eta_{XZ}\|_{\cH}<d_n\}$. We have
\begin{equation*}
\inf_{\phi_n} {\normalfont\err}(\phi_n, d_n) \ge 1 - \delta(\sqrt{\delta+4}-\delta),
\end{equation*}
where $\delta^2 = \E_{\PP_0}(d\PP_1/d\PP_0 - 1)^2$. 
\end{lemma}

\begin{proof}
The test is bounded below by $1-\|\PP_0 - \PP_1\|_{TV}$, where $\|\cdot\|_{TV}$ is the total variation distance between $\PP_0$ and $\PP_1$. By the theorem in \cite{ingster1987minimax}, we have 
\begin{equation*}
\frac{1}{2}\|\PP_0 - \PP_1\|_{TV} \leq \delta (1- \frac{1}{2}|\PP_0 - \PP_1\|_{TV})^{1/2},
\end{equation*}
which directly implies the result.
\end{proof}

\subsubsection{Proof of Lemma \ref{lemma:errbound}}\label{sec:a3}
\begin{proof}
As show in Lemma \ref{lem:errlowbound}, we have 
\begin{equation}\label{eq:lowboundrisk}
\inf_{\phi_n} \err(\phi_n, d_n) \ge 1 - \delta(\sqrt{\delta+4}-\delta).
\end{equation}
Next we show that if $d_n^2 \leq \frac{\sqrt{k_B(d_n)}}{4n}$,  we have that the last term in (\ref{eq:lowboundrisk}) is larger than $1/2$. For simplicity, denote $k=k_B(d_n)$. For any $b=(b_1,\ldots,b_k)\in\{-1,1\}^k$, let $\btheta_b = \frac{d_n}{\sqrt{k}}\sum_{i=1}^k b_i \eb_i \in \RR^N $, where $\eb_i$ is the standard basis vector with $i$th coordinate as one. We assume $b$ is 
uniformly distributed over $\{-1,1\}^k$ so that $\btheta_b$ is
uniformly distributed over $\QQ := \{\btheta_b: b\in\{-1,1\}^k\}$. Since  $\EE_{\PP_0} e^{\eta_{XZ}^{\btheta_b}} -1 =0 $,  we have $  e^{\eta_{XZ}^{\btheta_b}} -1\in \cH_{11}$.  Define 
\begin{equation}\label{eq:etaxg}
\exp(\eta_{XZ}^{\btheta_b}) -1  = \frac{d_n}{\sqrt{k}}\sum_{l=1}^k b_l\psi_l\phi_1, 
\end{equation}
where $\{\psi_l\phi_1\}_{l=1}^k$ are basis function for $\cH_{11}$. 
We denote $\PP^{(n)}_1$ and $\PP^{(n)}_0$ as the empirical meaures under the alternative and null respectively.  The ratio of densities of  $\PP^{(n)}_1$ and $\PP^{(n)}_0$ is
 $$\frac{d\PP^{(n)}_1}{d\PP^{(n)}_0}=\EE_{\btheta_b}\prod_{i=1}^n \exp(\eta^{\btheta_b}_{XZ}(\bY_i)).$$ 
 
Then, we denote the empirical version of $\delta$ as $\delta_n$ which can be written as
\begin{align*}
\delta_n^2 &= \E_{\PP^{(n)}_0}(d\PP^{(n)}_1/d\PP^{(n)}_0 - 1)^2\\
         & =  \E_{\PP^{(n)}_0}[\E_{\btheta_b}\prod_{i=1}^n \exp(\eta^{\btheta_b}_{XZ}(\bY_i))]^2 - 1\\
         & = \E_{\PP^{(n)}_0}[\EE_{\btheta_b}\prod_{i=1}^n \exp(\eta^{\btheta_b}_{XZ}(\bY_i))][\EE_{\btheta_{b^\prime}}\prod_{i=1}^n \exp(\eta^{\btheta_{b^\prime}}_{XZ}(\bY_i))]-1\\
         &=  \EE_{\btheta_b,\btheta_{b^\prime}} \prod_{i=1}^n \E_{\PP_0}\exp(\eta^{\btheta_b}_{XZ}(\bY_i))\exp(\eta^{\btheta_{b^\prime}}_{XZ}(\bY_i)) - 1 \\
         &= \EE_{\btheta_b,\btheta_{b^\prime}} [\E_{\PP_0}\exp(\eta^{\btheta_b}_{XZ}(\bY))\exp(\eta^{\btheta_{b^\prime}}_{XZ}(\bY))]^n - 1.
\end{align*}
Plugging (\ref{eq:etaxg}) in, we have
\begin{align*}
\delta_n + 1 &=  \EE_{\btheta_b}\EE_{\btheta_b^\prime} [\E_{\PP_0}(1 + \frac{d_n}{\sqrt{k}}\sum_{l=1}^k b_l\psi_l\phi_1)(1 + \frac{d_n}{\sqrt{k}}\sum_{l=1}^k b^\prime_l\psi_l\phi_1)]^n \\
& = \frac{1}{2^k} \sum_{b,b^\prime}  (1+ \frac{d_n^2}{k}b^T b^\prime)^n \\
& \le \frac{1}{2^k} \sum_{b}\exp\{ \frac{n d^2_n  b^T 1_{k}}{k} \}\\
&= \frac{1}{2^k}\sum_{i=0}^k\binom{k}{i} \exp\{\frac{n (k-2i)d_n^2}{k}\}\\
 &= \frac{1}{2^k} \big(\exp\{\frac{nd_n^2}{k}\} + \exp\{-\frac{nd_n^2}{k}\} \big)^k \\
 &\overset{(i)}{\leq}  (1+ \frac{n^2d_n^4}{k^2})^k \\
 & \overset{(ii)}{\leq}  \exp\{\frac{n^2d_n^4}{k}\},
\end{align*}
where (i) is due to the fact that  $\frac{1}{2}(\exp(x)+\exp(-x))\le 1+x^2$ for $|x|\le 1/2$ and (ii) is due to the fact $1+x \le e^x$. Thus for any $d_n^4 \leq \frac{k}{16n^2}$, we have 
\begin{equation*}
\inf_{\phi_n} \err(\phi_n, d_n) \ge 1- \delta_n(\sqrt{\delta_n+4}-\delta_n) \ge 1 - e^{1/16}(\sqrt{e^{1/16}+4}) \ge 1/2.
\end{equation*}

For $d_n \lesssim k^{1/4}/\sqrt{n}$, we have
\begin{equation*}
|\exp\{\eta_{XZ}^{\theta_b}\}-1| = \frac{d_n}{\sqrt{k}}|\sum_{l=1}^k b_l\psi_l\phi_1|\lesssim \frac{k^{3/4}}{\sqrt{n}}.
\end{equation*}
Thus, there exsits $c_1, c_2 >0$ such that
\begin{equation}
c_1 |\eta_{XZ}^{\theta_b}(\by)| < |\exp\{\eta_{XZ}^{\theta_b}\}-1| < c_2 |\eta_{XZ}^{\theta_b}(\by)|,
\end{equation}
which indicates that $\|\exp\{\eta_{XZ}^{\theta_b}\}-1\|_2 \asymp \|\eta_{XZ}^{\theta_b}\|_2$. By the definition of $r_B(\delta^\ast)$, we have $\err(\phi_n, d_n) > 1/2$ for all $d_n \le r_B(\delta^\ast)$ .
\end{proof}

\subsubsection{Proof of Lemma \ref{lemma:kb}}\label{sec:a4}
\begin{proof}
We show that $b_{k,2}(\cE_{11} )$ is bounded below by $\sqrt{\gamma_{k+1}}$. It is sufficient to show that $\cE_{11}$ contains a $l_2$ ball centered at $\eta_{XZ}=0$ with radius $\sqrt{\gamma_{k+1}}$. For any $v\in \cE_{11}$ with $\|v\|_2 \le \sqrt{\gamma_{k+1}}$, we have 
\begin{equation*}
b_{2,k} \overset{(i)}{\leq} \sum_{i=1}^{k+1} \frac{v_i^2}{\gamma_i} \overset{(ii)}{\leq} \frac{1}{\mu_{k+1}} \sum_{i=1}^{k+1} v_i^2
\end{equation*}
where the inequality (i) holds by set the $(k+1)$-dimensional subspace spaned by the eigenvectors corresponding to the first $(k+1)$ largest eigenvalues; the inequality (ii) holds by the decreasing order of the eigenvalues, i.e., $\gamma_1\geq\gamma_2\geq\dots\gamma_{k+1}$.

Recall that the definition of the Bernstein lower critical dimension is $k_B(\delta) = \argmax_{k}\{b^2_{k-1,2}(\cE_{11}) \geq \delta^2\}$, we have 
\begin{equation*}
 k_B(\delta) \ge \argmax_{k}\{\sqrt{\gamma_k} \geq \delta\}.
\end{equation*}
\end{proof}

\subsubsection{Proof of Theorem \ref{thm3}}\label{sec:a5}
\begin{proof}
By Lemme \ref{lemma:errbound}, we have 
\begin{equation*}
d_n \le \sup\{\delta: k_B(\delta) \ge  16n^2\delta^4\}.
\end{equation*}
Then we plug in the lower bound of $k_B$ in Lemma \ref{lemma:kb} and we have
\begin{equation}\label{eq:dnlowbound}
d_n \le \sup \{\delta: \argmax_{k}\{\sqrt{\gamma_k} \geq \delta\} \ge  16n^2\delta^4\}
\end{equation}
The eigenvalues have polynomial decay rate i.e., $\gamma_k\asymp k^{-2m/d}$, and consequently,
$\argmax_{k}\{\sqrt{\gamma_k}\geq \delta\}  \asymp \delta^{-d/m}$. Plugging this into (\ref{eq:dnlowbound}),
it is easy to see that the supremum on the right hand side has an order $n^{-\frac{2m}{4m+d}}$.
Proof is thus completed.
\end{proof}

\setcounter{subsection}{0}
\renewcommand{\thesubsection}{B.\arabic{subsection}}
\setcounter{equation}{0}
\renewcommand{\theequation}{B.\arabic{equation}}
\renewcommand{\thelemma}{B.\arabic{Lemma}}
\renewcommand{\theproposition}{B.\arabic{proposition}}

\section{Proofs of the Auxilary Results}

In this document, additional proofs of auxillary lemmas are included.

\begin{itemize}
\item Section \ref{sec:b1} includes the proof of Lemma \ref{lemma:s1}.
\item Section \ref{sec:b2} includes the proof of Lemma \ref{lemma:s2}.
\item Section \ref{sec:b3} includes the proof of Lemma \ref{lemma:s3}.
\item Section \ref{sec:b4} includes the proof of Lemma \ref{lemma:s5}.
\item Section \ref{sec:b6} includes the proof of Lemma \ref{lemma:s6}.
\end{itemize}

\subsection{Proof of Lemma \ref{lemma:s1}}\label{sec:b1}
Since $\rho_p \asymp p^{2m/d}$, we have
\begin{align*}
h^{-1}  =   \sum_{p=0}^{\infty} \frac{1}{(1+\lambda \rho_p)^2}\asymp\int_{1}^\infty \frac{1}{(1 + \lambda  p^{2m/d})^2 }
=\int_{\lambda^{d/2m}}^\infty \frac{1}{( 1 + x^{2m/d})^2} dx = O(\lambda^{-d/2m})
\end{align*}
Thus we have $h \asymp  \lambda^{d/2m}$. Similarly, $h_0 \asymp \lambda^{d/2m}$.

\subsection{Proof of Lemma \ref{lemma:s2}}\label{sec:b2}
For any $\by\in\cY$ and $\eta\in\cH$, we have $|\eta(\by)| = |\inner{\tilde{K}_\by}{\eta}| \leq \|\tilde{K}_\by\| \cdot \|\eta\|$. 
So it is sufficient to find the upper bound for $\|\tilde{K}_\by\|$. By Proposition A.1 and the boundedness of $\xi_p$'s, we have
\begin{equation}\label{eq:ky}
\|\tilde{K}_\by\|^2 = \tilde{K}(\by, \by) = \sum_{p=1}^{\infty} \frac{|\xi_p(\by)|^2}{1+ \lambda \rho_p} \leq c_m h^{-1}
\end{equation}
where $c_m>0$ is a constant free of $\by$ and $\eta$.

\subsection{Proof of Lemma \ref{lemma:s3}}\label{sec:b3}
The proof is rooted in \cite{gu2013smoothing}.
Consider the quadratic approximation of the integral $\int_{\cY} e^{\eta(\by)} d\by$:
\begin{equation}\label{eq:quadratic}
\int_{\cY} e^{\eta(\by)} d\by\approx\int_{\cY} e^{\eta^\ast(\by)} d\by + \int_{\cY} (\eta - \eta^\ast) e^{\eta^\ast(\by)} d\by +\frac{1}{2} V(\eta-\eta^\ast, \eta-\eta^\ast). 
\end{equation}
Dropping the terms that do not involve $\eta$, and plugging (\ref{eq:quadratic}) into (4), $\ell_{n,\lambda}(\eta)$ has
a quadratic approximation $q_{n,\lambda}(\eta)$:
\begin{equation}\label{eq:quadraticpl}
q_{n,\lambda}(\eta)=-\frac{1}{n} \sum_{i=1}^n \eta(\bY_i) + \int_{\cY} \eta e^{\eta^\ast} d\by + \frac{1}{2}V(\eta-\eta^\ast, \eta-\eta^\ast) +\frac{1}{2} J(\eta, \eta). 
\end{equation}
Consider the Fourier expansions of $\eta$ and $\eta^\ast$: 
\begin{equation*}
\eta(x,z) = \sum_{k=1}^{\infty}\sum_{l=1}^a \beta_{kl}\phi_k(x)\psi_l(z),   
\,\,\,\,\eta^\ast(x,z) = \sum_{k=1}^{\infty}\sum_{l=1}^a \beta^\ast_{kl}\phi_k(x)\psi_l(z).
\end{equation*}
Then, we have
\begin{eqnarray}\label{eq:quad}
q_{n,\lambda}(\eta)
&=&\sum_{k=1}^\infty\sum_{l=1}^a \left\{-\beta_{kl}(\frac{1}{n}\sum_{i=1}^{n}\phi_k(x_i)\psi_l(z_i) - \bE\{\phi_k(X)\psi_l(Z)\}\right. \nonumber\\
&&\left.+ \frac{1}{2}(\beta_{kl}-\beta^\ast_{kl})^2 + \frac{\lambda}{2}\mu_k\nu_l \beta_{kl}^2 \right\}. 
\end{eqnarray}
Write $\gamma_{kl} = n^{-1}\sum_{i=1}^n \phi_k(X_i)\psi_l(Z_i) - \bE\{\phi_k(X)\psi_l(Z)\} $. Minimizing (\ref{eq:quad}) with respect to $\beta_{kl}$'s, we get the optimizer:
\begin{equation*}
\widetilde{\beta}_{kl} = (\gamma_{kl} + \beta^\ast_{kl})/(1+\lambda\mu_k\nu_l),\,\,\,\, k\ge1, l=1,\ldots,a.
\end{equation*} 
Then  $\widetilde{\eta}=\sum_{k=1}^\infty\sum_{l=1}^a \widetilde{\beta}_{kl} \phi_k\psi_l$ becomes a linear approximation of $\hetafull$.
By direct calculations we get that
\begin{align*}
V(\widetilde{\eta} - \eta^\ast) = \sum_{k=1}^\infty \sum_{l=1}^a(\beta_{kl} - \beta^\ast_{kl})^2,\,\,\,\, 
\lambda J(\widetilde{\eta} - \eta^\ast) = \sum_{i=1}^\infty \sum_{j=1}^a \lambda\mu_k\nu_l(\beta_{kl} - \beta^\ast_{kl})^2. 
\end{align*}
Since $\bE\gamma_{kl} = 0$ and $\bE\gamma^2_{kl}=1/n$, we have
\begin{equation}\label{eq:vjexp}
\begin{aligned}
\bE\{V(\widetilde{\eta}-\eta^\ast)\}  = \sum_{i=1}^\infty \sum_{j=1}^a \frac{1}{(1+\lambda \mu_k\nu_l)^2} + \lambda \sum_{i=1}^\infty \sum_{j=1}^a \frac{\lambda\mu_k\nu_l}{(1+\lambda \mu_k\nu_l)^2}\mu_k\nu_l\beta^\ast_{kl}\beta^\ast_{kl}\\
\bE\{\lambda J(\widetilde{\eta}-\eta^\ast)\}  = \sum_{i=1}^\infty \sum_{j=1}^a \frac{1}{(1+\lambda \mu_k\nu_l)^2} + \lambda \sum_{i=1}^\infty \sum_{j=1}^a \frac{(\lambda\mu_k\nu_l)^2}{(1+\lambda \mu_k\nu_l)^2}\mu_k\nu_l\beta^\ast_{kl}\beta^\ast_{kl} 
\end{aligned}
\end{equation}
By similar derivations in Lemma \ref{lemma:s2}, it can be verified that
\begin{equation*}
   \sum_{i=1}^\infty \sum_{j=1}^a \frac{1}{(1+\lambda \mu_k\nu_l)^2} =O(\lambda^{-1/2m}),  
\end{equation*}
\begin{equation*}
    \sum_{i=1}^\infty \sum_{j=1}^a \frac{\lambda\mu_k\nu_l}{(1+\lambda \mu_k\nu_l)^2} =O(\lambda^{-1/2m}),
\end{equation*}
\begin{equation*}
    \sum_{i=1}^\infty \sum_{j=1}^a \frac{1}{(1+\lambda \mu_k\nu_l)} =O(\lambda^{-1/2m})
\end{equation*}
Plugging into (\ref{eq:vjexp}), we obtain that
\begin{equation}\label{eq:esterror}
\|\widetilde{\eta} - \eta^\ast\|^2 = (V+ \lambda J)(\widetilde{\eta} - \eta^\ast) = O_p(n^{-1}\lambda^{-1/2m} + \lambda).
\end{equation}

We now turn to the approximation error $\widehat{\eta} -\widetilde{\eta}$. We calculate the Fr\'{e}chet derivative of the quadratic approximation in (\ref{eq:quadraticpl}) as
\begin{equation}\label{Dq:eq}
Dq_{n,\lambda}(\eta)\Delta\eta = - \frac{1}{n} \sum_{i=1}^n \Delta \eta(\bY_i) + \int_{\cY} \Delta \eta e^{\eta^\ast} d\by + \lambda V(\eta-\eta^\ast,\Delta\eta) + \lambda J(\eta, \Delta\eta).
\end{equation}
Since $Dq_{n,\lambda}(\widetilde{\eta}) = 0$, setting $\Delta \eta = \hetafull - \widetilde{\eta}$, (\ref{Dq:eq}) is equal to
\begin{multline}\label{eq:l1}
- \frac{1}{n} \sum_{i=1}^n (\hetafull - \widetilde{\eta})(\bY_i) + \int_{\cY} (\hetafull - \widetilde{\eta})(\by)e^{\eta^\ast(\by)} d\by 
+ V(\widetilde{\eta}-\eta^\ast,\hetafull - \widetilde{\eta}) + \lambda J(\widetilde{\eta}, \hetafull - \widetilde{\eta})
\end{multline}
Since $D\ell_{n,\lambda}(\hetafull) =0 $, setting $\Delta \eta = \hetafull - \widetilde{\eta}$ yields
\begin{eqnarray}\label{eq:q1}
D\ell_{n,\lambda}(\eta)\Delta\eta&=&-\frac{1}{n}\sum_{i=1}^{n}( \hetafull - \widetilde{\eta})(\bY_i) + \int_\cY(\hetafull - \widetilde{\eta})(\by)e^{\hetafull(\by)}d\by\nonumber + \lambda J(\hetafull, \hetafull - \widetilde{\eta}). 
\end{eqnarray}
Combining (\ref{eq:l1}) and (\ref{eq:q1}), we have
\begin{multline*}
\int(\hetafull - \widetilde{\eta})(\by) e^{\hetafull(\by)} d\by -    \int_\cY (\hetafull - \widetilde{\eta})(\by) e^{\widetilde{\eta}(\by)} d\by + \lambda J(\hetafull - \widetilde{\eta}) \\
 = V(\widetilde{\eta}-\eta^\ast, \hetafull - \widetilde{\eta}) + \int_\cY(\hetafull - \widetilde{\eta})(\by) e^{\eta^\ast(\by)} d\by -    \int_\cY (\hetafull - \widetilde{\eta})(\by) e^{\widetilde{\eta}(\by)} d\by. 
\end{multline*}

By Taylor expansion,
\begin{equation*}
 \int (\hetafull - \widetilde{\eta})(\by) e^{\widetilde{\eta}(\by)} d\by - \int_\cY(\hetafull - \widetilde{\eta})(\by) e^{\eta^\ast(\by)} d\by  = V(\hetafull - \widetilde{\eta}, \widetilde{\eta} - \eta^\ast)(1+o_p(1)),
\end{equation*}
where the $o_P$ term holds as $\lambda\to 0 $ and $n\lambda^{1/2m} \to \infty$. Define 
$$D(\alpha) = \int_\cY (\hetafull - \widetilde{\eta})(\by) e^{\hetafull(\by) + \alpha(\hetafull-\widetilde{\eta})(\by)} d\by.$$ 
It can be shown that $\dot{D}(\alpha) = V_{\widetilde{\eta} + \alpha(\hetafull-\widetilde{\eta}} (\hetafull - \widetilde{\eta})$. By the mean value theorem,
\begin{align*}
& \int_\cY(\hetafull - \widetilde{\eta})(\by) e^{\hetafull(\by)} d\by -    \int_\cY (\hetafull - \widetilde{\eta})(\by) e^{\widetilde{\eta}(\by)} d\by\\
= & D(1) - D(0) 
= \dot{D}(\alpha) = V_{\widetilde{\eta} + \alpha(\hetafull-\widetilde{\eta}} (\hetafull - \widetilde{\eta}),
\end{align*}
for some $\alpha\in[0,1]$. Then by Assumption 1, we have
\begin{multline*}
c_1V(\hetafull - \widetilde{\eta}) + \lambda J(\hetafull - \widetilde{\eta}) \leq o_p ( V(\widetilde{\eta} -\eta^\ast, \widehat{\eta} - \widetilde{\eta})) 
= o_p(\{V(\hetafull-\widetilde{\eta})V(\widetilde{\eta}-\eta^\ast)\}^{1/2})
\end{multline*}
Combine with the estimation error (\ref{eq:esterror}), we have
\begin{equation*}
\|\hetafull - \eta^\ast\|^2 = V(\hetafull - \eta^\ast) + \lambda J(\hetafull - \eta^\ast) = O_p(n^{-1}\lambda^{1/2m} +\lambda).
\end{equation*}

\subsection{Proof of Lemma \ref{lemma:s5}}\label{sec:b4}
Suppose the $\eta_0^\ast$ is the projection of $\eta^\ast$ on $\cH_0$. Define an index set 
$\cI_0 = \{(k,l)|k = 1 \mbox{ or } l=1\}$ corresponding to the basis, $\{\phi_k\psi_l| k=1 \mbox{ or }l=1\}$, of $\cH_0$.
When restricted to $\cH_0$, the Fourier expansion of $\eta^\ast$ is
\begin{equation*}
\eta^\ast_0(x,z) = \sum_{(k,l)\in\cI_0} \beta^0_{kl}\phi_k(x)\psi_l(z).
\end{equation*}
Substituting the above $\eta_0^\ast$ as well as its Fourier expansion
into the proof of Lemma A.4, all results remain valid, provided the following truth:
\begin{align*}
\bE\{\frac{1}{n}\sum_{i=1}^n\phi_k(X_i)\psi_l(Z_i)-\bE_{\eta^\ast}(\phi_k\psi_l)\}^2 = \frac{1}{n}\\
\bE\{\frac{1}{n}\sum_{i=1}^{n}\phi_k(X_i)\psi_l(Z_i)\phi_{k'}(X_i)\psi_{l'}(Z_i)-\bE_{\eta^\ast}(\phi_k\psi_l\phi_{k'}\psi_{l'})\}^2\leq \frac{c}{n},
\end{align*}
where $c$ is a positive constant. The existence of such $c$ is guaranteed by the uniform boundedness of $\phi_k(x)$'s as proved by \cite{shang2013local}. Let $\eta^\ast_0$ be the projection of $\eta^\ast$ on the subspace $\cH_0$
and $g=\hetanull - \eta_0^\ast$. Substituting $\eta^\ast_0$ and $\hetanull$ into the proof of Lemma A.4, the results would follow.

\subsection{Proof of Lemma \ref{lemma:s6}}\label{sec:b6}
Let $\eta^\ast_0$ be the projection of $\eta^\ast$ on the subspace $\cH_0$
and $g=\hetanull - \eta_0^\ast$. Substituting $\eta^\ast_0$ and $\hetanull$
into the proof of Lemma 3.4, one can show the desired results.

\section{Additional Simulation with Beta and Beta Mixture
}\label{sec:cngd}

In this section, we consider the distribution with different shapes. Specifically, we considered the follow two settings:

\noindent {Setting 5: }  the simple Beta distributions,
\begin{eqnarray*}
 \ \ \  \ X\mid Z=z &\sim & Beta \left(2(1+\delta_5 \mathbbm{1}_{z=1}),2(1+\delta_5\mathbbm{1}_{z=1})\right)
\end{eqnarray*}
where $\delta_5= 0, 0.4, 0.6$.

\noindent Setting 6: Beta mixture distributions,
\begin{eqnarray*}
 \ \ \  \ X \mid Z=z  &\sim &  0.5 Beta \left(2(1+\delta_6\mathbbm{1}_{z=1}),6(1+\delta_6 \mathbbm{1}_{z=1})\right) \\
&  +& 0.5 Beta \left(6(1+\delta_6 \mathbbm{1}_{z=1}),
2(1+\delta_6 \mathbbm{1}_{z=1})\right)
\end{eqnarray*}
where $\delta_6= 0, 0.3, 0.45$. We calculated the size and power based $1000$ independent trials.

Setting 5 corresponds to a Beta distribution while Setting 6 corresponds to a mixture of Beta distributions.
With $\delta_5=0$ and $\delta_6=0$, we intended to examine the size of the test under the $H_0$. 
The power of the testing methods were examined with positive $\delta_5$'s and $\delta_6$'s.

\begin{figure}[h!]
  \centering
 \begin{tabular}{cc}
    \includegraphics[width=0.4\textwidth]{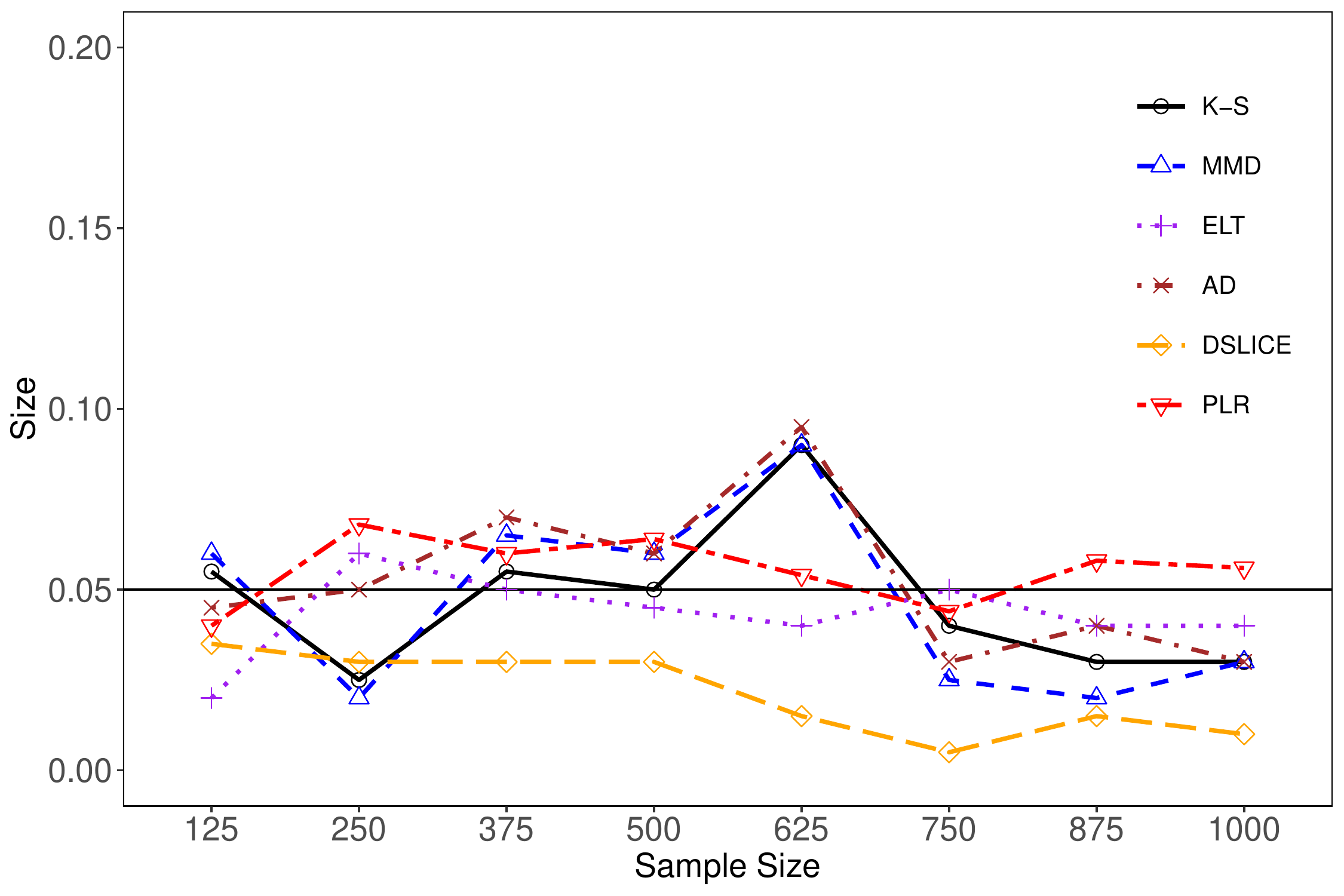}&
    \includegraphics[width=0.4\textwidth]{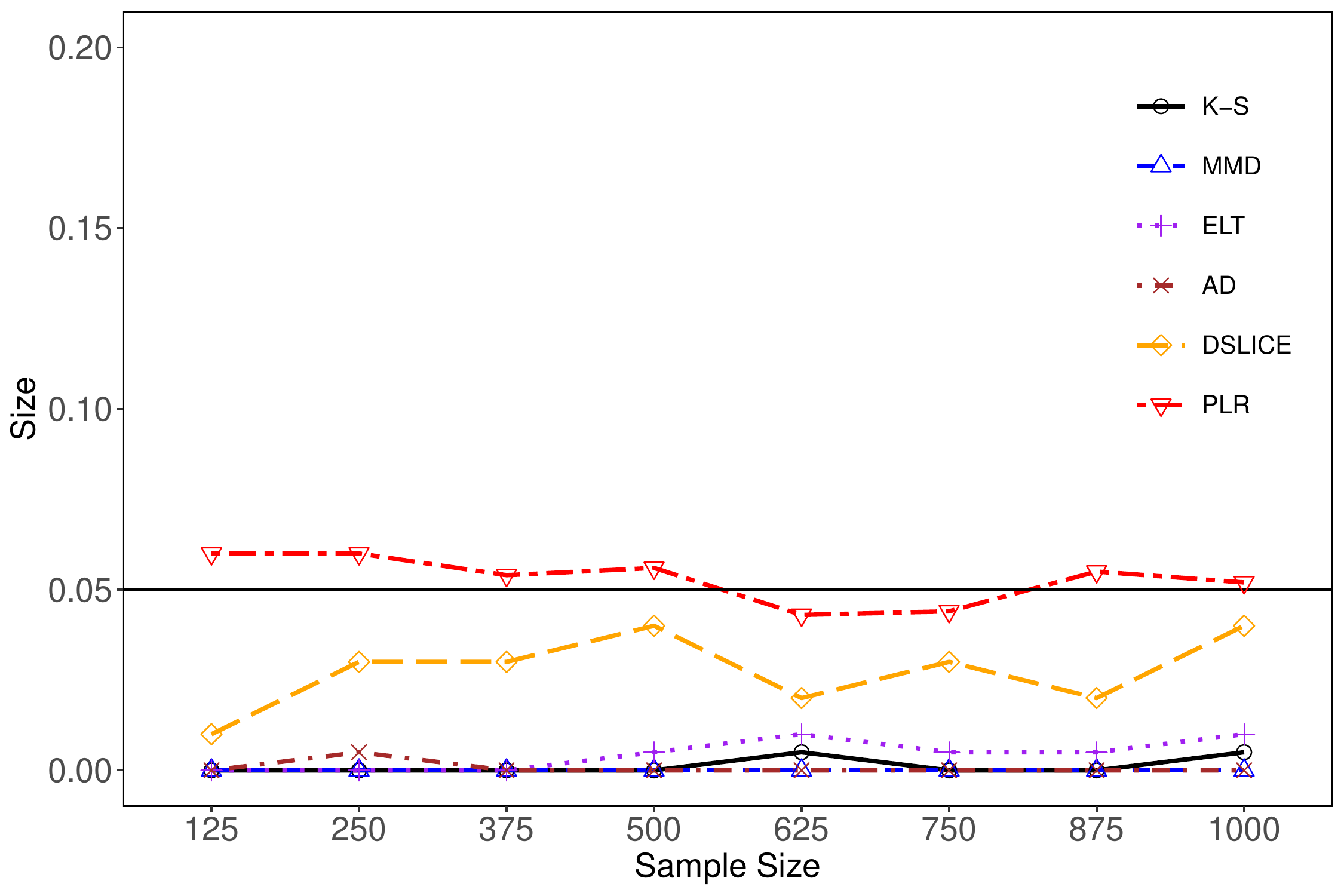}\\
    ~~~~~{\footnotesize $(a)$ Setting 5}  &~~~~~ {\footnotesize $(b)$ Setting 6}\\
\end{tabular}
  \caption{\it\footnotesize Size vs. sample size for KS, MMD, ELT, AD, DSLICE and PLR tests.}
  \label{fig:s4}
\end{figure}

\begin{figure}[h!]
   \centering
 \begin{tabular}{cc}
    \includegraphics[width=0.4\textwidth]{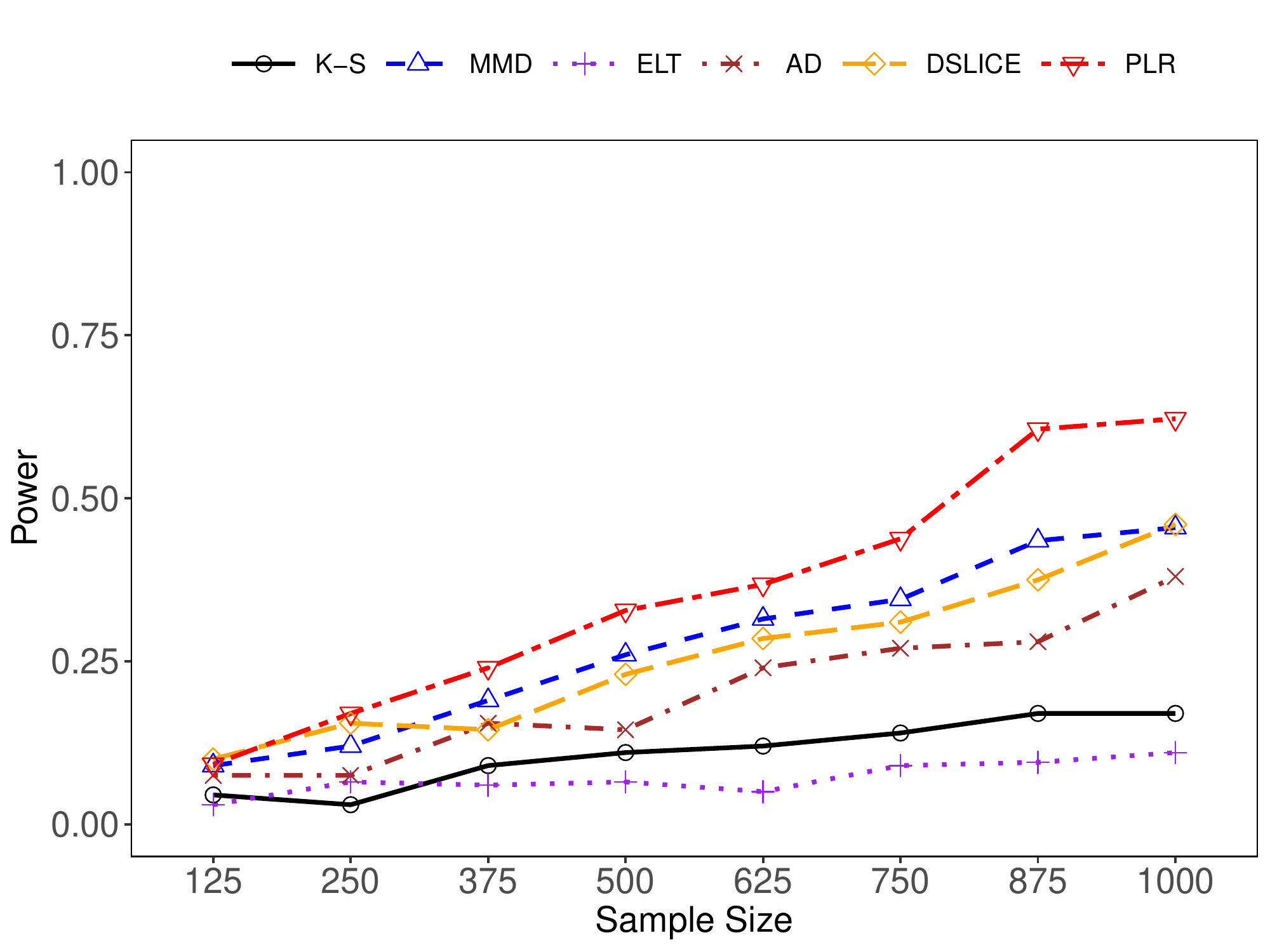}&\includegraphics[width=0.4\textwidth]{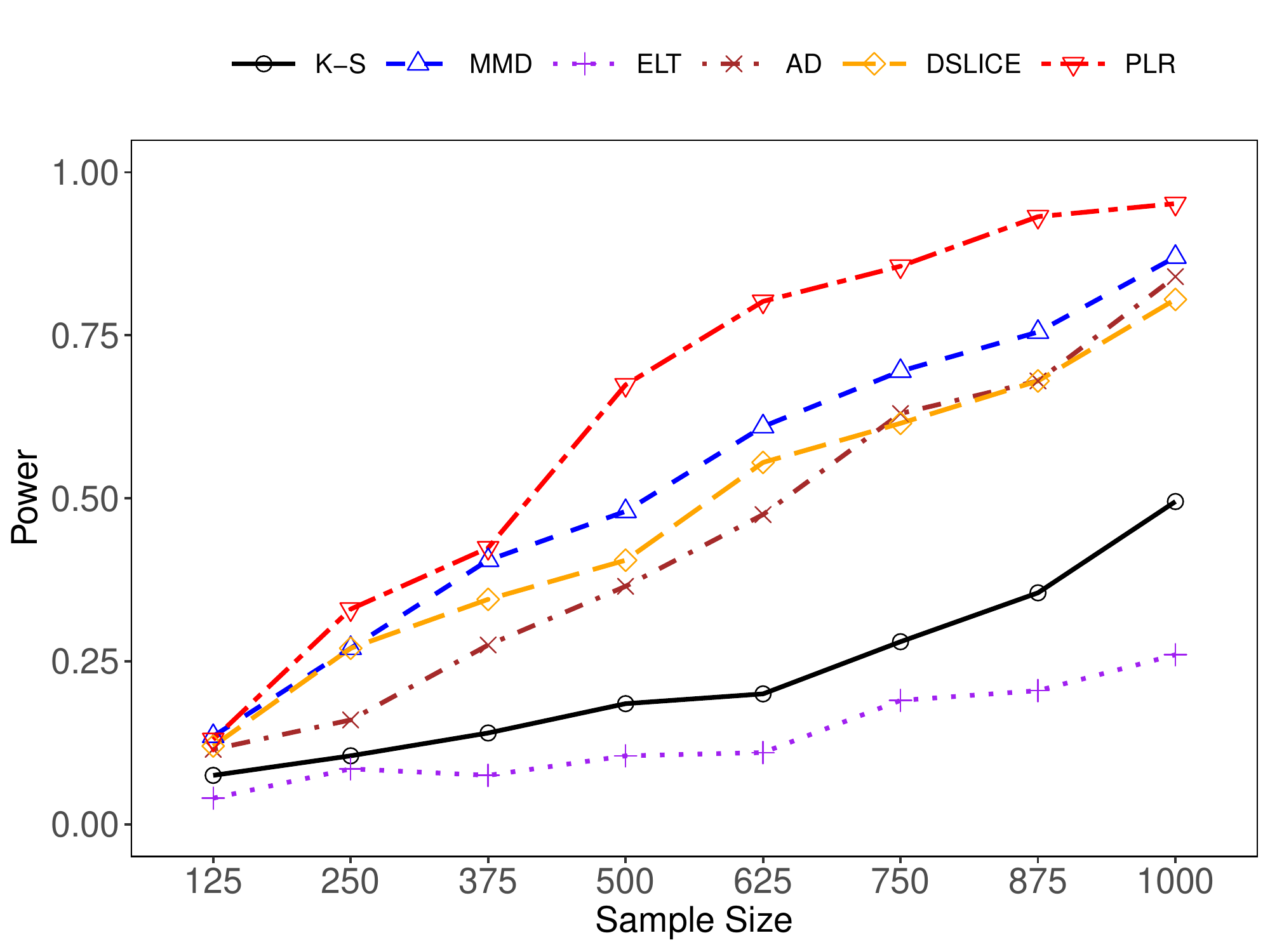} \\
    ~~~~~{\footnotesize (a) Setting 5: $\delta_5=0.4$}&~~~~~ {\footnotesize (b) Setting 5: $\delta_6=0.6$} \\
    \includegraphics[width=0.4\textwidth]{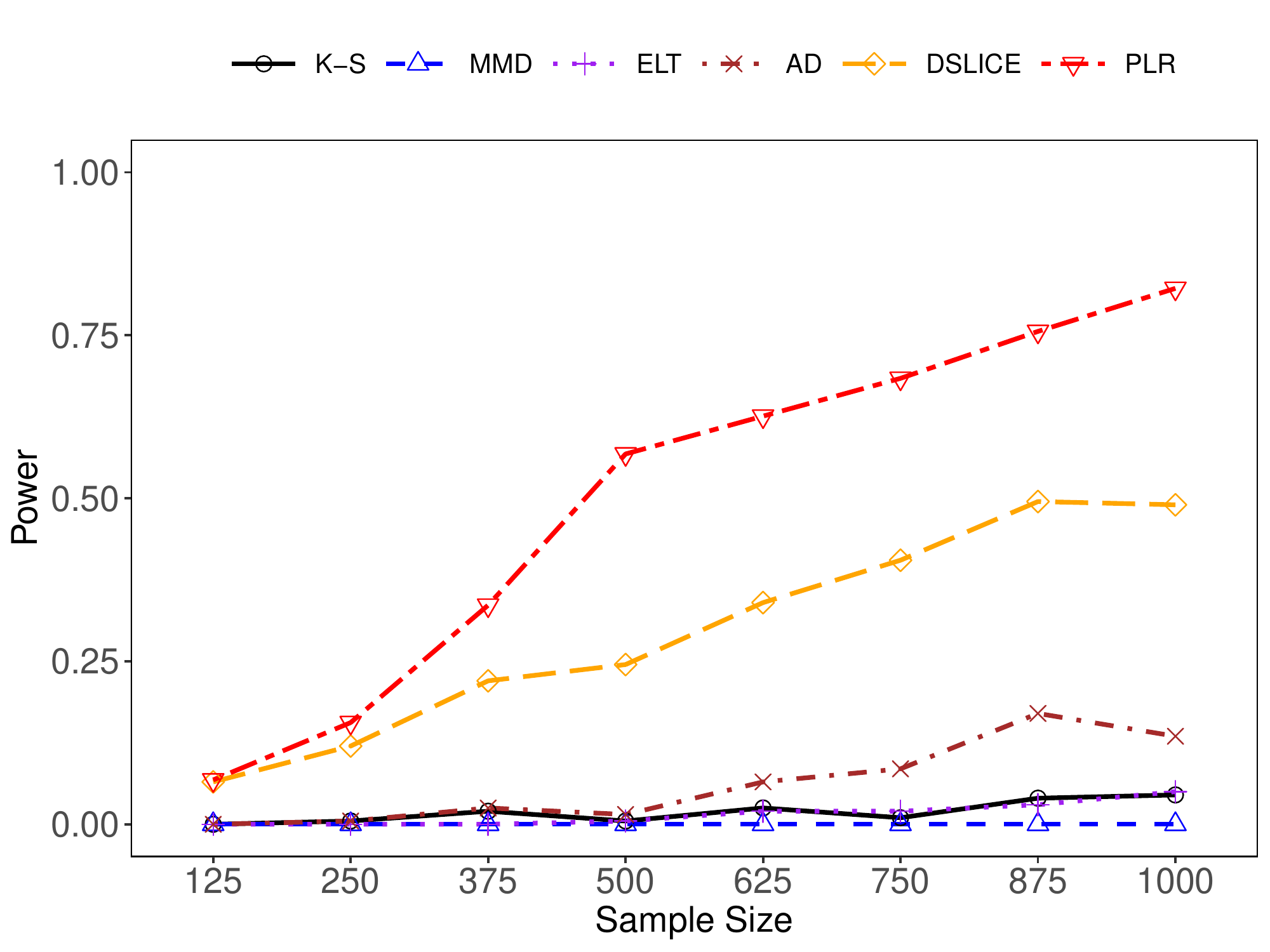}&\includegraphics[width=0.4\textwidth]{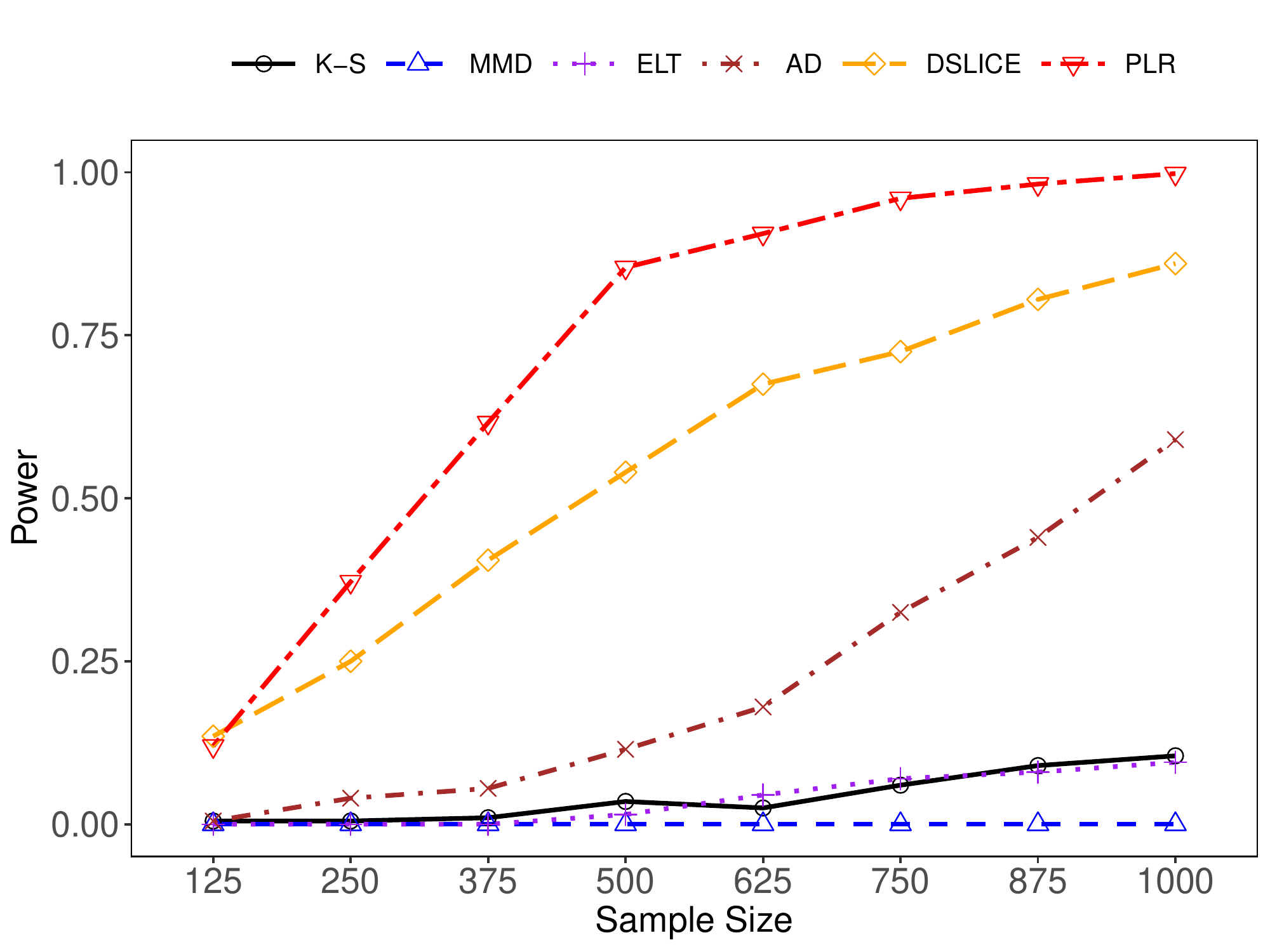} \\
    ~~~~~{\footnotesize (c) Setting 6: $\delta_6=0.3$}&~~~~~ {\footnotesize (d) Setting 6: $\delta_6=0.45$}
\end{tabular}
  \caption{\it\footnotesize Power vs. sample size for PLR, KS, MMD, ELT, AD, and DSLICE. }
  \label{fig:s5}
\end{figure}

As shown in Figure \ref{fig:s4}(a), the empirical sizes of Setting 5 were all around $0.05$ for the six test procedures when the density is a unimodal Beta distribution.  Whereas, for setting 6, Figure \ref{fig:s4}(b) shows that the empirical sizes of KS, MMD, ELT, AD, and DSLICE tests were significantly lower than $0.05$, while the sizes of PLR test were still around $0.05$. This demonstrates that our PLR test is  asymptotically correct for both unimodal and bimodal distributions.

Figure \ref{fig:s5}(a) and (b) examine the power of the three tests under Setting 5. In Setting 5, 
when $\delta_5=0.6$, the empirical powers of the MMD, AD and PLR test approached 1 as $n$ increased. In contrast, the power of KS and ELT test were lower than 0.5 even when the averaged sample size in each group reaches $1000$. DSLICE has power slightly over $0.5$ when $\delta_5=0.6$ when $n=1000$. In Setting 6, as shown in Figure \ref{fig:s5}(c)(d) the power of KS, MMD, and ELT test were below 0.2 even when the averaged sample size in each group is $1000$. The power of AD and DSLICE is slightly over $0.5$ when $n=1000$ and $\delta_6=0.45$. In contrast, the power of PLR test approached $1$ rapidly when $\delta_6$ 
was $0.30$ or $0.45$. We conclude that the PLR test is still the most powerful among the four tests
in all the considered settings, even when the data distribution is multimodal and non-Gaussian.

\end{document}